\documentclass[sigconf]{acmart}

\copyrightyear{2021}
\acmYear{2021}
\setcopyright{acmcopyright}
\acmConference[CCS '21] {Proceedings of the 2021 ACM SIGSAC Conference on Computer and Communications Security}{November 15--19, 2021}{Virtual Event, Republic of Korea.}
\acmBooktitle{Proceedings of the 2021 ACM SIGSAC Conference on Computer and Communications Security (CCS '21), November 15--19, 2021, Virtual Event, Republic of Korea}
\acmPrice{15.00}
\acmISBN{978-1-4503-8454-4/21/11}
\acmDOI{10.1145/3460120.3484750}

\usepackage{graphicx}
\usepackage{balance}
\usepackage{subfigure}
\usepackage[noend, ruled]{algorithm2e}
\usepackage{xspace}
\usepackage{url}
\usepackage{multirow}
\usepackage{amsmath}
\usepackage{enumerate}
\usepackage{enumitem}
\usepackage{pifont}
\usepackage{mathtools,comment,booktabs}

\usepackage{colortbl}

\definecolor{mygray}{gray}{.9}

\DeclareMathOperator*{\argmax}{arg\,max}
\DeclareMathOperator*{\argmin}{arg\,min}

\allowdisplaybreaks

\newcommand{\says}[2]{\noindent\textcolor{purple}{\textbf{#1 says: }}\textcolor{blue}{#2}\xspace}
\newcommand{\mypara}[1]{\vspace*{0.05in}\noindent\textbf{#1.} \xspace}

\definecolor{revision}{RGB}{0,0,0}
\newcommand{\revision}[2]{{\color{revision} #2\xspace}\xspace}

\newcommand{\EV}[1]{\ensuremath{\mathbb{E}\left[\, #1 \,\right]}}

\newcommand{\Var}[1]{\ensuremath{\mathsf{Var}\left[#1\right]}\xspace}
\newcommand{\Bias}[1]{\ensuremath{\mathsf{Bias}\left[#1\right]}\xspace}

\renewcommand{\Pr}[1]{\ensuremath{\mathsf{Pr} \left[#1\right] }\xspace}

\renewcommand{\AA}{\mathbf{A}}

\newcommand{\myexp}[1]{\ensuremath{e^{#1}}\xspace}

\newcommand{\Lapp}[1]{\mathsf{Lap}\left(#1\right)\xspace}
\newcommand{\pak}{PAK\xspace}
\newcommand{\method}{ToPS\xspace}
\newcommand{\methodldp}{ToPL\xspace}

\newcommand{\sqrtee}{e^{\epsilon/2}}

\usepackage{etoolbox}
\makeatletter
\patchcmd{\hyper@makecurrent}{%
    \ifx\Hy@param\Hy@chapterstring
        \let\Hy@param\Hy@chapapp
    \fi
}{%
    \iftoggle{inappendix}{%true-branch
        \@checkappendixparam{chapter}%
        \@checkappendixparam{section}%
        \@checkappendixparam{subsection}%
        \@checkappendixparam{subsubsection}%
        \@checkappendixparam{paragraph}%
        \@checkappendixparam{subparagraph}%
    }{}%
}{}{\errmessage{failed to patch}}

\newcommand*{\@checkappendixparam}[1]{%
    \def\@checkappendixparamtmp{#1}%
    \ifx\Hy@param\@checkappendixparamtmp
        \let\Hy@param\Hy@appendixstring
    \fi
}
\makeatletter

\newtoggle{inappendix}
\togglefalse{inappendix}

\apptocmd{\appendix}{\toggletrue{inappendix}}{}{\errmessage{failed to patch}}
\begin{document}

\pagestyle{plain}
\fancyhead{}

\title{Continuous Release of Data Streams under both Centralized\\ and Local Differential Privacy}

% \author{
% {Tianhao Wang$^{2,3}$, Joann Qiongna Chen$^4$, Zhikun Zhang$^5$, Dong Su$^6$, Yueqiang Cheng$^7$, \\Zhou Li$^4$, Ninghui Li$^1$, Somesh Jha$^8$}\\
% {Purdue University$^1$, CMU$^2$, University of Virginia$^3$, UC Irvine$^4$, CISPA$^5$, Alibaba$^6$, NIO$^7$, University of Wisconsin$^8$}
% }
\settopmatter{authorsperrow=4}

\author{Tianhao Wang}
% \authornote{Work done at Purdue.}
\authornote{Tianhao did most of the work while at Purdue University.}
% \email{tianhao@cmu.edu}
\affiliation{
{Carnegie Mellon University $\&$ University of Virginia}
\country{}
}

\author{Joann Qiongna Chen}
\affiliation{
\institution{University of California, Irvine}
\country{}
}

\author{Zhikun Zhang}
\affiliation{
\institution{CISPA}
\country{}
}

\author{Dong Su}
\affiliation{
\institution{Alibaba Inc.}
\country{}
}

\author{Yueqiang Cheng}
\affiliation{
\institution{NIO Security Research}
\country{}
}

\author{Zhou Li}
\affiliation{
\institution{University of California, Irvine}
\country{}
}

\author{Ninghui Li}
\affiliation{
\institution{Purdue University}
\country{}
}

\author{Somesh Jha}
\affiliation{
\institution{University of Wisconsin, Madison}
\country{}
}

\begin{abstract}
We study the problem of publishing a stream of real-valued data satisfying differential privacy (DP). One major challenge is that the maximal possible value in the stream can be quite large, leading to enormous DP noise and bad utility. To reduce the maximal value and noise, one way is to estimate a threshold so that values above it can be truncated. The intuition is that, in many scenarios, only a few values are large; thus truncation does not change the original data much. We develop such a method that finds a suitable threshold with DP. Given the threshold, we then propose an online hierarchical method and several post-processing techniques. 

Building on these ideas, we formalize the steps in a framework for the private publishing of streaming data. Our framework consists of three components: a threshold optimizer that privately estimates the threshold, a perturber that adds calibrated noise to the stream, and a smoother that improves the result using post-processing. Within our framework, we also design an algorithm satisfying the more stringent DP setting called local DP. Using four real-world datasets, we demonstrate that our mechanism outperforms the state-of-the-art by a factor of $6-10$ orders of magnitude in terms of utility (measured by the mean squared error of the typical scenario of answering a random range query). 
\end{abstract}




\maketitle

\section{Introduction}
% \says{ninghui}{"a bunch of" is too informal, either "many" or "a number of" or something else }
Continuous observation over data streams has been utilized in several real-world applications. 
% Continuous observation over infinite stream enables service provider to design better services and applications based on the statistical analysis of the stream.  
For example, security companies continuously analyze network traffic to detect abnormal Internet behaviors~\cite{csur:chandola2009anomaly}.
% Besides, government agencies publish more data on the Internet every day[\says{jc}{reference}]. Sharing this data enables greater transparency; delivers more efficient public services; and encourages greater public and commercial use of government information. [\says{jc}{adding reference}]
% network traffic monitoring makes it possible for network company to detect abnormal internet behaviors in the first place. 
% Range queries enable service provider like retail traders to study the behavior of customer. \says{zhikun}{why do we mention range query here?}\says{joann}{just simply because our measurement uses range queries. feel free to change it.}
% Utility companies monitors the electricity usage from smart meters to check whether people obey the stay-at-home order during pandemics like Covid-19 \says{zhikun}{reference?}.
% Even today, electricity usage of smart meters from a city may tell how good a region is practicing stay-at-home order during a pandemic like Covid-19. 
However, analyzing and releasing streams raise privacy concerns when these data contain sensitive individual information.
% While all these applications are beneficial, privacy concerns have become the major challenge when it comes to releasing this kind of useful aggregated raw data where individual users are involved. 
Directly publishing raw statistics may reveal individual users' private information.
% Publishing raw statistics like sum and average can result in revealing user’s unique event like their daily routines and cause huge privacy issues. 
For instance, electricity usage data from smart meters can reveal whether a user is at home or even what household appliances are used at some specific time~\cite{meter-memoir}.
% \says{zhikun}{shall we mention the electricity usage example in previous part, which has been commented?}
% For instance, smart meters may reveal the unusual presence at home of users or even their daily routines\cite{meter-memoir}. 
% These potential privacy risks would discourage data sharing. 
% Many services can benefit from real-time monitoring of statistics from customer data. Examples include electricity usage in a neighbourhood collected through smart meters, customers' expenditure in a supermarket, and commute time of residents from a city. Statistics for these applications can be obtained through a variety of sensors; and analysts and planners can use them to optimize services. 

% Privacy concerns, however, preclude release of raw statistics. For instance, a customer at a pharmacy would not be willing to disclose the purchase of medicines linked to a peculiar health condition. Likewise, analysis of smart meter data can likely reveal the activities of a particular household or even whether anyone is at home or not. Such privacy violations have been demonstrated for the case of smart-meter data where patterns such as the number of people in the household as well as sleeping and eating routines were revealed even without any prior training~\cite{meter-memoir}. The goal therefore is to enable monitoring of statistics without compromising individual privacy.

A promising technique for releasing private statistics is differential privacy (DP)~\cite{tcc:DworkMNS06}, which has become the gold standard in the privacy-research community. 
% Other privacy-enhancing techniques based on trusted execution environments (TEE)\says{jc}{adding reference} or cryptography \says{jc}{adding reference}enhance security therefore preserve privacy, but they are not privacy techniques in and of themselves. Therefore, they are outside of the scope of our discussion.
% A natural candidate for privacy protection is the rigorous framework of differential privacy (DP)~\cite{tcc:DworkMNS06,dpbook}. 
% Informally, DP guarantees that any individual data in a dataset has limited impact on the final statistics.
Informally, any algorithm satisfying DP has the property that its output distribution on a given database is close to the output distribution on a similar dataset where any single record is replaced.  The closeness is quantified by a parameter $\epsilon$, where a smaller $\epsilon$ offers a better privacy guarantee.
% Most of the work on differential privacy has focused on static (input) datasets.
% and there has been very little focus on datasets that are continuously being updated as in our setting~\cite{stoc:dwork2010differential, tissec:chan2011private}. 

To publish streams with DP, a widely accepted approach is to use the hierarchical structure ~\cite{tissec:chan2011private,stoc:dwork2010differential}.
% To provide statistics (e.g., the moving average) over data streams while satisfying DP, a widely-accepted technique is hierarchical algorithm~\cite{tissec:chan2011private,stoc:dwork2010differential}.  
The idea is to partition the time series into multiple granularities and then add noise to the stream to satisfy DP.  Because of the additive noise, it is impossible to accurately publish any single value in the stream.  Thus \revision{the goal here}{our goal} is to {\it accurately estimate the sum of values over any range of time}.
One challenge is that to satisfy DP, the magnitude of the noise should be proportional to the upper bound of the data, which is typically large.
Perrier et al.~\cite{ndss:perrier2019private} (termed \pak in this paper) observed that data in the stream is often concentrated below a value much smaller than the upper bound.  
To exploit this insight, a technique called contribution limitation is commonly-used~\cite{pvldb:ZNC12,pvldb:li2012privbasis,sigmod:day2016publishing,popets:wilson20differentially,pvldb:kotsogiannis2019privatesql}. It truncates the data using a specified threshold $\theta$ (i.e., values larger than $\theta$ are replaced by $\theta$).
The rationale is to reduce the noise (now the noise is proportional only to $\theta$) while preserving utility.
To find such a threshold while maintaining DP, \pak developed a method based on smooth sensitivity~\cite{stoc:nissim07}.  
The result can then be applied to the hierarchical algorithm to publish streams with improved utility.

We find three key limitations in existing work of \pak's.  First, it tries to privately find the $99.5$-th percentile to serve as the threshold $\theta$.  Unfortunately, using the $99.5$-th percentile (or any other fixed percentile) is unlikely to work across all settings of $\epsilon$ values and data distributions. 
Second, in order to get an analytical upper bound of the error caused by truncation (this error is also called bias), the authors further increase the estimated $99.5$-th percentile by first adding a positive term, and then multiplying a coefficient greater than 1.  As a result, the chosen $\theta$ is often unnecessarily large.  When $\epsilon$ is small (e.g., $\epsilon \le 0.1$), the value of $\theta$ is usually larger than the maximal possible value, running against the original purpose of choosing the threshold.  
Third, the method directly utilizes a basic hierarchical approach to output the stream, and does not fully take advantage of post-processing optimizations.  As a consequence, the accuracy of the output is far from ideal, and the results are worse when answering range queries with small selectivity.

In this paper, we propose a new approach by addressing the above-mentioned three limitations. 
Instead of using a fixed percentile, we design a data-dependent method to find the threshold $\theta$ that considers the overall data distribution.  Our goal is to minimize the overall error due to bias and DP noise simultaneously. 
Given $\theta$, we then propose a new hierarchical algorithm to obtain accurate results.
One major contribution is a novel online algorithm to enforce consistency over the noisy estimates \revision{}{(i.e., to make sure the number on any node equal the sum of its children's, which is violated if independently sampled noise is added to the hierarchy)} on the hierarchy to provide better utility.  While there exists consistency methods that work on the noisy hierarchies, our observation is that we can pre-compute all the noise, and then make the noise consistent first.  As the true values are naturally consistent, we can then add the consistent noise to the true values in an online manner and thus \revision{achieves}{achieve} an online consistency algorithm.
We prove the algorithm achieves minimum squared error and also satisfies DP.
Another contribution is that \revision{, we observe that the estimates in the lower levels of the hierarchy tend to be overwhelmed by the noise, leading to a low signal-noise ratio.  Thus, we further extend the algorithm to prune the lower-level nodes based on an optimization criteria.}{
we further extend the algorithm to prune the lower-level nodes based on an optimization criterion, based on the observation that the estimates in the lower levels of the hierarchy tend to be overwhelmed by the noise, leading to a low signal-noise ratio.} 
Our new hierarchical algorithm is also able to handle infinite streams.
% First, we design a data-dependent method to find the threshold $\theta$ that considers the data distribution instead of a fixed percentile.
% Second, we identify that the existing method~\cite{ndss:perrier2019private} only considers the bias due to truncation, and ignores the noise of DP.
% To address this conflict, we propose an algorithm built upon the differentially private exponential mechanism (EM~\cite{focs:mcsherry07}) by considering both errors simultaneously.
% The new algorithm is able to select an appropriate $\theta$ with minimal overall errors.
% Third, we propose a new hierarchical algorithm to obtain accurate results.
% Specifically, we design an online algorithm to enforce consistency over the noisy estimates on the hierarchy to provide better utility.  
% In addition, we observe that the estimates in the lower levels of the hierarchy tend to be overwhelmed by the noise, leading to a low signal-noise ratio.  Thus, we further extend the algorithm to prune the lower-level nodes based on an optimization criteria. 
% Our new hierarchical algorithm is also able to handle infinite streams.

Next, we generalize the above-mentioned algorithms into a new framework for streaming data publication.  It consists of three components: a {\it Threshold optimizer}, a {\it Perturber}, and a {\it Smoother}.  The threshold optimizer consumes a portion of the input stream, and finds a threshold $\theta$.  It then truncates \revision{any incoming value}{all incoming values} by $\theta$ and sends them to the perturber.  The perturber adds noise to each incoming element of the stream, and releases noisy counts to the smoother.  Finally, the smoother performs further post-processing on the noisy counts and outputs the final stream.  Together with the new algorithms described above, we call our solution {\it \method}.

Finally, based on the framework of \method, we design an algorithm to output streams while satisfying local DP (LDP), which \revision{offers stronger privacy protection}{protects privacy under a stronger adversary model} than DP.  We call the resulting method \methodldp.  Under LDP, only the users know the true values and thus \revision{removing}{removes} the dependence on the trusted central server.
In \methodldp, we use state-of-the-art LDP mechanisms for the Threshold optimizer and the Perturber.  
While the design of \methodldp relies on the findings in \method,  
% In particular, \methodldp implements the idea of minimizing both the bias and noise into the LDP primitive to find the optimal $\theta$.  
we also adapt existing LDP mechanisms to our setting to get better performance.

We implemented both \method and \methodldp, and evaluated them using four real-world datasets, including anonymized DNS queries, taxi trip records, click streams, and merchant transactions.
We use the Mean Squared Error (MSE) over random range queries as the metric of performance evaluation.  The experimental results demonstrate that our \method significantly outperforms the previous state-of-the-art algorithms.
More specifically, the most significant improvement comes from our new technique to finding $\theta$.  It contributes an improvement of $4-8$ orders of magnitude over \pak.
% \method can consistently output a reasonable threshold $\theta$, while \pak gives either very large or very small $\theta$ when $\epsilon$ is small.  
Even given the same reasonable $\theta$, \method can answer range queries $100\times$ more accurately than \pak.  Putting the two together, \method improves over \pak by $6-10$ orders of magnitude in terms of MSE.

\mypara{Contributions}
To summarize, the main contributions of this paper are threefold:
\begin{itemize}[leftmargin=*]

    \item We design \method for releasing real-time data streams under differential privacy.  Its contributions include an EM-based algorithm to find the threshold, an \revision{on-line}{online} consistency algorithm, the use of a smoother to reduce the noise, and the ability to handle infinite streams.
    % \item We extend the hierarchical method to handle infinite streams.
    % \item We study methods for outputting reasonable values without spending privacy budget.  
    % \item We also design a set of algorithms using both SVT and EM for finding any percentile, which is of independent interest.  More interestingly, we find that the SVT-based algorithm outperforms its EM-based version, contradicting to the previous statement in~\cite{pvldb:lyu2017understanding} that EM is always preferable than SVT.  
    %We amend that statement by introducing our guidance of designing the SVT algorithm.
    % \says{Yueqiang}{summarize and merge it into the first item.}We then note that using there is no single percentile that works the best as the threshold.  We thus design an algorithm to find the threshold without knowing which percentile works the best.  The algorithm is based on EM and takes both variance and bias into account.
    % \item For the perturber component, we propose an online variant of the constraint inference algorithm in addition to other ideas to improve the accuracy of the hierarchical method.
    \item 
    We extend \method to solve the problem in the more stringent setting of LDP and propose a new algorithm called \methodldp.  
    % To the best of our knowledge, \methodldp is the first technique for publishing streaming data under LDP.
    % Our method uses 
    \item We evaluate \method and \methodldp using several real-world datasets.  The experimental results indicate that both can output streams \revision{pretty}{} accurately in their settings\revision{, respectively}{}.  Moreover, \method outperforms the previous state-of-the-art algorithms by a factor of $6-10$ orders.  Our code is open sourced at \url{https://github.com/dp-cont/dp-cont}.
    % \methodldp  can output the stream with good accuracy  
\end{itemize}

\mypara{Roadmap}
In \autoref{sec:back}, we present the problem definition and the background of DP and LDP.
We present the existing solutions and our proposed method in \autoref{sec:method} and \ref{sec:method_ldp}.
Experimental results are presented in~\autoref{sec:exp}.
Finally, we discuss related work in Section~\autoref{sec:related} and provide concluding remarks in Section~\autoref{sec:conc}.
\section{Problem Definition and Preliminaries}
\label{sec:back}

We consider the setting of publishing a stream of real values under differential privacy (DP).  
The length of the stream could be unbounded.
Due to the constraint of DP, it is unrealistic to make sure every single reading of the stream \revision{}{is} accurate\revision{.  
So}{, so} the goal is to ensure the aggregated estimates \revision{}{are} accurate.

\subsection{Formal Problem Definition}
\label{subsec:formal_problem_def}
There is a sequence of readings $V=\langle v_1, v_2, \ldots\rangle$, each being a real number in the range of $[0, B]$.  We publish a private sequence $\tilde{V}$ of the same size as $V$ while satisfying DP, with the goal of accurately answering range queries.  Range query is an important tool for understanding the overall trend of the stream.  Specifically, a range query $V(i, j)$ is defined as the sum of the stream from index $i$ to $j$, i.e., $V(i, j) = \sum_{k=i}^j v_k$.
% \zzk{Shall we explain why accurately answering range queries is important in the setting of stream data, maybe given an example or citation or say range query is the building block for downstream tasks?}
We want a mechanism that achieves a low expected squared error of any randomly sampled range queries, i.e., 

% aggregator wants to learn $n$ readings, in a way that protects the privacy of each individual reading.  While individual estimation might be inaccurate, aggregated estimation could give usable information of the $n$ readings.

% We assume that the distribution is stable during the $n$ readings.  But the distribution can vary.

% \begin{table}[!ht]
% 	\centering
% 	\begin{tabular}{|c|l|}\hline
% 	\textbf{Notion} & \textbf{Description}\\ \hline \hline
% 	\end{tabular}
% 	\caption{Summary of Notations.}\label{tab:notation_summary}
% \end{table}

% \mypara{Timely Monitoring}
% The first $m$ observations are used to estimate $\theta$. 
% After that, at each time $t$, the server would like to learn about the distribution to take prompt actions.  

% Following the setting of~\cite{ndss:perrier2019private}, the first $m$ readings are reserved for pre-processing (finding the threshold in~\cite{ndss:perrier2019private}).
% For the rest of them, 
% We measure the similarity between the private stream and the estimated stream using range queries, similar to the metric in PAK~\cite{ndss:perrier2019private}.  
% \zzk{Shall we consider the expectation over all range queries here using some notations like $\underset{(i,j) \in \mathbb{Q}}{\mathbb{E}}$, where $\mathbb{Q}$ is the set of all possible range queries.}\tw{it's fine; $\mathbb{Q}$ introduces an extra notation}
\begin{align}
    \EV{\left(\tilde{V}(i, j) - V(i, j)\right)^2}.
    % \nonumber
    \label{eq:ana_metric}
\end{align}

\subsection{Differential Privacy}

% We work on continuous observations, i.e., $V = \langle v_1, v_2, \ldots, v_n \rangle$, where each $v_i$ is a positive number.  The privacy definition aims to protect against any single observation.

We follow the setting of PAK~\cite{ndss:perrier2019private} and adopt the notion of \emph{event-level} DP~\cite{stoc:dwork2010differential}, which protects the privacy of any value in the stream.

\begin{definition}
\label{def:dp}
(Event-level $(\epsilon, \delta)$-DP)
An algorithm $\AA(\cdot)$ satisfies $(\epsilon, \delta)$-differential privacy ($(\epsilon, \delta)$-DP),
if and only if for any two neighboring sequences $V$ and $V'$ and for any possible output set $O$,
\begin{equation*}
 \Pr{\AA(V) \in O} \leq e^{\epsilon}\, \Pr{\AA(V') \in O} + \delta,
\end{equation*}
where two sequences $V = \langle v_1, v_2, \ldots  \rangle$ and $V' = \langle v'_1, v'_2, \ldots \rangle$ are neighbors, denoted by $V\simeq V'$, when $v_i=v'_i$ for all $i$ except one index.
\end{definition}
% As the algorithm $\AA$ has access to the true sequence $V$, this model is called the centralized DP model (DP for short).
For brevity, we use $(\epsilon, \delta)$-DP to denote Definition~\ref{def:dp}.  
When $\delta=0$, which is the case we consider in this paper, we omit the $\delta$ part and write $\epsilon$-DP instead of $(\epsilon, 0)$-DP.
% \subsection{Semantics of Privacy}
% In this section, we discuss the semantics of privacy ensured by Definition~\ref{def:dp} and justify our choice of this privacy goal.

% According to our definition of adjacent streams and that of differential privacy. It can be guaranteed that individual's information can remain private.
% Event-level DP means that information about each individual's event remains private. 
% In many cases, event-level DP is a suitable guarantee of privacy, e.g., individuals might be happy to disclose their routine trip to work while unwilling to share the occasional detour. 
% The definition of event-level DP is also adopted by other works~\cite{ccs:chen2017pegasus,stoc:dwork2010differential,tissec:chan2011private}.  
% We note that while Definition~\ref{def:dp} is stated for the DP notion, it can be extended to the recently introduced Local DP (LDP)~\cite{focs:DuchiJW13} setting.  

% \subsection{Local Differential Privacy}
% \label{subsec:ldp}

\mypara{Justification of Event-Level DP}
Although event-level DP only protects one value, it is a suitable guarantee in many cases.  For example, individuals might be happy to disclose their routine trip to work while unwilling to share the occasional detour. 
Note that the data model is general and $V$ can also come from multiple users.  
For example, $V$ consists of the customers' expenditure from a grocery store, and we want to protect some unusual transaction.
Moreover, our model is a generalization of the basic model where every value is binary~\cite{focs:dwork2010boosting,tissec:chan2011private}, and can be used in building private algorithms with trusted hardware~\cite{soda:chan2019foundations}.

\mypara{\revision{}{Extension to Event-Level LDP}}
We also work in the local version of DP~\cite{siamcomp:kasiviswanathan2011can}.  Compared to the centralized setting, local DP offers a stronger \revision{level of protection}{trust model}, because each value is reported to the server in a perturbed form. \revision{The user's privacy}{Privacy} is protected even if the server is malicious.  For each value $v$ in the stream of $V$, we have the following guarantee:
% In particular, each user perturbs the input value $v$ using an algorithm $\AA_l$ and reports $\AA_l(v)$ to the aggregator.  
% Specifically, assume that the dataset is composed of data records (tuples), $D=\tuple{t_1, \ldots, t_n}$, we have $\AA(D) = \tuple{\AA(t_1), \ldots, \AA(t_n)}$.  
% ; and assume each tuple is the concatenation of sensitive and non-sensitive parts, i.e., $t_i=p_i||s_i$.  We have $\AA(D) = \tuple{\AA(t_1), \ldots, \AA(t_n)} = \tuple{p_1||\AA(s_1), \ldots, p_n||\AA(s_n)}$.  
% LDP concerns with the sensitive part.  
% The algorithm $\AA_l(\cdot)$ satisfies the following property:
\begin{definition}[$(\epsilon, \delta)$-LDP] \label{def:ldp}
	An algorithm $\AA(\cdot)$ satisfies $(\epsilon, \delta)$-local differential privacy ($(\epsilon, \delta)$-LDP),
	if and only if for any pair of input values $v, v'$, and any set O of possible outputs of $\AA$, we have
    \begin{equation*}
\Pr{\AA(v)\in O} \leq e^{\epsilon}\, \Pr{\AA(v')\in O} + \delta.
	\end{equation*}
\end{definition}
Typically, $\delta = 0$ in LDP\revision{}{~\cite{sp:WangLJ18,icde:WangXYZ19,sigmod:li2019estimating,jasa:Warner65} (one reason is that many LDP protocols are built on randomized response~\cite{jasa:Warner65}, which ensures $\delta=0$}).  
Thus we simplify the notation and call it $\epsilon$-LDP.
The notion of LDP differs from DP in that each user perturbs the data before sending it out and thus do not need to trust the server under LDP.  

\subsection{Mechanisms of Differential Privacy}
\label{subsec:dp_primitive}

% \says{Ninghui}{Should given the reasons that DP primitives are in main body, but LDP primitives are in the appendix.  Perhaps it is because understanding the DP primitives are necessary to understand the technical contributions, and the LDP primitives are used as blackboxes?  In any case, "defer to Appendix" does not make sense, because it is then not just deferring.}
We first review primitives proposed for satisfying DP.  We defer the descriptions of LDP primitives to~\autoref{app:ldp_primitive} as our LDP method mostly uses the LDP primitives as blackboxes.

\mypara{Laplace Mechanism}
The Laplace mechanism computes a function $f$ on the input $V$ in a differentially private way, by adding to $f(V)$ a random noise.  The magnitude of the noise depends on $\mathsf{GS}_f$, the \emph{global sensitivity} or the $L_1$ sensitivity of $f$, defined as,
\[
\mathsf{GS}_f = \max\limits_{V\simeq V'} ||f(V) - f(V')||_1.
\] 
When $f$ outputs a single element, such a mechanism $\AA$ is given below:
$$
\begin{array}{crl}
& \AA_f(V) & =f(V) + \Lapp{\frac{\mathsf{GS}_f}{\epsilon}}.
% \\
% \mbox{where}& \Pr{\Lapp{\beta}=x} & = \frac{1}{2\beta} \myexp{-|x|/\beta}
\end{array}
$$
In the definition above, $\Lapp{\beta}$ denotes a random variable sampled from the Laplace distribution with scale parameter $\beta$ such that $\Pr{\Lapp{\beta}=x} = \frac{1}{2\beta} \myexp{-|x|/\beta}$, and it has a variance of $2\beta^2$.  When $f$ outputs a vector, $\AA$ adds independent samples of $\Lapp{\frac{\mathsf{GS}_f}{\epsilon}}$ to each element of the vector.

% \mypara{Noisy Max}
% The Noisy Max (NM) technique was first proposed in~\cite{dpbook} by Dwork and Roth.  Given a set of queries $f_1\ldots f_m$, one adds Laplace noise $\Lap(2\mathsf{GS}_{f}/\epsilon)$ to each query answer, where 
% $$ \mathsf{GS}_f = \max_{i} \max_{(V,V'):d(V, V')=1} |f_i(V) - f_i(V')|, $$
% and then outputs the index of the largest answer.  It is proved this is equivalent to the more general Exponential Mechanism described as following.
% \input{primitive_exp.tex}

\mypara{Noisy Max Mechanism}
The Noisy Max mechanism (NM)~\cite{dpbook} takes a collection of queries, computes a noisy answer to each query, and returns the index of the query with the largest noisy answer.

More specifically, given a list of queries $q_1, q_2, \ldots$, where each $q_i$ takes the data $V$ as input\revision{,}{} and outputs a real-numbered result, the mechanism computes $q_i(V)$, samples a fresh Laplace noise $\Lapp{\frac{2\mathsf{GS}_q}{\epsilon}}$ and add\revision{}{s} it to the query result, i.e., 
\begin{align}\label{eq:nm_lap}
    \tilde{q}_i(V) = q_i(V) + \Lapp{\frac{2\mathsf{GS}_q}{\epsilon}},
\end{align} 
and returns the index 
% \begin{align}\label{eq:nm_max}
    $j = \argmax_i \tilde{q}_i(V)$.
% \end{align}  
Here $\mathsf{GS}_q$ is the global sensitivity of queries and is defined as:
\begin{align}
    \mathsf{GS}_q = \max_{i} \max_{V\simeq V'} |q_i(V) - q_i(V')|.
    % \label{eq:gs_exp}
    \nonumber
\end{align} 

Dwork and Roth prove this satisfies $\epsilon$-DP~\cite{dpbook} (recently Ding et al.~\cite{arXiv:ding2021permute} proved using exponential noise also satisfy DP).  Moreover, if the queries satisfy the monotonic condition, meaning that when the input dataset is changed from $V$ to $V'$, the query results change in the same direction, i.e., for any neighboring $V$ and $V'$
\revision{
\[\left(\exists_o\, q(V,o) \!<\! q(V',o)\right) \!\implies\! \left(\forall_{o'}\,q(V,o') \!\leq \! q(V',o')\right)\]
}{
\[\left(\exists_i\, q_i(V) \!<\! q_i(V')\right) \!\implies\! \left(\forall_{i'}\,q_{i'}(V) \!\leq \! q_{i'}(V')\right).\]}

Then one can remove the factor of $2$ in the Laplace noise.  This improves the accuracy of the result.

\subsection{Composition Properties}
\label{subsec:composition}
The following composition properties hold for both DP and LDP algorithms, each commonly used for building complex differentially private algorithms from simpler subroutines. 

\mypara{Sequential Composition}
% Differential privacy is composable in the sense that 
Combining multiple subroutines that satisfy DP for $\epsilon_1, \cdots,\epsilon_k$ results in a mechanism that satisfies $\epsilon$-DP for $\epsilon=\sum_{i} \epsilon_i$.
% Because of this, we refer to $\epsilon$ as the privacy budget of a privacy-preserving data analysis task.  When a task involves multiple steps, each step uses a portion of $\epsilon$ so that the sum of these portions is no more than $\epsilon$.

\mypara{Parallel Composition}
Given $k$ algorithms working on disjoint subsets of the dataset, each satisfying DP for $\epsilon_1, \cdots,\epsilon_k$, the result satisfies $\epsilon$-DP for $\epsilon=\max_{i} \epsilon_i$. 

\mypara{Post-processing}
Given an $\epsilon$-DP algorithm $\AA$, releasing $g(\AA(V))$ for any $g$ still satisfies $\epsilon$-DP.  That is, post-processing an output of a differentially private algorithm does not incur any additional loss of privacy.

% The composition properties allow us to execute multiple differentially private computations and reason about the cumulative privacy risk. In our applications, we want to bound the total risk so we impose a total ``privacy budget'' $\epsilon$ and allocate a portion of $\epsilon$ to each private computation.

% \input{sec_existing.tex}
\section{Differentially Private Streams}
\label{sec:method}
For privately releasing streams and supporting range queries over the private stream, the most straightforward way is to add independent noise generated through the Laplace distribution. However, this results in a cumulative error (following the tradition, we use absolute error here, which measures the difference from the true sum) of $O(\sqrt{n})$ after $n$ observations.

\mypara{The Hierarchy Approach}
To get rid of the dependency on $\sqrt{n}$, the hierarchical method was proposed~\cite{tissec:chan2011private,stoc:dwork2010differential}.  Given a stream of length $n$, the algorithm first constructs a tree: the leaves are labeled $\{1\}, \{2\}, \ldots, \{n\}$ and the label of each parent node is the union of labels from its child nodes.  Given $h=\log n$ layers, the method adds Laplace noise with $\epsilon/h$ in each layer.  
To obtain the noisy count $\tilde{V}(i,j)$, we find at most $\log n$ nodes in the hierarchy, whose labels are disjoint and their union equals $[i, j]$.  Given that the noise added to each node is $O(\log{n})$, this method has an error of $O(\log^{1.5}{n})$. 

In the online setting, where the stream data come one-by-one, we want to release every node in the hierarchy promptly.
% \zzk{I can understand the phrase ``as soon as possible'', not sure whether the reviewers understand this.
% How about something like ``In the online setting, where the stream data come one-by-one, we want to release every node in the hierarchy promptly without knowing the upcoming nodes.''}
To do so, at any time index $t$, we publish all nodes that contain $t$ as the largest number in their labels. 
% \zzk{It seems that we do not define online setting before, reading this sentence is a bit wired.}\tw{changed}

\subsection{Existing Work: \pak}
\label{subsec:existing}
% Dwork et al.~\cite{stoc:dwork2010differential} proposed an approach which injects Laplace noise to each data point in the stream and uses the hierarchical method for answering range queries.  
%In~\cite{stoc:dwork2010differential}, the hierarchical method, which we will describe shortly, is proposed to better answer range queries.  
To satisfy DP in the hierarchy method, one needs to add noise proportional to $B$, the maximal possible value in the stream, and $B$ can be quite large in many cases.  
Perrier et al.~\cite{ndss:perrier2019private} (we call it by the authors' initials, \pak, for short) observed that in practice, most of the values are concentrated below a threshold much smaller than $B$ (e.g., the largest possible purchase price of supermarket transactions is much larger than what an ordinary customer usually spends), and proposed a method to find such a threshold and truncate data points below it to reduce the scale of the injected Laplace noises. 
%In what follows, we describe the method proposed in~\cite{ndss:perrier2019private}, which we call \pak for short (the initials of the authors' names).  
In particular, the first $m$ values \revision{}{are used} to estimate the threshold $\theta$ with differential privacy.
After obtaining $\theta$, the following values in the stream are truncated to be no larger than $\theta$.  Reducing the upper bound from $B$ to $\theta$ reduces the DP noise (via reducing sensitivity).
The hierarchical method is used for estimating the stream statistics with the remaining $n-m$ values (\pak assumes there are $n$ observations).  

%They extend the hierarchical method accordingly.  

% \says{Ninghui}{I feel that we should explain why there is the need to estimate a threshold.  Essentially explain why the structure of PAK is the way it is, instead of just describing it.  Naive approach is to add Laplace noises to each value.  PAK introduces two improvements.  One deal with error accumulation, which is solved by using Hierarchies.  Probably want to cite something when mentioning Hierarchical method.  The other is to truncate the values and limit the noises.  Some texts under the Phases can be probably moved here.}

% Below we describe the two steps in more details.

\mypara{Finding the Threshold}
To obtain $\theta$, \pak proposed a specially designed algorithm based on Smooth Sensitivity (SS)~\cite{stoc:nissim07} to get the $p$-quantile (or $p$-percentile) as $\theta$, i.e., $p\%$ of the values are smaller than $\theta$.  
% Its high-level idea is as follows.  
%
% (we describe the algorithm of it in \autoref{app:dp_primitive}) 
% The reasoning behind this step is that the DP constraint is proportional to the data range, and truncation helps to reduce the error.
% Second, use the hierarchy (binary tree) as described in \autoref{subsec:bt} with $\theta$ to estimate range sums.  The second step is standard.  We now describe the first step in details.
%
% \says{zhikun}{there should be a sentence to describe the relationship between $\theta$ and the p-percentile.}
SS was used to compute the median with DP in the original work~\cite{stoc:nissim07}, and SS can also be easily extended to privately release the $p$-quantile. 
\pak proved that the result of SS is unbiased, but they further wanted to make sure the result is always larger than the real $p$-quantile.  This is because if the estimated percentile is smaller than the real one, the truncation in the next phase will introduce greater bias.  Thus \pak modified the original SS method to guarantee that the result is unlikely to be smaller than the real $p$-quantile.  As the details of the method are not directly used in the rest of the paper, we defer the details of both SS and the algorithm itself to \autoref{app:dp_primitive} and~\ref{app:imp_ss}.
% As the details of the method is not directly used in the rest of the paper, we defer its description to \autoref{app:imp_ss}.

There are two drawbacks of this method.  First, it requires a $p$ value to be available beforehand.  But a good choice of $p$ actually depends on the dataset, $\epsilon$, and $m$.  \pak simply uses $p=99.5$.  As shown in our experiment in \autoref{sec:exp_threshold}, $p=99.5$ does not perform well in every scenario.
Second, to ensure that $\theta$ is no smaller than the real $p$-quantile, \pak introduces a positive bias to $\theta$.

\subsection{Overview of Our Approach}
\label{sec:method_overview}

% \says{Ninghui}{The following sentence probably should be removed.  It is not clear what exactly does it mean.  What it tries to convey is probably in the sentences after it.}
% While finding a threshold can improve the result in the hierarchical method, 
% \pak focuses on limiting the error to an arbitrary small factor.  

The design of \pak was guided by asymptotic analysis.  Unfortunately, for the parameters that are likely to occur in practice, the methods and parameters chosen by asymptotic analysis can be far from optimal, as such analysis ignores important constant factors.  \revision{We advocate an approach that uses}{Instead, we use} concrete analysis to \revision{better}{} guide the choice of methods and parameters.

% \begin{figure}[t]
%     \centering
%     \includegraphics[width=0.46\textwidth]{figure/model.eps}
%     \caption{Illustration of \method where differentially private noise is added to the incoming entry of a stream based on threshold $\theta$ and final output is released after post-processing.}
%     \label{fig:model}
% \end{figure}
% \says{zhikun}{it seems that this paragraph express the same meaning with the above paragraph, maybe delete it?}
% The method used in \pak helps us understand the behavior of the method asymptotically.  But unfortunately, the analysis oftentimes is valid only in impractical settings (e.g., when $\epsilon$ is large).  We intend to advocate that in DP, we should care more about the reasonable cases because the parameters in DP cannot be arbitrarily large, e.g., the acceptable $\epsilon$ values are typically small.  To this end, we pursue designing algorithms and methods that perform reasonably well practically.  
% rather than theoretically guaranteed ones which work only 

In this section, we first deal with the threshold selection problem using the Noisy Max mechanism (NM, introduced in \autoref{subsec:dp_primitive}), which satisfies DP with $\delta = 0$.  
% \says{Ninghui}{There is no need to explicitly say "Our method does not have a provable bound on accuracy".  One can always prove some bound, it is just that oftentimes such bound are meaningless for realistic parameters.}
% Our method does not have a provable bound on accuracy, but the 
Empirical experiments show its superiority especially in small $\epsilon$ scenarios (which means compared to \pak, we can achieve the same performance with better privacy guarantees).  We then introduce multiple improvements for the hierarchical methods including an online consistency method, and a method to reduce the noise in the lower levels of the hierarchy.
We integrate all the components in a general framework\revision{}{,} {\it ToPS}, which consists of {\it a Threshold optimizer, a Perturber, and a Smoother}:
% \method provides an end-to-end solution towards publishing streaming data with differential privacy.
% Basically, \method consists of a novel combination of three components:  
% \says{zhikun}{this sentence repeat sentence ``we put all the components ...'' in the last paragraph.}
\begin{itemize}[leftmargin=*]
    \item {\it Threshold optimizer}: The threshold optimizer uses a small portion of the input stream to find a threshold $\theta$ for optimizing the errors due to noise and bias.  It then truncates any incoming values by $\theta$ and releases them to the perturber.  
    \item {\it Perturber}: The perturber adds noise to the truncated stream, and releases noisy counts to the smoother.
    \item {\it Smoother}: The smoother performs further post-processing on the noisy counts and outputs the final stream.
\end{itemize}

While design of \method is inspired by \pak, we include unique design choices to handle the problem.  In particular, \pak has two phases, the threshold finder and the hierarchical method.  We improve both phases.  Note that as we use NM~\cite{dpbook}, our algorithm satisfies $\epsilon$-DP while \pak satisfies $(\epsilon, \delta)$-DP.  Moreover, we introduce a smoother that further improves accuracy.
The fundamental reason that our method works better is that we design our method focusing on a good empirical performance rather than theoretical bounds.

\subsection{Threshold Optimizer}
\label{sec:threshold}

Different from \pak~\cite{ndss:perrier2019private} that focuses on bias when choosing the threshold $\theta$, our approach is to consider both bias and variance (due to DP noise).  As bias and variance go in opposite ways (i.e., when $\theta$ is large, bias will be small, but variance will go large, vice versa), there will be an optimal $\theta$ that minimizes the overall error.  Note that the overall error depends on the whole distribution of the data, which might be too much information to accurately estimate with DP.  To handle this issue, we choose to use the Noisy Max mechanism (NM)~\cite{dpbook}, which looks into the data privately and outputs only a succinct information of $\theta$.
In what follows, we first examine the error.  
% and then introduce a function that approximates the error function while having a low sensitivity so that we can effectively use  to solve it.  
 %We design the Threshold optimizer by using the Noisy Max mechanism while using queries that mimic the expected value of the squared error.
 %The origin of this method stems from the analysis of the error.

\mypara{\revision{Modeling the Error}{A Basic NM Query Definition}}
We first consider the expected squared error of estimating a single value $v$.  Assuming that $\tilde{v}$ is the estimation of $v$, it is well known that the expected squared error is the summation of variance and the squared bias of $\tilde{v}$: 
\begin{align}
\EV{(\tilde{v}-v)^2} = \Var{\tilde{v}} + \Bias{\tilde{v}}^2.\label{lemma:variance_bias}
\end{align}
% \begin{lemma}
% \label{lemma:variance_bias}
% $\EV{(\tilde{v}-v)^2} = \Var{\tilde{v}} + \Bias{\tilde{v}}^2$.
% \end{lemma}
% The proof is deferred to \autoref{app:proof}.
Note that $\Bias{\tilde{v}}$ equals $\EV{\tilde{v}} - v$.  
Given a threshold $\theta$ and privacy budget $\epsilon$, $\Bias{\tilde{v}}=\max(v-\theta,0)$ and $\Var{\tilde{v}}=\frac{2\theta^2}{\epsilon^2}$ (because we add Laplace noise with parameter $\theta/\epsilon$).
% \zzk{Bias is easy to understand here, shall we give a bit more details about Var in case some reviewers do not familiar with the Laplace mechanism?}\tw{added, also added in ealier part}

% \says{tianhao}{the parameter 3 can be tweaked to work for all datasets}
Since we are using the hierarchical method to publish streams for answering range queries, we use the error estimations of the hierarchical method to instantiate \autoref{lemma:variance_bias}.  Qardaji et al.~\cite{pvldb:QardajiYL13} show that there are approximately
% \zzk{on average?}\tw{I omitted smaller terms so it is approximate} 
$(b-1)\log_b(r)$ nodes to be estimated given any random query of range smaller than $r$, where $b$ is the fan-out factor of the hierarchical method; and the variance of each node is $\log^2_b(r)\frac{2\theta^2}{\epsilon^2}$.  
For bias, within range limitation $r$, a random query will cover around $\frac{r}{3}$ leaf nodes on average~\cite{pvldb:QardajiYL13} (We assume a random query can take any range in $[1,r]$. Thus there are $r(r+1)/2$ possible range queries. Among them, for any range of length $j\in[1,r]$, there are $r-j+1$ such ranges. The expected length of a random query is $\frac{\sum_{j=1}^{r}j(r-j+1)}{r(r+1)/2}=\frac{r+2}{3}\approx\frac{r}{3}$).
Denote $f_t$ as the frequency of value $t$, the combined error of the hierarchical method for answering random range queries would be: 
% \zzk{Where is the cubic over log in Var comes from?}\tw{$\epsilon$ is divided by log layers, and we sum over log nodes. I changed the earlier part, so now it should be clear}
\begin{align}
    (b-1)\log^3_b(r)\frac{2\theta^2}{\epsilon^2} + \left(\frac{r}{3} \sum_{\theta<t<B}f_t(t-\theta)\right)^2.\label{eq:range_error}
\end{align}

\mypara{\revision{Approximate the Error}{The Final NM Query Definition}}
For NM to be effective, the queries should have low sensitivity, meaning that changing one value perturbs the queries by a tiny amount.  However, if we directly use \autoref{eq:range_error} as the queries, the sensitivity is large: a change of value from $0$ to $B$ will result in the increase of \autoref{eq:range_error} by $(B-\theta)^2/9$.  Thus we choose to approximate the mean squared error by defining the queries in the following ways.  

Denote $m$ as the number of values to be used in NM, and $m_\theta$ as the number of values that are greater than $\theta$ from these $m$ values: 
\begin{align}
    m_\theta=|\{i\mid v_i \ge \theta, i\in[m] \}|,
    \label{eq:m}
\end{align}
where $[x]=\{1,2,\ldots,x\}$ and $|X|$ denotes the cardinality of set $X$.  
The first approximation method we use is to replace the variance and squared bias with their squared roots (standard deviation and bias).
Second, we use $c\cdot m_\theta/m$ ($c$ is a constant to be discussed later) to approximate $\sum_{\theta<t<B}f_t(t-\theta)$.
% Finally, we use the negation of both the standard deviation and bias errors as
Third, we multiply both the standard deviation and bias errors by $-\frac{3m}{c\cdot r}$ to ensure the sensitivity is $1$, and the query result of the target is the highest.
% (because of the normalization effect, this multiplication does not change 
% \says{zhikun}{does multiplying the error by a constant useful for improving the utility? it seems that this would not affect the probability of selecting any index $i$. otherwise, one can multiply the queries by an even smaller value to further improve the utility.}\says{tianhao}{we just multiply the bias}. 
% \autoref{eq:eme_q} gives us the queries:
Thus we have:
% we multiply the noise standard deviation by $m/r$ because $m_i$ is estimated using the $m$ values, while the bias is in a query of size at most $r$.  The resultant queries are given as follows:
\begin{align}
    % q_i(V) = -\frac{m}{r}\sqrt{2(b-1)\frac{2i^2\log_b^3(r)}{\epsilon^2}} - m_i\label{eq:eme_q}
    q_\theta(V) =& -\frac{3m}{c\cdot r}\sqrt{(b-1)\log^3_b(r)\frac{2\theta^2}{\epsilon^2}} - m_\theta\nonumber\\
    =& -\frac{3m \theta}{c r\epsilon}\sqrt{2(b-1)\log_b^3(r)} - m_\theta.\label{eq:eme_q}
\end{align}
% where $m$ is the number of data points that participate in the NM, and $r$ is the remaining data points for hierarchy.
% Each threshold value is associated with one query result, which is the negative noise error (in standard deviation) minus the truncation error.  The noise error is the same for each threshold; while the truncation error differs.  To ensure the queries has a small sensitivity, we take the number of values above the threshold as an approximation for the truncation error.
The first term is a constant depending on $\theta$ but independent of the private data, while the second term has a sensitivity of $1$.  

\mypara{Running NM}
To run NM, the set of possible $\theta$ values considered in the queries $q_\theta(V)$ should be a discrete set that covers the range $[0,B]$.  The granularity of the set is important.  If it is too coarse-grained (e.g., $\theta\in\{0, B/2, B\}$), the method is inaccurate, because the desired value might be far from any possible output.  On the other hand, if it is too fine-grained, the NM algorithm will run slowly, but it does not influence the accuracy.
In the experiment, we use all integers in the range of $[B]=\{1,2,\ldots, B\}$ as the possible set of $\theta$.

One unexplained parameter in \autoref{eq:eme_q} is $c$.  
There are two factor that contributes to $c$: (1) Using $m_\theta/m$ to approximate the bias term $\sum_{\theta<t<B}f_t(t-\theta)$ leads to underestimation. 
% $m_\theta$ by a constant $3$ approximates the bias well.  We thus use it in other datasets.
(2) As we will describe later in \autoref{sec:perturber} and~\ref{sec:smoother}, the actual squared error will be further reduced by our newly proposed method.  
% There is no closed-form for it, but empirically, the improvement is roughly $20\times$.
% Therefore we set $c=3\times 20 = 60$. 
While $c$ intends to be a rough estimation of the underestimation, it does not need to be chosen based on one particular dataset.
One can run experiments with a public dataset of similar nature under different parameters, the best level of error that can be achieved is usually a good indicator of $c$.  When a public dataset is unavailable, one can generate a synthetic dataset under some correlation assumption and run experiments. In experiments conducted for this paper, we choose $c=60$, and use it for all datasets and settings.

\begin{figure}[t]
    \centering
    % \includegraphics[width = 0.46\textwidth]{figure/figuresdns_1k_m4096_noise_verti.eps}
    % DNS dataset, $\epsilon=0.1, m=2^{12}, r=2^{16}$
    % \includegraphics[width = 0.46\textwidth]{figure/figurespos_m4096_noise_verti.eps}
    % POS, $\epsilon=0.1, m=2^{12}, r=2^{16}$
    \includegraphics[width = 0.46\textwidth]{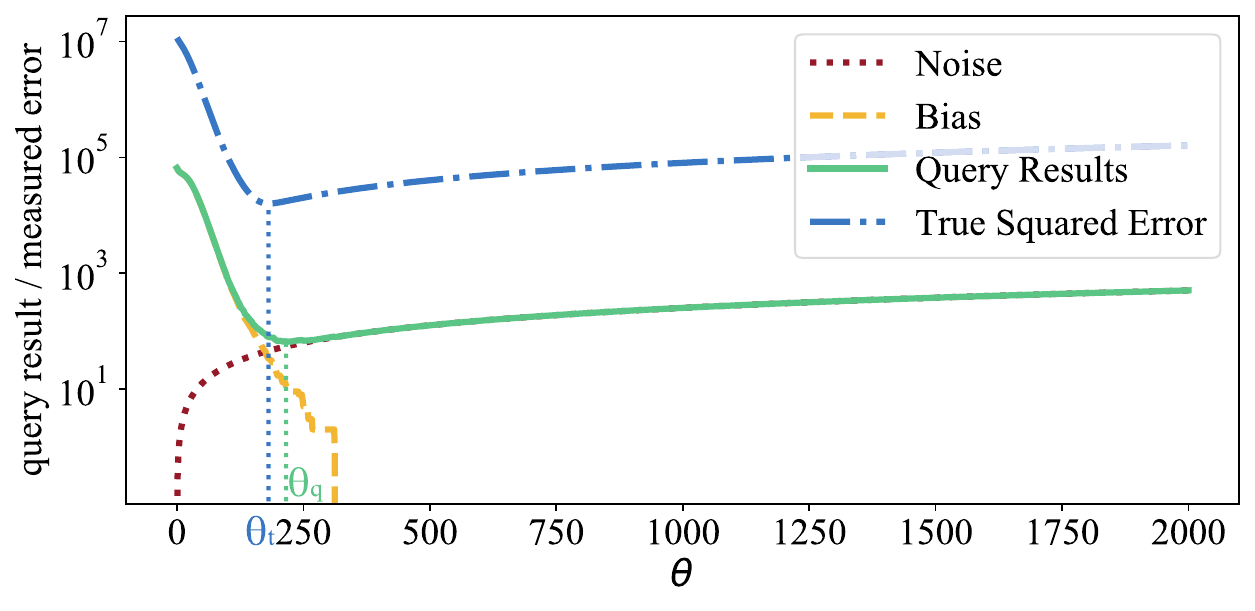}
    % DNS, $\epsilon=0.1, m=2^{16}, r=2^{20}$
    % \includegraphics[width = 0.46\textwidth]{figure/figurespos_m65536_noise_verti.eps}
    % POS, $\epsilon=0.1, m=2^{16}, r=2^{20}$
    \caption{Empirical comparison of approximated query results (\autoref{eq:eme_q}) and the true squared errors and their minimum points $\theta_{q}$ and $\theta_{t}$ on a real-world datasets (DNS).  We use $\epsilon=0.1, m=2^{16}, r=2^{20}$.  The $x$-axis is the possible value of $\theta$, and the $y$-axis is the query result or the measured error.
    % \zzk{Shall we capitalize the legends and give some explanations of the legends in the caption?}\tw{will do}
    % The NM mechanism considers both noise and bias and can return the target quite stably.
    }
    \label{fig:eme}
\end{figure}

\mypara{Verify the Approximation}
\autoref{fig:eme} illustrates the distribution of \autoref{eq:eme_q} and measured errors on a dataset that is used in the experiments in \autoref{sec:exp}.  The dataset is a network streaming dataset, called DNS.
% and the other one is a transaction dataset, called POS.  
We use $\epsilon=0.1, m=2^{16}, r=2^{20}$, which are the same as those parameters used in experiments.  From \autoref{fig:eme}, we can see that the distributions between the truly measured errors (\autoref{eq:range_error}) and the corresponding \autoref{eq:eme_q} on two datasets are very close.  The figure also illustrates the bias and variance factors of \autoref{eq:eme_q}.  The two factors grow in opposite directions which makes a global minimum where the target $\theta$ lies.  In addition, we also show that the threshold $\theta_q$ that minimizes our queries \autoref{eq:eme_q} is close to the target threshold $\theta_t$ which minimizes the real measured errors (\autoref{eq:range_error}).  Therefore, the above empirical evaluation results show the capability of the threshold optimizer in finding accurate $\theta$ values.

\subsection{Perturber}
\label{sec:perturber}

The perturber inherits the hierarchical idea~\cite{tissec:chan2011private,stoc:dwork2010differential} (also described in \autoref{subsec:existing}).  In this section, we start from the binary hierarchy used in \pak and put together three improvements to it to obtain a solution that is practical across a wide range of datasets.

\mypara{1. Better Fan-out}
According to Qardaji et al.~\cite{pvldb:QardajiYL13}, using a fan-out $b=16$ instead of $2$ in the hierarchy can give better utility.  
% This is done by optimizing the expected error of random range queries.  We thus use fan-out $b=16$ by default.
The result of optimal fan-out $b=16$ is derived by analyzing the accuracy (variance) of answering range queries.  In particular, we assume the range query is random and all layers in the hierarchy receive the same amount of privacy budget.  We then measure the expected accuracy (measured by variance) of answering the range query.  The optimal value $b=16$ is obtained by minimizing the variance.  It does not change on different datasets because the analysis is data-independent.  We thus use fan-out $b=16$ by default.

\mypara{2. Handling Infinite Streams}
% Second, we extend the hierarchical method to support infinite streams.  
\pak requires a fixed length $n$ a priori in order to build the hierarchy.  As a result, their algorithm stops after $n$ observations.
% We show that it is straight-forward to extend the method to support unlimited stream.  
% However, the stream can go endless, and one still wants to have timely monitoring over the data.  
% To handle this infinite stream model, we propose to keep applying hierarchies.  There is a question of how high the hierarchy should be.  We can have higher hierarchies for handling larger queries.  This is proposed by 
In order to support infinite streams, Chan et al.~\cite{tissec:chan2011private} proposed to have an infinitely high hierarchy, and each layer receives a privacy budget inversely proportional to the height of the layer.  For example, the bottom layer receives $0.9\epsilon$, then its parent layer receives $0.9^2\epsilon$, and so on.  While this ensures the overall privacy budget will never exceed $\epsilon$, the higher layers are essentially receiving a tiny amount of privacy budget.

In this paper, we note that most of the queries focus on limited ranges, and propose to have an upper bound on the query range denoted by $r$ and then split $\epsilon$ equally to the $h=\lceil \log_b r\rceil$ layers.  The value of $r$ stands for a limit below which most queries' ranges fall, and is determined externally (e.g., if we receive one value per minute, then it is unlikely that a query spans over a year).
% This is because while the stream can be infinite, the queries are not likely to span too long.
For each chunk of $r$ observations, we output a height-$h$ hierarchy. 
We note that each hierarchy handles a disjoint sub-stream of observations and thus this extension does not consume additional privacy budget because of the parallel composition property of differential privacy.  In the evaluation, we choose $r=2^{20}$.

\mypara{3. Online Consistency}
Given a noisy hierarchy, Hay et al.~\cite{pvldb:HayRMS10} proposed an efficient algorithm for enforcing consistency among the values in it.  By enforcing consistency, the accuracy of the hierarchy can be improved.  Note that this is a post-processing step and does not consume any privacy budget.
Unfortunately, the algorithm is {\it off-line} and requires the whole hierarchy data to be available.  Here we propose an {\it online} version of the enforce-consistency algorithm, so that we can output the noisy streams promptly.

Our method is built on the work of Hay et al.~\cite{pvldb:HayRMS10}.  
% So before going into details of our method, we first describe the consistency framework of theirs.
% Hay et al.~\cite{pvldb:HayRMS10}.  
Due to space limitation, we provide details of the algorithm in \autoref{app:consist}.  \revision{But the}{The} intuition is that, if we have two estimations of the same value, their (weighted) average would be closer to the true value.  
% Hay et al.~\cite{pvldb:HayRMS10} provides a 

\revision{Known}{Knowing} that the noisy estimates can be decomposed into true values and pure noise, our method generates all required noise in advance, followed by the consistency enforcement. In this way, the consistent off-line noise can be directly added to the incoming true values during online publishing. Because of the consistency of both the true values and the noise, the noisy estimates will also be consistent.
Moreover, we prove that the result of our online algorithm is equivalent to that of the off-line algorithm (the proof is deferred \autoref{app:consist}).
\begin{theorem}
\label{thm:online_consist}
The online consistency algorithm gives identical results as the off-line consistency algorithm.
\end{theorem}
This together with the fact that the off-line consistency algorithm can be seen as post-processing and thus satisfies DP, we can argue that our online algorithm also satisfies DP.

\subsection{Smoother}
  \label{sec:smoother}
% \subsection{Look-ahead Publish}
% We note that in the hierarchy $H$, each individual value contains essentially noise
In this section, we introduce a smoother to further improve the utility of the algorithm.  Assuming the hierarchy from perturber has $h$ layers, the smoother is designed to replace the values from the $s$ lowest layers with predictions (based on previous estimations).
The first question is how to choose $s$.
% By doing this, the perturber can focus on only the aggregated summations and be more accurate.

% When the magnitude of noise is comparable to the true data of the lower-level nodes of $H$, these nodes contain essential noise. 
% For example, the true value at each leaf node is at most the threshold $\theta$ because of the truncation.  The noise level (according to the Laplace mechanism given in~\ref{subsec:dp_primitive}) is $h\cdot \theta/\epsilon$ (where $h$ is the height of the hierarchy $H$).  
% We thus introduce a smoother to provide better utility.  
% More specifically, the original perturber works on $h$ layers.
% (although only the leaf nodes will be outputted, we draw Laplace noise from the other higher layers to make the noise consistent; and having a hierarchy helps when answering range queries).  
% The smoother replaces some $s$ lower levels of $H$ with predicted values, and the perturber runs on the upper $h-s$ layers.  Next, we show how to determine $s$.

\mypara{Optimizing $s$}
% Given the method to implement the smoother, we now turn to determining the value of $s$.
% Denote the finest granularity of the hierarchy by $b^s$.  For every $b^s$ values, we output their noisy sum.  For every individual value, we do not spend any privacy budget and use the Smoother to output it (we will shortly describe the implementation of the Smoother).
Selecting $s$ is important.  A larger $s$ results in smaller noise errors: because there are now $h-s$ layers in the hierarchy, each layer will receive more privacy budget according to sequential composition (given in \autoref{subsec:composition}).  
On the other hand, a larger $s$ probably leads to a larger bias (because we are only doing the actual estimate once every $b^s$ values; other estimates are from predictions based on previous values, thus are independent of the true values and less accurate).  Choosing a good value of $s$ thus is a balance between noise errors and bias.  Note that we already have noise error term from \autoref{eq:range_error}, but we need to calculate the bias introduced by the smoother (the truncation bias in \autoref{eq:range_error} already exists and does not change with $s$).

To estimate the smoothing bias, we assume that for each value, the bias amount is approximately $\theta/3$.  We then assume there are approximately $b^s/2$ values in a query\revision{,}{.}
% Then the average squared bias is $\frac{\theta^2}{16}\cdot \frac{\sum_{i=1}^{2b^s} i^2}{2b^s} = \frac{\theta^2(2b^s+1)(4b^s+1)}{96} $.
Then the average squared bias is approximated by $\frac{b^{2s}}{4}\frac{\theta^2}{9}$.  Therefore, we use the following equation to approximate the squared error:

\begin{align}
     &(b-1)\left(\log_b(r)-s\right)^3\frac{2\theta^2}{\epsilon^2} + \frac{b^{2s}}{4}\frac{\theta^2}{9}.
    %  \nonumber
     \label{eq:mse_guess}
\end{align}
Given $\epsilon$ and $r$, $s$ can be computed by minimizing the above error.

\mypara{Smoothing Method}
Given $s$, we now describe choices of implementing the smoother.  
We consider a set of methods proposed in the literature, and present their details in \autoref{app:smooth}.  Among them, the most straightforward one is ``Recent'' smoother, which predicts the next values based on the most recent estimation.  
In evaluation, we find it works the best, probably because the dataset we use is spiky.

% More precisely, proposed to avoid using DP to publish these values.  Instead, we publish them based on some ``guess'' of previous data.  
% \says{zhikun}{seems that we do not have experiments to compare these smooth methods? shall we give some guidelines to choose among them?}
% \says{tianhao}{we do. the diff is very small. }

\subsection{Summary and Discussions}
In summary, our method takes the raw stream $V=\langle v_1, v_2, \ldots\rangle$ as input and outputs a private stream $\tilde{V}=\langle \tilde{v}_1, \tilde{v}_2, \ldots\rangle$. 
\autoref{alg:ToPS2} gives the details of our method:  We first cache the first $m$ values and obtain $\theta$.  Then for each of the following \revision{value}{values}, we first truncate it and then use the hierarchical method together with the smoother to output the noisy value.

\begin{algorithm}[t]
\footnotesize
\LinesNumbered
\caption{\method}
\label{alg:ToPS2}

\KwIn{$V=\langle v_1, v_2, \ldots\rangle$, $\epsilon$, $m$, upper bound on query range $r$}
\KwOut{$\tilde{V}=\langle \tilde{v}_1, \tilde{v}_2, \ldots\rangle$}

$V_m \gets \langle v_1, \ldots v_m\rangle$ \tcp*{Cache the first $m$ values}

\For{$\theta = 1$ to $B$}{
 $m_\theta\gets |\{i\mid v_i \le \theta, i\in[m] \}|$ \tcp*{\autoref{eq:m}}
 
 $q_\theta(V_m) \gets -\frac{m \theta}{20 r\epsilon}\sqrt{2(b-1)\log_b^3(r)} - m_\theta$ \tcp*{\autoref{eq:eme_q}}
 
 $\tilde{q}_\theta(V) = q_\theta(V_m) + \Lapp{\frac{1}{\epsilon}}$ \tcp*{$\mathsf{GS}_g=1$, $q$ is monotonic}
 }

 $\theta \gets \argmax_\theta \tilde{q}_\theta(V) $ 
   \tcp*{Find $\theta$ via Noisy Max (\autoref{subsec:dp_primitive})}
 
  \tcc{The previous part of the code finds $\theta$. \\
  Now we are ready to release the stream.}

 $h\gets \log_{16} r$ $;$
 
 $s \gets \argmin_s \left[15\left(\log_{16}(r)-s\right)^3\frac{2\theta^2}{\epsilon^2} + \frac{16^{2s}}{4}\frac{\theta^2}{9}\right]$ \tcp*{\autoref{eq:mse_guess}}
 
 $u\gets \frac{16^s\theta}{2}$ $;$
 
 $build \gets $ True \tcp*{Indicator to build a tree}
 
 \ForEach{$i>m$}{
   $v_i \gets \min(v_i, \theta)$ \tcp*{Truncate}
   
   \uIf {\text{build}}{
     Init an $(h-s)$-layer hierarchy with fan-out $16$ $;$
     
     Assign $0$ to all nodes \tcp*{Build the virtual tree}
     
     Add independent noise $\Lapp{\frac{h-s}{\epsilon}}$ to each node $;$
     
     Make the tree consistent \tcp*{\autoref{app:consist}}
     
     $cur\_node \gets$ left-most noisy node on tree $;$
     
     $build \gets $ False $;$
   }
   
   \uIf {$(i - m) \mbox{ mod } r = 0$}{
     $build \gets $ True \tcp*{Time to build another tree}
   }
   
   $cur\_node\gets cur\_node+ v_i$ $;$
   
   \uIf {$(i - m) \mbox{ mod } 16^s = 0$}{
     
     {\bf Output} $cur\_node - u \times (16^s - 1)$ $;$
     
     $u \gets cur\_node $ $;$
     
     $cur\_node\gets$ next noisy node on tree $;$
   }
   \uElse{
     
     {\bf Output} $u/16^s$ \tcp*{``Recent'' smoother in \autoref{app:smooth}}
     
   }
} 
\end{algorithm}

In this paper, we focus on the setting used in \pak, where the threshold optimizer does not publish the first $m$ values, but uses them to obtain $\theta$.  After the first $m$ values, it sends $\theta$ to the perturber and truncates any incoming value by $\theta$.   
The perturber then outputs values using the hierarchical method, and there is a smoother that further processes the result.  
% The details of the algorithm are given in \autoref{app:alg}.
Our method is also flexible and can work in other settings.  
% For example, the threshold optimizer can be called several times to accommodate the distribution of the values changing over time. It can also let the stream go through while at the same time finding the threshold, using a portion of the privacy budget. 
We will discuss more about the flexibility of \method in \autoref{sec:conc}.

We claim that \method satisfies $\epsilon$-DP.  The perturber uses $\epsilon/h$ to add Laplace noise to each layer of the hierarchical structure.  By sequential composition, the overall data structure satisfies $\epsilon$-DP.  To find the threshold, \method uses a disjoint set of $m$ observations and runs an $\epsilon$-DP algorithm.  Due to the parallel composition property of DP, the threshold optimizer and the perturber together satisfy $\epsilon$-DP.  The online consistency algorithm and the smoother's operations are post-processing procedures and do not affect the privacy guarantee.  
% For \methodldp, similarly, it satisfies $\epsilon$-LDP as we only use existing LDP primitives.

% While we use the Laplace mechanism for estimating the hierarchy; it is possible that using the Gaussian mechanism which introduces $\delta>0$ gives additional benefit.  In \autoref{subsec:delta}, we show that this is not the case in our setting.  The high-level idea is that the Gaussian mechanism and modern composition technique is helpful only if we split the privacy budget for more than $20$ times.

% \subsection{Putting Things Together}

% \begin{algorithm}[th]
% \footnotesize
% % \SetCommentSty{small}
% \LinesNumbered
% \caption{\method}
% \label{alg:dpsyn}

% \KwIn{Raw stream $V$, privacy budget $\epsilon$, time lag $m$, fan-out $b$;}
% \KwOut{Private stream $\tilde{V}$;}

% Estimate threshold $\theta$ in the range of $[0,B]$
% based on the first $m$ values using EM with a quality function of \autoref{eq:eme_q};

% Construct an $h$-layer hierarchy with fan-out $b$ from the stream truncated by $\theta$ (from each $r$ values, if infinite stream);

% For each leaf/node $x$, add noise from Laplace mechanism using the online consistency algorithm \autoref{app:consist} with an $\epsilon/h$ to each layer, which yields $\epsilon$-DP;

% Optimize smoother parameter $s$ by minimizing error in \autoref{eq:mse_guess};

% Replace output of $s$ lower levels of the hierarchy using smoother ``Recent'';

% \end{algorithm}

% \section{\methodldp: LDP version of \method}
\section{Publishing Streams in LDP Setting}
\label{sec:method_ldp}
In this section, we introduce \methodldp for publishing streaming data under local DP (LDP).  
To the best of our knowledge, this is the first algorithm that deals with this problem under LDP.

In LDP, users perturb their values locally before sending them to the server, and thus do not need to trust the server.
Applying to the streaming values in our setting, each value should be perturbed before being sent to the server.  What the server does is only post-processing of the perturbed reports.  
% After that, the results can be used or shared with other parties.  

\methodldp follows the design framework of \method.  There is a threshold optimizer to find the threshold based on the optimal estimated error, and the threshold is used to truncate the users' values in the later stage.  Different from the centralized DP setting, in the local setting, the obtained threshold will be shared with the users so that they can truncate their values locally.  
% Note that this step does not violate the privacy guarantee as the threshold is derived in a way that already satisfies LDP.  
The perturber section is also run within each user's local side, because of the privacy requirement that no other parties other than the users themselves can see the true data.  There is no smoother section.
In what follows, we describe the construction for the threshold optimizer and the perturber. 
% \says{zhikun}{We don't have smoother for \methodldp?}\says{tianhao}{correct}

% \begin{figure}[t]
%     \centering
%     \includegraphics[width=0.46\textwidth]{figure/ToPL.eps}
%     \caption{Illustration of ToPL where each incoming value is truncated by $\theta$ and is protected by local differential privacy.}
%     \label{fig:ToPL}
% \end{figure}

\subsection{Design of the Threshold Optimizer}
\label{sec:threshold_ldp}
% We first describe our method of estimating a suitable threshold under LDP.
In LDP, each user only has a local view (i.e., they only know their own data; no one has a global view of the true distribution of all data), thus there is no Noisy Max mechanism (NM) (described in \autoref{subsec:dp_primitive}) that we can use as in the DP setting.  
Instead, most existing LDP algorithms rely on frequency estimation, i.e., estimation of how many users possess each value, as what the Laplace mechanism does in DP.
We also rely on the frequency estimation to find the optimal threshold.  Although the distribution estimation is more informative, it is actually less accurate than the Noisy Max mechanism because (to publish more information) more noise needs to be added.
% We first review methods that estimate the entire distribution.

\mypara{Frequency Estimation in LDP}
% The setting is as follows.  There are $n$ users, each having a value $x$ to report.  The user does not want others to learn $x$, so he perturbs $x$, using a randomizer $R$, to get $y$, satisfying differential privacy, and then reports $y$ to an aggregator. The aggregator, possessing reports from all users, processes the information using a count estimator $C$.  We say that $R$ satisfies $\epsilon$-local differential privacy (LDP) if and only if for all possible inputs $ x_1,x_2$ and all possible output $y$, $ \;\; \frac{\Pr{R(x_1)=y}}{\Pr{R(x_2)=y}} \leq e^\epsilon.$
% 
% 
% The basic mechanism in LDP is called randomized response~\cite{jasa:Warner65}.  It works for categorical or discrete domain.  Assume the values are in the domain of $\{0,1,\ldots,B\}$, the idea of randomized response is to let the users randomize their values so that the true value is reported with probability $p=\frac{e^\epsilon}{e^\epsilon+B}$ and any other value with probability $q=\frac{1}{e^\epsilon+B}$.  While each user's value looks random, when aggregated together, statistical properties (in this case the frequencies) appear.  The problem has been exhaustively investigated (e.g., ~\cite{ccs:ErlingssonPK14,stoc:BassilyS15,aistats:acharya2019hadamard,uss:WangBLJ17}).  
% Among them, \olh (optimal local hashing) from~\cite{uss:WangBLJ17} achieves the best accuracy with a small communication cost.  It maps the value to a smaller domain using a hash function randomly picked from a universal family, and apply the randomized response idea to the hashed result.  
% 
Li et al.~\cite{sigmod:li2019estimating} propose the Square Wave mechanism (SW for short) for ordinal and numerical domains.  It extends the idea of Randomized Response~\cite{jasa:Warner65} in that values near the true value will be reported with high probability, and those far from it have a low probability of being reported.  
% The method estimates the whole distribution together, thus gives smaller errors when we consider multiple values together. 
% \says{zhikun}{I don't quiet understand this sentence. How does \SW estimate together and consider multiple values?}
The server, after receiving the reports from users, runs a specially designed Expectation Maximization algorithm to find an estimated density distribution that maximizes the expectation of observing the output.
For completeness, we describe details about SW in \autoref{app:ldp_fo}.  

\mypara{Optimized Threshold with Estimated Distribution}
% To find the threshold, the baseline method is to use any frequency oracle \says{zhikun}{It seems that we do not mention the notion of frequency oracle in the main content. Maybe refer the readers to the Appendix here?} to find the distribution, and obtain an optimal threshold from the estimated distribution.  
% \says{zhikun}{Shall we provide some shortcomings of the baseline method?}
% Similar to DP, in LDP, a specific percentile is also used to find the a threshold when dealing with transactional data such as finding frequent itemset~\cite{sp:WangLJ18}.
To find the threshold, the baseline method is to find a specific percentile as the threshold $\theta$.  This method is used for finding frequent itemset~\cite{sp:WangLJ18}.
Based on the lessons learned from the threshold optimizer in the DP setting, we use the optimization equation given in \autoref{lemma:variance_bias} to find $\theta$.

Specifically, denote $\tilde{f}$ as the estimated distribution where $\tilde{f}_t$ is the estimated frequency of value $t$.  
Here the set of all possible $t$ to be considered can no longer be $[B]=\{1,2,\ldots,B\}$.  Instead, we sample $1024$ values uniformly from $[B]$.  This is because SW uses the Expectation Maximization algorithm, and a large domain size makes it time- and space-consuming.  Similar to \autoref{eq:range_error} considered in the DP setting, we use an error formula:
\begin{align}
    \frac{r}{3}\cdot \Var{\tilde{v}} + \frac{r^2}{6}\left( \sum_{\theta<t<B}\tilde{f}_t(t-\theta)\right)^2.\label{eq:range_error_ldp}
\end{align}
Here $\Var{\tilde{v}}$ denotes the variance of estimating $v$, which we will describe later.  It is multiplied by $\frac{r}{3}$ because in expectation, a random range query will involve $\frac{r}{3}$ values, and each of them is estimated independently.  For the second part of \autoref{eq:range_error_ldp}, it can be calculated directly with SW.  The multiplicative coefficient $\frac{r^2}{24}$ is the averaged case over all possible range queries.  That is, denote $j$ as the range of a query, there are $r-j+1$ range-$j$ queries within a limit $r$.  In total, there are $\sum_{j=1}^r (r-j+1)$ possible queries.  For each of them, we have a $j^2$ coefficient in the squared bias.  Thus, we have $\frac{\sum_{j=1}^r (r-j+1) j^2}{r(r+1)/2} = \frac{(r+1)(2r+1)}{3} - \frac{r(r+1)}{2}\approx \frac{r^2}{6}$ as the average-case coefficient.

% The analytical variance is derived from the estimated mean of any range query, which we will describe later; and the bias is obtained from the result of SW.  

\mypara{Using SW as a White Box}
\revision{One thing to note is that, SW is proposed for estimating smoothed distributions, while in our case, the distribution is very skewed, because the majority of values are expected to be concentrated below a threshold.}{}
\revision{To make SW output a reasonable threshold,}{To find a reasonable threshold using SW,} \revision{instead of using SW as a black-box, }{}we make the following modifications.  First, we eliminate the smoothing step from SW, because we observe that in some cases, \revision{the}{} smoothing \revision{operation}{} will ``push'' the \revision{density}{estimated probability density} to the two ends \revision{}{of the range}\revision{, which deviates even further from the underlying dataset}{}.  \revision{}{If some density is moved to the high end, the chosen threshold $\theta$ can be unnecessarily large.} 

Second, we add a post-processing step to prune the small densities outputted by SW. 
% \jc{what is $w$ here? it's the value of the estimate right?}
In particular, we find the first \revision{$w$ so that the consecutive $5$ estimates are all below $0.01\%$.}{qualified value $w$, whose next $5$ consecutive estimates are all below $0.01\%$.}   This is a signal that the density after $w$ will converge to $0$.  We thus replace the estimated density after $w$ with $0$.  In the experiment, we observe that the two steps help to find a more accurate $\theta$.
% \SW by default has a step to avoid overfitting called smoothing.  This step makes the estimation more smoothed, but this is not the case in our setting.  

% However, as each estimate is done independently, this method has a large error due to the aggregated noise.

% To mitigate this concern, we propose to use the square wave (\SW) method.  

\subsection{Design of the (Local) Perturber}
\label{sec:perturber_ldp}
After obtaining the threshold $\theta$, the server sends $\theta$ to all users.  When a user reports a value, it will first be truncated.  The user then reports the truncated value using the Hybrid mechanism.  The method is described in \autoref{app:ldp_mean}.  It can estimate $v$ with worst-case variance given in \autoref{eq:var_hm}, which can be plugged into \autoref{eq:range_error_ldp} to find $\theta$.
Note that the reports are unbiased by themselves.  So to answer a range query, we just need to sum up values from the corresponding range, and there is no need for a smoother.

\begin{figure*}[h]
	\centering
	
	\subfigure[\method, $\epsilon=0.01$]{
		\includegraphics[width=0.48\textwidth]{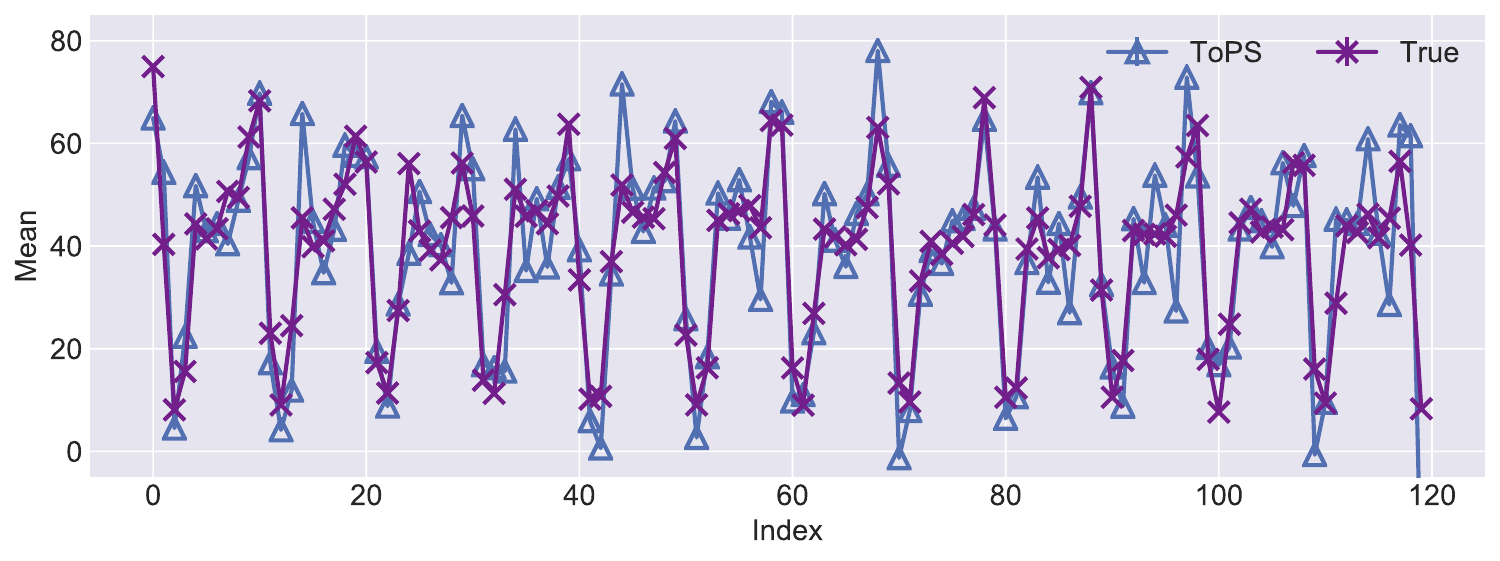}
	}
	\subfigure[\method, $\epsilon=0.05$]{
		\includegraphics[width=0.48\textwidth]{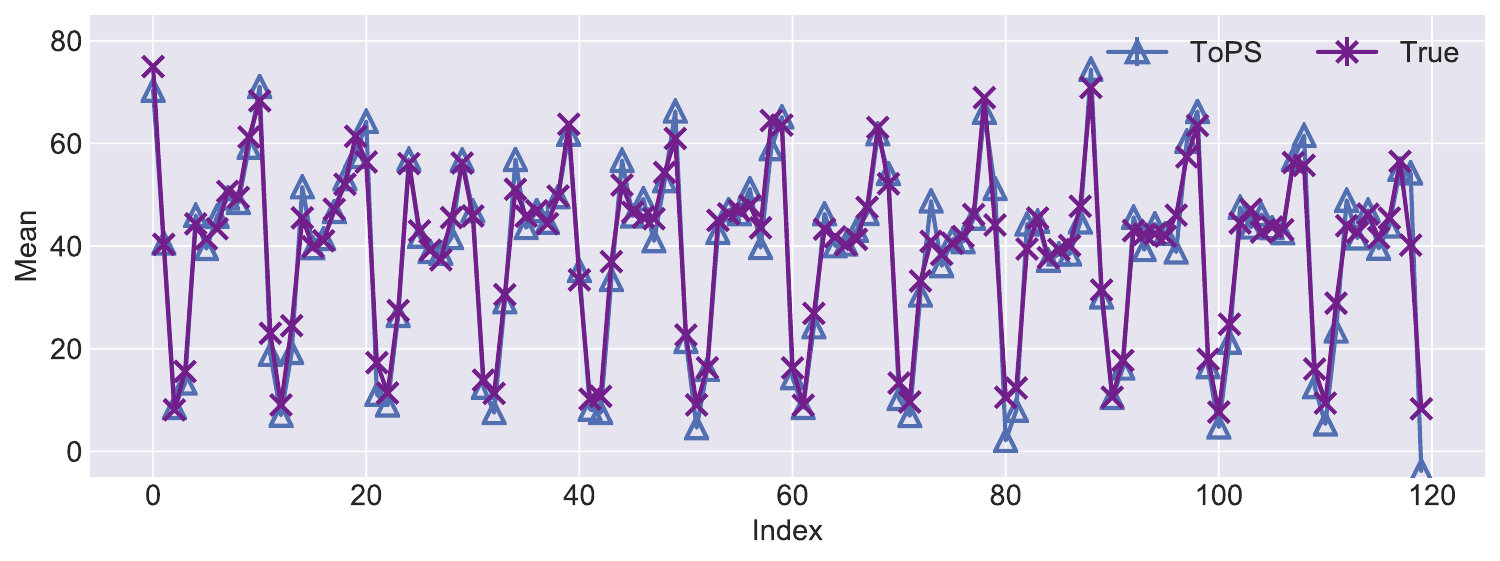}
	}
	\subfigure[\pak, $\epsilon=0.1$]{
		\includegraphics[width=0.48\textwidth]{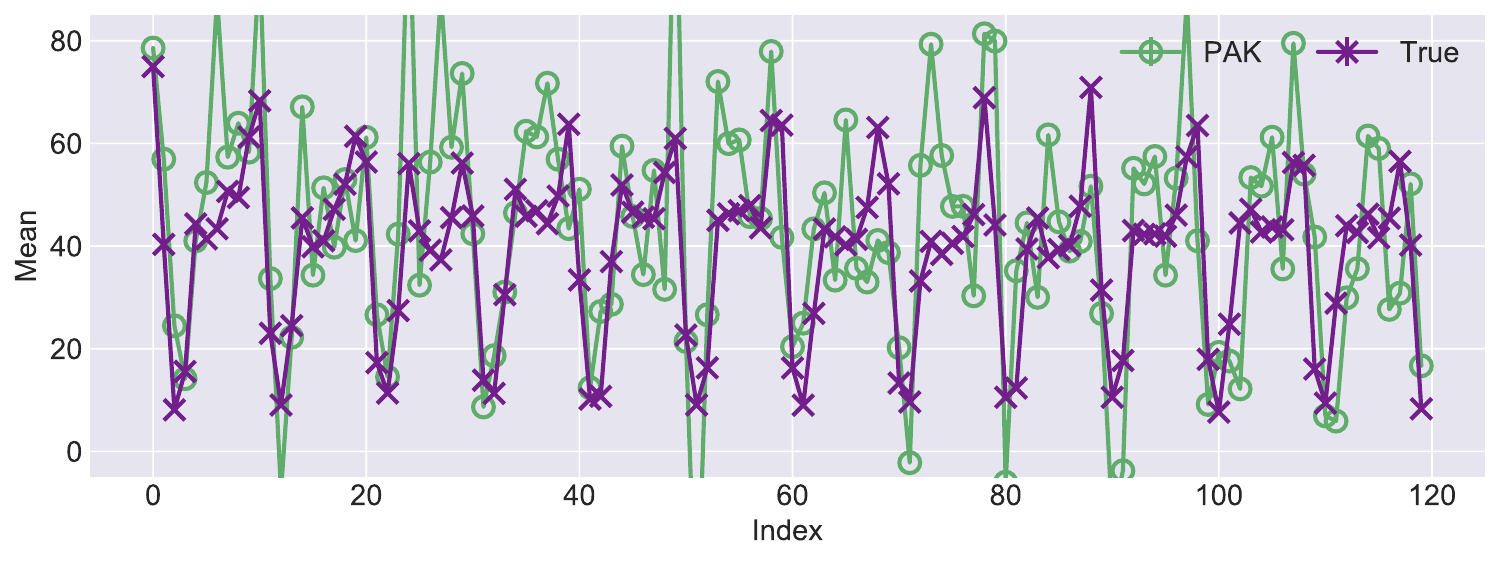}
	}
	\subfigure[\pak, $\epsilon=0.5$]{
		\includegraphics[width=0.48\textwidth]{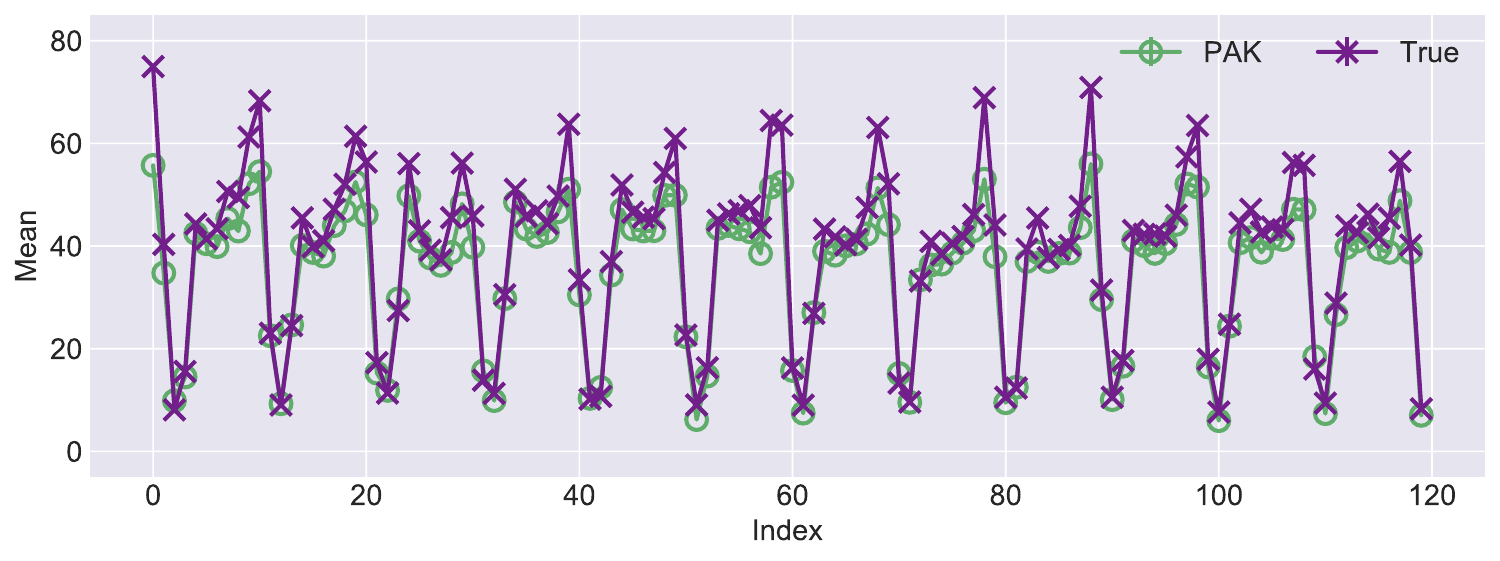}
	}
	
	\caption{Visualizations of the DNS stream.  The $x$-axes correspond to time (we partition the 14-day timeframe into 120 intervals, so each point corresponds to the mean of roughly 9000 values or 1.4 hours), and $y$-axes denotes the moving average.  Our \method at $\epsilon=0.01$ can output predictions that are pretty close to the ground truth.  \pak gives noisier result even with larger $\epsilon$ values. }
	\label{fig:vis}
\end{figure*}

\section{Experimental Evaluation}
\label{sec:exp}
The experiment includes four phases.  First, we give a high-level end-to-end evaluation of the whole process.
Second, we evaluate the performance of the hierarchical method with a fixed truncation threshold.
Third, we fix the hierarchical method and test different algorithms that give the threshold.
% \zzk{Shall we merge the second and third phases as an ablation study?}\tw{let's keep it as it is for now. what's the point of merging?}
Fourth, we evaluate the performance in the local setting.

\subsection{Evaluation Setup}
\mypara{Datasets}
A total of four real-world datasets are examined.
% \zzk{We may move the details of the datasets to appendix if the space runs out.}\tw{sure}
% Our experiments are conducted on four real world datasets.  
\begin{itemize}[leftmargin=*]
    \item DNS: 
    This dataset is extracted from a set of DNS query logs collected by a campus resolver with all user ids and source IP addresses removed~\footnote{The data collection process has been approved by the IRB of the campus.}. It includes 14 days of DNS queries.
    %and its average and maximum of query per second are 38 and 617 respectively. 
    \revision{The number of queries is estimated, which can be used to assess Internet usage in a region.}{The network administrator can use the number of queries to assess Internet usage in a region.  We take number of queries as the stream in the evaluation.}
% median is 35

    % \says{tianhao}{application? data mining task on it?}
    % \says{joann}{are we gonna include the application somewhere in this sec? i'll start writing the application of DNS here in comments}
    % \says{tianhao}{can we say to foster research, our dataset is available upon request? we don't need to give the original full data}
    % \says{joann}{unfortunately, we don't have the permission of publishing the original data in any form}

    \item Fare~\cite{data:nyctaxi}: New York City taxi travel fare.  We use the Yellow Taxi Trip Records for January 2019.
    
    \item Kosarak~\cite{data:kosarak}: A dataset of clickstreams on a Hungarian website that contains around $10^6$ users and 41270 categories. 
    % The data is formatted so that each user click through multiple categories.  
    We take it as streaming data and use the size of click categories as the value of the stream.
    
    \item POS~\cite{data:pos}: A dataset containing merchant transactions of half a million users and 1657 categories.   We use the size of the transaction as the value of the stream.
    % \says{tianhao}{maybe remove pos}
\end{itemize}

Table~\ref{tbl:datadist} gives the distribution statistics of the datasets.  

\begin{table}[!t]
\footnotesize
\centering
\caption{Dataset Characteristics}
\label{tbl:datadist}
% \begin{tabular}{@{}c|c|c|c|c|c|c|c@{}}
\begin{tabular}{@{}c c c c c c c c@{}}
\toprule
Dataset & $n$ & max  & $p_{85}$ & $p_{95}$ & $p_{99.5}$  & \revision{}{$p_{100}$} & $\mathrm{avg}$\\ 
\midrule
DNS & 1141961 & 2000  & 63 & 85 & 135  & \revision{}{617} & 37.9\\ 
Fare & 8704495 & 30000  & 440 & 1036 & 2037  & \revision{}{26770}& 279.9\\
Kosarak & 990002 & 41270  & 10 & 28 & 133 & \revision{}{2498} & 8.1\\ 
POS & 515597 & 1657 & 13 & 21 & 39 & \revision{}{165}  & 7.5\\ 
% URL & 355872 & 500 & 36 & 3 & 4 & 9 & 2.1 \\ 
% time & 8735777 & 21600 & 21588 & 1295 & 1959 & 3638 & 787.2 \\ \hline
% dist & 8704477 & 10000 & 9900 & 440 & 1036 & 2037 & 279.9 \\ \hline
% fare & 8704495 & 30000 & 26770 & 440 & 1036 & 2037 & 279.9 \\ \hline
% \hline
% kosarak & 990002 & 42178 & 2498 & 10 & 28 & 133 & 8.1 \\ \hline
% pos & 515597 & 1638 & 165 & 13 & 21 & 39 & 7.5 \\ \hline
% retail & 88162 & 16470 & 76 & 18 & 27 & 44 & 10.3 \\ \hline
% \hline
% web1 & 59601 & 497 & 267 & 4 & 7 & 21 & 2.5 \\ \hline
% web2 & 77511 & 3340 & 161 & 8 & 15 & 37 & 4.6 \\ \hline
% star & 49046 & 2088 & 63 & 53 & 55 & 57 & 50.5 \\ \hline
% time & 990002 & 21600 & 2497 & & & & 8.10 \\ \hline
% dist & 515596 & 30000 & 164 & & & & 6.53\\ \hline
% fare & 88162 & 10000 & 76 & & & & 10.23\\ \hline \hline
% kosarak & 990002 & 42178 & 2497 & & & & 8.10\\ \hline
% pos & 515596 & 1657 & 164 & & & & 6.53\\ \hline
% retail & 88162 & 16470 & 76 & & & & 10.23\\ \hline \hline
% BMS-WebView-1 & 59602 & 497 & 267 & & & & 2.5\\ \hline
% BMS-WebView-2 & 77512 & 3340 & 161 & & & & 5.0\\ \hline
% pumsbstar & 49046 & 2088 & 63 & & & & 50.48\\ \hline \hline
% AOL & 647377 & 2290685 & 48070 & & & & 34.90\\ \hline
\bottomrule
\end{tabular}
\end{table}

\begin{figure*}[h]
	\centering
	\subfigure[Fare]{
		\includegraphics[width=0.23\textwidth]{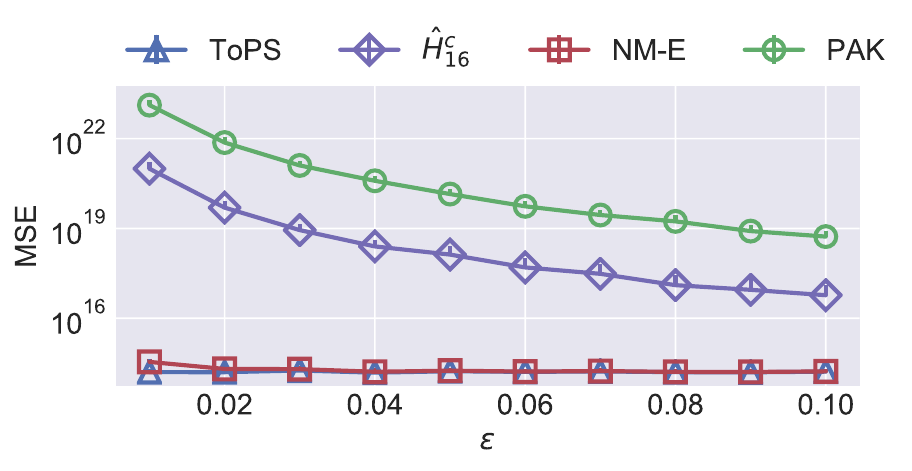}
	}
	\subfigure[DNS]{
		\includegraphics[width=0.23\textwidth]{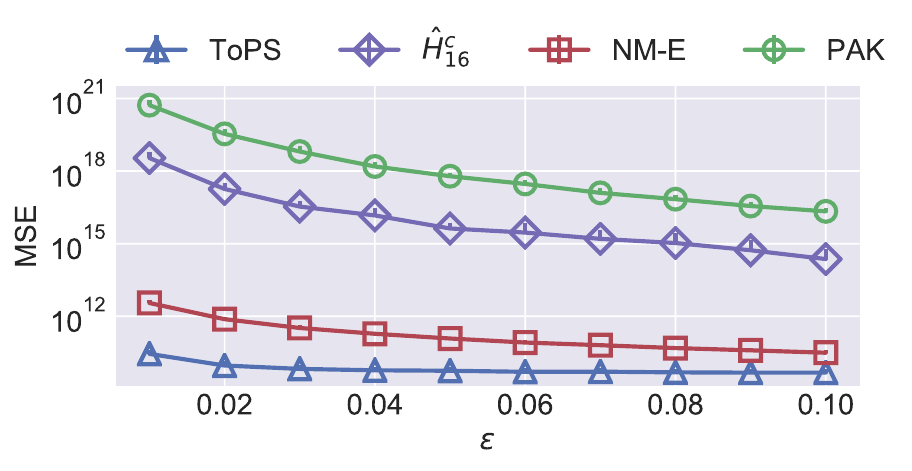}
	}
	\subfigure[Kosarak]{
		\includegraphics[width=0.23\textwidth]{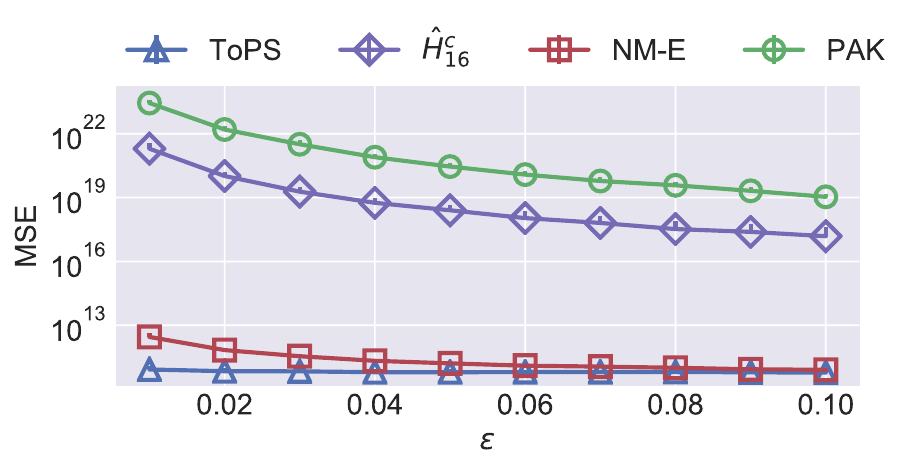}
	}
	\subfigure[POS]{
		\includegraphics[width=0.23\textwidth]{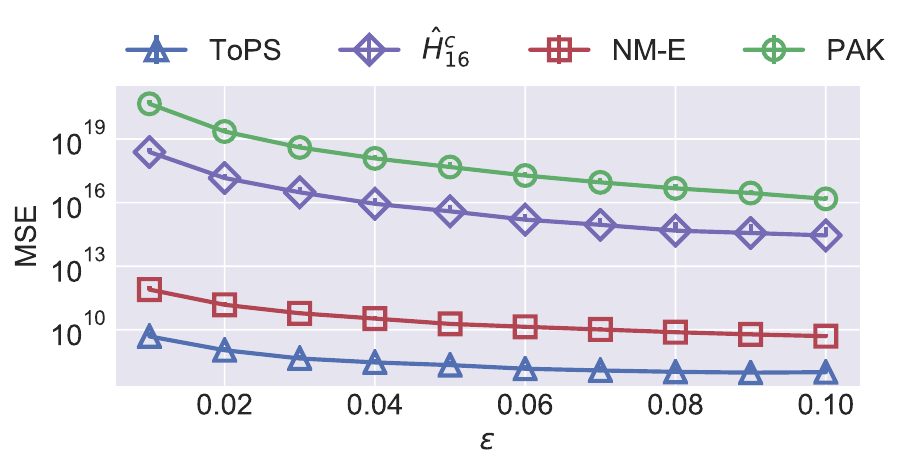}
	}
	\caption{Comparison between \pak and \method when answering range queries.  We also include two intermediate methods NM-E (our proposed threshold optimizer) and $\hat{H}_{16}^c$ (our proposed the perturber and smoother) that replace the corresponding two phases of \pak to demonstrate the performance boost due to our new design.}
	\label{fig:query_mse_vary_e_med}
\end{figure*}

\mypara{Metrics} 
To evaluate the performance of different methods,  we use the metric of Mean Squared Error (MSE) to answer randomly generated queries.  
In particular, we measure
\begin{align}
    \mbox{MSE}(Q) = \frac{1}{|Q|}\sum_{(i,j)\in Q}{\left[\tilde{V}(i,j) - V(i,j)\right]^{2}}.\label{eq:exp_metric}
\end{align}
where $Q$ is the set of the randomly generated queries. 
It reflects the analytical utility measured by \autoref{eq:ana_metric} from \autoref{subsec:formal_problem_def} (to demonstrate the actual accuracy, we also have results for mean absolute error in \autoref{app:more_results}).  
We set $r=2^{20}$ as the maximal range of any query.
% To evaluate the performance of finding the threshold, we use the following: we take the first $m$ values to find $\theta$.  We then duplicate the first $m$ values $10$ times and evaluate the MSE of range queries on the hierarchy.  This helps eliminate the influence from the distribution change.

\begin{figure*}[h]
	\centering
	\subfigure[Fare]{
		\includegraphics[width=0.23\textwidth]{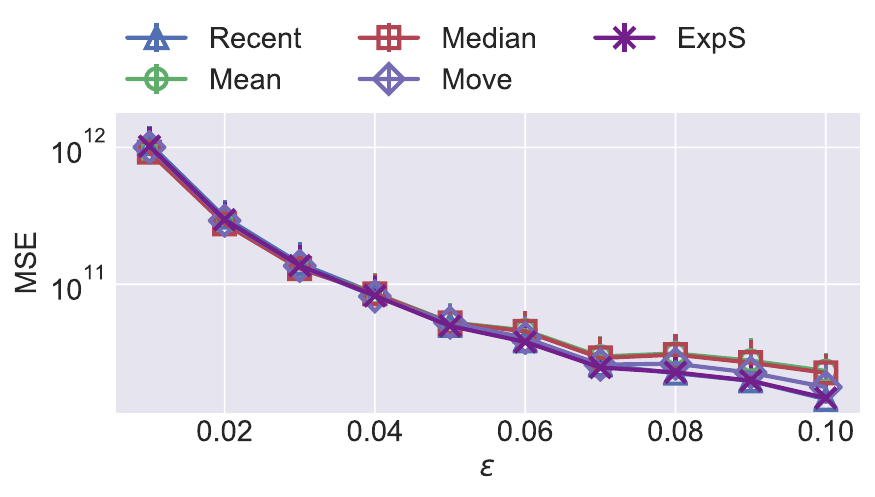}
	}
	\subfigure[DNS]{
		\includegraphics[width=0.23\textwidth]{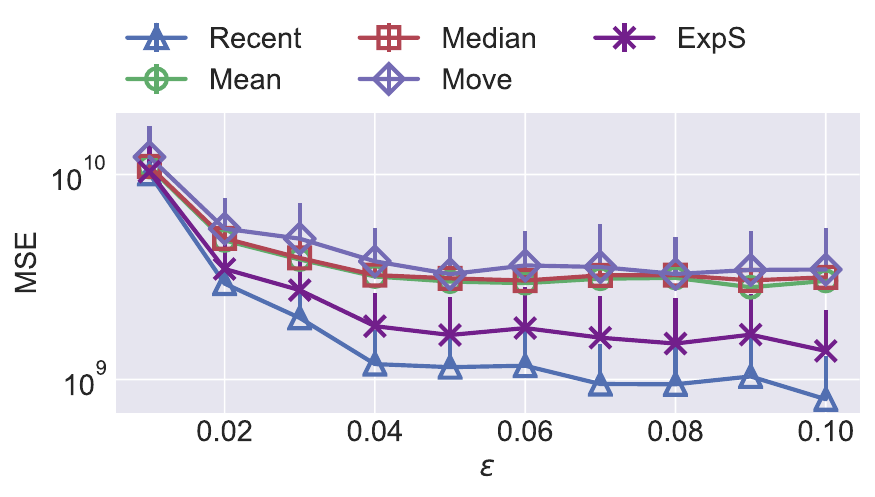}
	}
	\subfigure[Kosarak]{
		\includegraphics[width=0.23\textwidth]{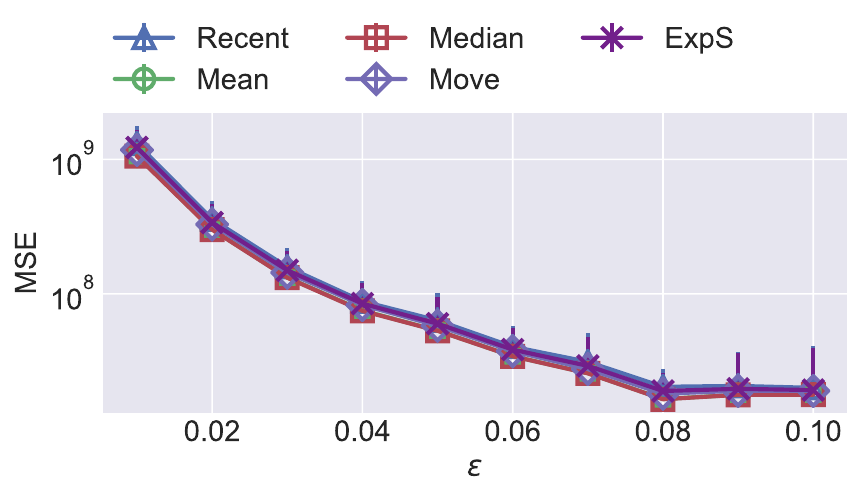}
	}
	\subfigure[POS]{
		\includegraphics[width=0.23\textwidth]{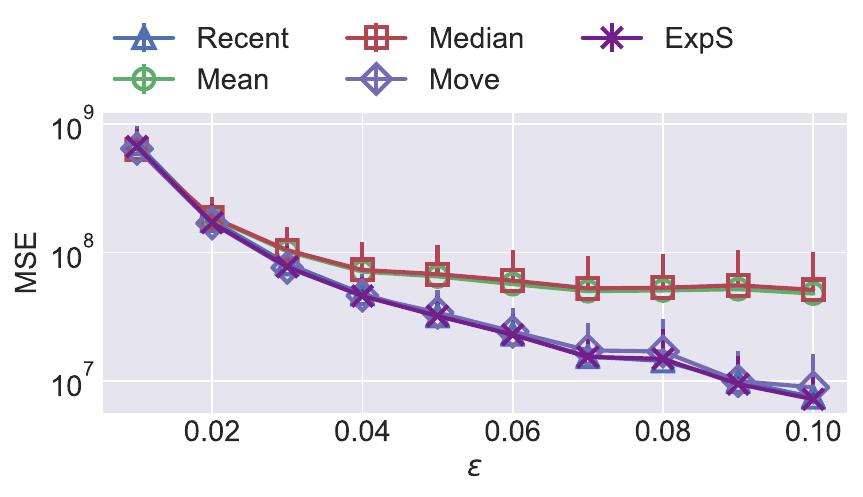}
	}
	\caption{Evaluation of different smoothing techniques.  We vary $\epsilon$ from $0.01$ to $0.1$ in the $x$-axis. The $y$-axis shows the query accuracy (MSE).}
	\label{fig:smooth_vary_e_med}
\end{figure*}

\mypara{Methodology} 
The prototype was implemented using Python 3.7.3 and NumPy 1.15.3 libraries.  
The experiments were conducted on servers running Linux kernel version 5.0 with Intel Xeon E7-8867 v3 CPU @ 2.50GHz and 576GB memory.  
For each dataset and each method, we randomly choose 200 range queries and calculate their MSE. We repeat each experiment 100 times and report the result of mean and standard deviation.  Note that the standard deviation is typically very small, and barely noticeable in the figures.

% \begin{figure*}[h]
% 	\centering
	
% 	\subfigure[Time (seconds)]{
% 		\includegraphics[width=0.23\textwidth]{figure/dist/time_cdf}
% 	}
% 	\subfigure[Fare (cents)]{
% 		\includegraphics[width=0.23\textwidth]{figure/dist/fare_cdf}
% 	}
% 	\subfigure[Distance (0.01mi)]{
% 		\includegraphics[width=0.23\textwidth]{figure/dist/dist_cdf}
% 	}
	
% 	\subfigure[kosarak]{
% 		\includegraphics[width=0.23\textwidth]{figure/dist/kosarak_cdf}
% 	}
% 	\subfigure[pos]{
% 		\includegraphics[width=0.23\textwidth]{figure/dist/pos_cdf}
% 	}
% 	\subfigure[retail]{
% 		\includegraphics[width=0.23\textwidth]{figure/dist/retail_cdf}
% 	}
	
% 	\subfigure[web1]{
% 		\includegraphics[width=0.23\textwidth]{figure/dist/web1_cdf}
% 	}
% 	\subfigure[web2]{
% 		\includegraphics[width=0.23\textwidth]{figure/dist/web2_cdf}
% 	}
% 	\subfigure[star]{
% 		\includegraphics[width=0.23\textwidth]{figure/dist/star_cdf}
% 	}
	
% 	\caption{Statistical properties of the datasets.  The top row indicates the CDF of the distributions.
% 	In the bottom row, the $x$-axes indicates the sorted value index and the $y$-axes is its count.}
% 	\label{fig:dist}
% \end{figure*}

\subsection{End-to-end Comparison}
First, as a case study, we visualize the estimated stream of our method \method on the DNS dataset (\autoref{fig:vis}).  
The top row shows the performance of \method while the bottom row shows that of \pak.  
We run algorithms once for each setting to demonstrate the real-world usage.  Similar to the setting of \pak, in \method, we use the first $m=65,536$ observations to obtain the threshold $\theta$ (we will show later that \method does not need this large $m$ observations to be held).  
\autoref{fig:vis} indicates that our method \method can give fairly accurate predictions when $\epsilon$ is very small.  On the other hand, \pak, though under a larger $\epsilon$, still performs worse than \method.  Note that to obtain the threshold $\theta$, \pak satisfies $(\epsilon, \delta)$-DP while our \method satisfies pure $\epsilon$-DP during the whole process. 
% In \autoref{subsec:delta}, we show that there is no benefit to introduce $\delta$ in our setting.

We then compare the performance of \method and \pak with our metric of MSE given in \autoref{eq:exp_metric}, and show the results in \autoref{fig:query_mse_vary_e_med}.  Between \method and \pak, we also include two intermediate methods that replace Phase 1 (finding $\theta$) and Phase 2 (hierarchical method) of \pak by our proposed method NM-E (used in threshold optimizer) and $\hat{H}_{16}^c$ (used for the perturber and smoother together), respectively, to demonstrate the performance boost due to our new design (we will evaluate the two phases in more details in later subsections).  From the figure, we can see that the performance of all the algorithms gets better as $\epsilon$ increases, which is as expected.  Second, our proposed \method can outperform \pak by $7$ to $11$ orders of magnitude.  Third, the effect (in terms of improving utility) using NM-E is much more significant than using $\hat{H}_{16}^c$.  Interestingly, the performance of ToPS and NM-E is similar in the Fare and Kosarak datasets.  This is because in these cases, the bias (due to truncation by $\theta$) is dominant.

\begin{figure*}[h]
	\centering
	\subfigure[Fare]{
		\includegraphics[width=0.23\textwidth]{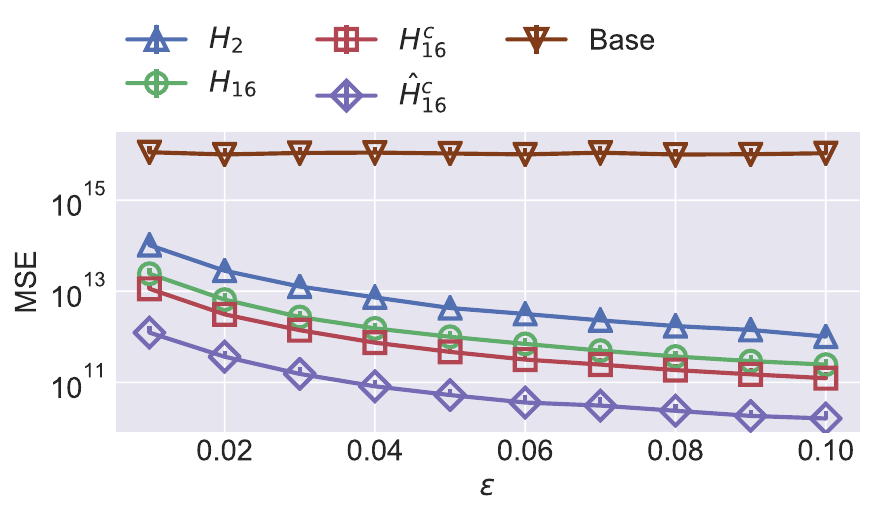}
	}
	\subfigure[DNS]{
		\includegraphics[width=0.23\textwidth]{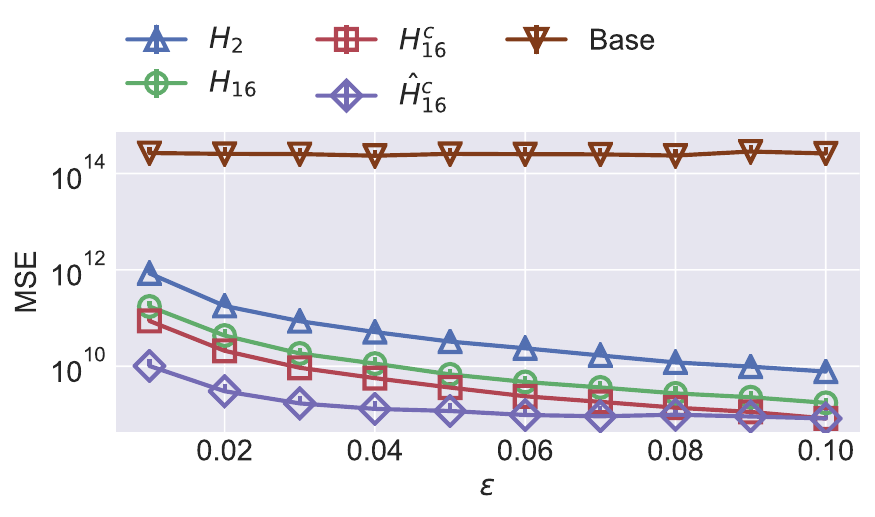}
	}
	\subfigure[Kosarak]{
		\includegraphics[width=0.23\textwidth]{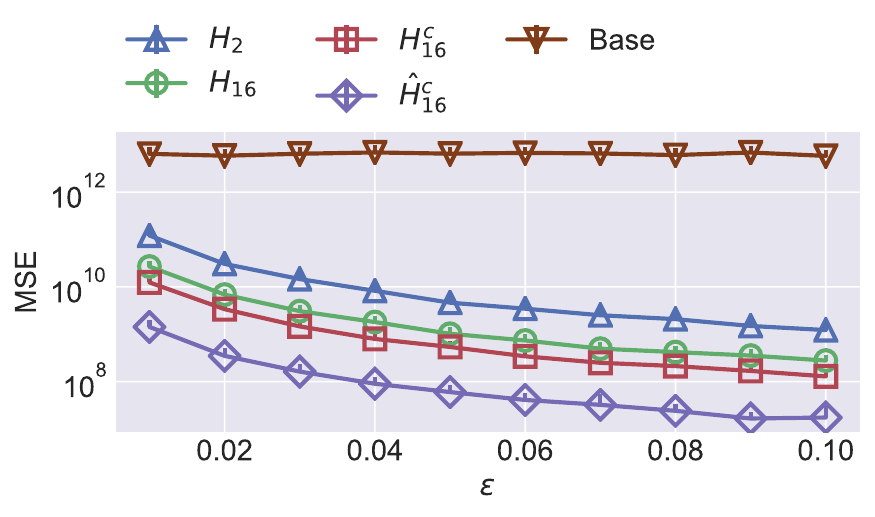}
	}
	\subfigure[POS]{
		\includegraphics[width=0.23\textwidth]{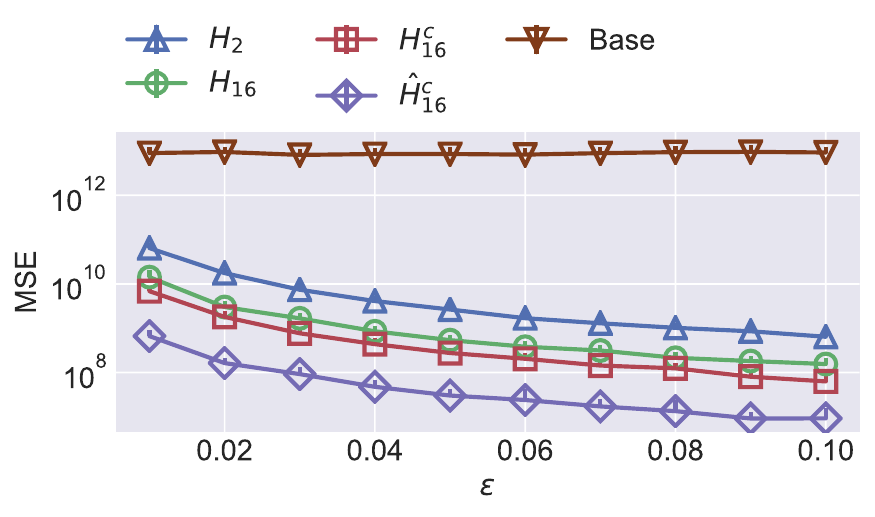}
	}
	\caption{Evaluation of different methods of outputting the stream.  We vary $\epsilon$ from $0.01$ to $0.1$ in the $x$-axis and plot the query accuracy (MSE) in the $y$-axis.}
	\label{fig:hie_vary_e_med}
\end{figure*}

\subsection{Comparison of Stream Publication Phase}
% \says{zhikun}{it seems that the performance of the next subsection is also range query; thus i think it is not proper to name the title of this subsection as performance of range queries.}
Several components contribute to the promising performance of \method.  To demonstrate the precise effect of each of them, we next analyze them one by one in the reverse order.  We first fix other \revision{things}{configurations} and compare different smoothers and perturbers.  Subsequently, we analyze the methods of obtaining the threshold $\theta$ in \autoref{sec:exp_threshold}.  To make the comparison clear, we set $\theta$ to be the $95$-th percentile of the values.  Moreover, we assume the true values are no larger than $\theta$ (the ground truth is truncated).  We will compare the performance of different methods in obtaining $\theta$ in \autoref{sec:exp_threshold}.

\mypara{Comparison of Different Smoothers}
Fixing a threshold $\theta$ and the hierarchical method optimized in \autoref{sec:perturber}, we now compare the performance of five smoother algorithms listed in \autoref{sec:smoother} (note that the smoothers will replace the $16^s$ values where $s$ is given in \autoref{eq:mse_guess}).  
\autoref{fig:smooth_vary_e_med} shows the MSE of the smoothers given $\epsilon$ from $0.01$ to $0.1$.
%we vary the value of $\epsilon$ from $0.01$ to $0.1$ and plot the MSE of the smoothers.
As $\epsilon$ increases, that is, privacy budget loosens, the overall performance improves although the difference is very small in Fare and Kosarak datasets.
% Overall the performance gets better when $\epsilon$ increases, which is as expected.  The difference of them are very small in the Fare and Kosarak datasets.  
This is because there is no clear pattern in these datasets.  In the DNS dataset, Recent performs better than others, as the data is stable in the short term.  In POS, the method of Mean and Median performs worse than the other three.  This is because Mean and Median consider all the history (all the previous $u_i$ values given from the hierarchy, as described in \autoref{sec:smoother}), while the other methods consider more recent results.  Methods that utilize the recent output (i.e., the more recent $u_i$) will perform better due to the stability property in the dataset (similar to the case of DNS).
Since Recent performs the best in DNS, and is among the best in other datasets, we use it as the default smoother algorithm.

Note that large MSE does not always mean poor utility. Here MSE is large because (1) the original values are large (e.g., $>$1000), (2) the range query takes the sum, and (3) the square operation further enlarges the values.  We also have results for mean absolute error (MAE) and mean query results to better demonstrate the utility in \autoref{app:more_results}.

% \begin{figure*}[h]
% 	\centering
% 	\subfigure[Fare]{
% 		\includegraphics[width=0.23\textwidth]{figure/eps_small/fare_m65536_smooth}
% 	}
% 	\subfigure[DNS]{
% 		\includegraphics[width=0.23\textwidth]{figure/eps_small/dns_1k_m65536_smooth}
% 	}
% 	\subfigure[Kosarak]{
% 		\includegraphics[width=0.23\textwidth]{figure/eps_small/kosarak_m65536_smooth}
% 	}
% 	\subfigure[POS]{
% 		\includegraphics[width=0.23\textwidth]{figure/eps_small/pos_m65536_smooth}
% 	}
% 	\caption{Small-range evaluation of different smoothing techniques.}
% 	\label{fig:smooth_vary_e_small}
% \end{figure*}

\mypara{Comparison of Different Hierarchical Algorithms}
To demonstrate the precise effect of each design detail, we line up several intermediate protocols for the hierarchy and threshold, respectively.  
For the hierarchy component, we evaluate:
\begin{itemize}[leftmargin=*]
    \item $H_2$: Original binary tree used in \pak.
    \item $H_{16}$: \revision{A good}{The optimal} fan-out $b=16$ is used in the hierarchy.
    \item $H_{16}^c$: \revision{The optimal}{} $H_{16}$ with consistency method.
    \item $\hat{H}_{16}^c$: We use the hat notation to denote the Recent smoother.  It is built on top of $H_{16}^c$.
    % We apply the mean estimation of the previous node as the leaf estimates.  Note that initially there is no estimations from the hierarchy.  We use half of the threshold as mean.
    \item Base: A baseline method that always outputs $0$. It is used to understand whether a method gives meaningful results. 
\end{itemize}

% \mypara{Setup}

% Having examined the performance of different methods of outputting the threshold, we now switch gear to look at the performance of the answering range queries. 
% We fix a threshold $\theta$ to be the $95$-th percentile of the values; and all the methods uses the same threshold $\theta$. 
% To eliminate the influence of unexpected factors, we truncate the ground truth stream by $\theta$ when evaluating range queries.

% \mypara{Varying \epsilon}
\autoref{fig:hie_vary_e_med} gives the result varying $\epsilon$ from $0.01$ to $0.1$.  First of all, all methods (except the baseline) yield better accuracy as $\epsilon$ increases, which is as expected.  Moreover, the performance of all methods (except $\hat{H}_{16}^c$) increases by a factor of $100\times$ when $\epsilon$ increases from $0.01$ to $0.1$.  This observation is consistent with the analysis that variance is proportional to $1/\epsilon^2$.  Comparing each method, using the optimal branching factor ($H_{16}$ versus $H_2$) can improve the performance by $5\times$. Moreover, we have approximately another $2\times$ ($H^c_{16}$ versus $H_{16}$) of accuracy boost by adopting the consistency algorithm.
For $\hat{H}_{16}^c$, a constant $10\times$ improvement can be observed over ${H}_{16}^c$ except in the DNS dataset, where $\hat{H}_{16}^c$ performs roughly the same as ${H}_{16}^c$ when $\epsilon > 0.08$.  
% But when $\epsilon$ is larger, its performance is actually worse than ${H}_{16}^c$.  
% This effect is similar to the performance of NM-E in the threshold evaluation.  
The reason is that the error of $\hat{H}_{16}^c$ is composed of two parts. One is the noise error from the constraint of DP, the other the bias error of outputting the predicted values.  When $\epsilon$ is large, the bias dominates the noise in the DNS dataset.

% We then examine the smaller $\epsilon$ range from $0.001$ to $0.01$ and plot the result in \autoref{fig:hie_vary_e_small}.  The overall trend and conclusion is similar to those derived from the previous figure.  In this case, the performance of $H_{2}$ and $H_{16}$ is already similar to the baseline method of always outputting $0$ when $\epsilon=0.001$.  Note that $H_{2}$ is used in \pak.  This means even if the $\theta$ is accurate (in the original \pak, $\theta$ is given by smooth sensitivity and is unlikely to be accurate), the overall performance is still not significantly better than always outputting $0$.  Note that in this setting, the performance of $\hat{H}_{16}^c$ is consistently than other methods.  This indicates that the noise error becomes dominant in this setting.  

\begin{figure}[t]
	\centering
% 	\subfigure[Fare]{
% 		\includegraphics[width=0.23\textwidth]{figure/eps_range_small/fare_m4096_ell_est}
% 	}
    \subfigure[$m=4096$, smaller $\epsilon$ range]{
		\includegraphics[width=0.22\textwidth]{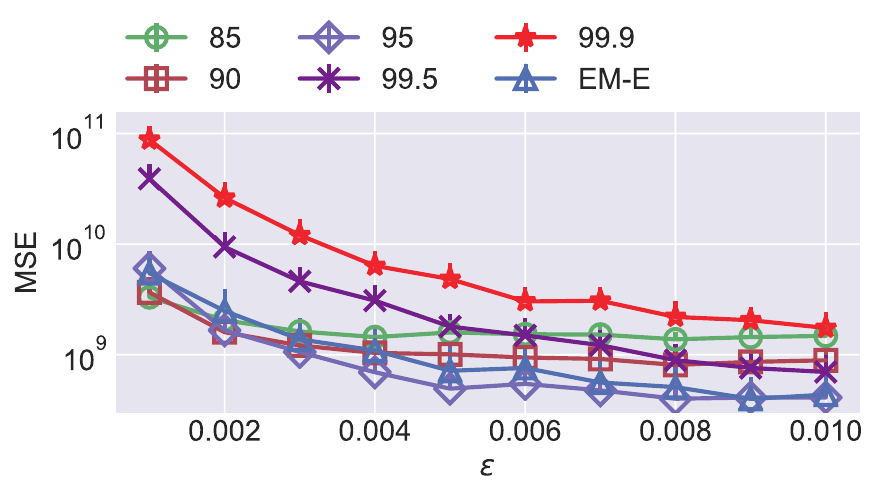}
	}
% 	\subfigure[Kosarak]{
% 		\includegraphics[width=0.23\textwidth]{figure/eps_range_small/kosarak_m4096_ell_est}
% 	}
% 	\subfigure[POS]{
% 		\includegraphics[width=0.23\textwidth]{figure/eps_range_small/pos_m4096_ell_est}
% 	}
% 	\subfigure[Fare]{
% 		\includegraphics[width=0.23\textwidth]{figure/eps_range_small/fare_m65536_ell_est}
% 	}
    \subfigure[$m=65536$, smaller $\epsilon$ range]{
		\includegraphics[width=0.22\textwidth]{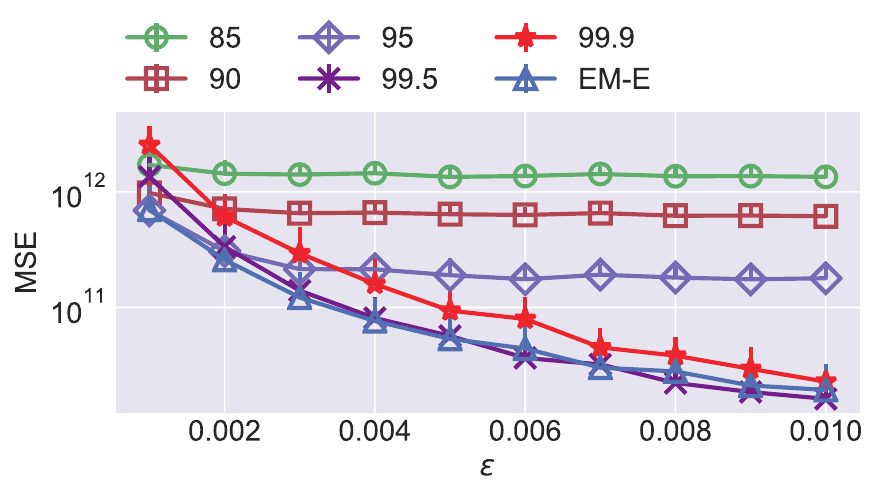}
		\label{subfig:diff_p_m65536_eps_small}
	}
% 	\subfigure[Kosarak]{
% 		\includegraphics[width=0.23\textwidth]{figure/eps_range_small/kosarak_m65536_ell_est}
% 	}
% 	\subfigure[POS]{
% 		\includegraphics[width=0.23\textwidth]{figure/eps_range_small/pos_m65536_ell_est}
% 	}

	\centering
% 	\subfigure[Fare]{
% 		\includegraphics[width=0.23\textwidth]{figure/eps_range_medium/fare_m4096_ell_est}
% 	}
	\subfigure[$m=4096$, larger $\epsilon$ range]{
		\includegraphics[width=0.22\textwidth]{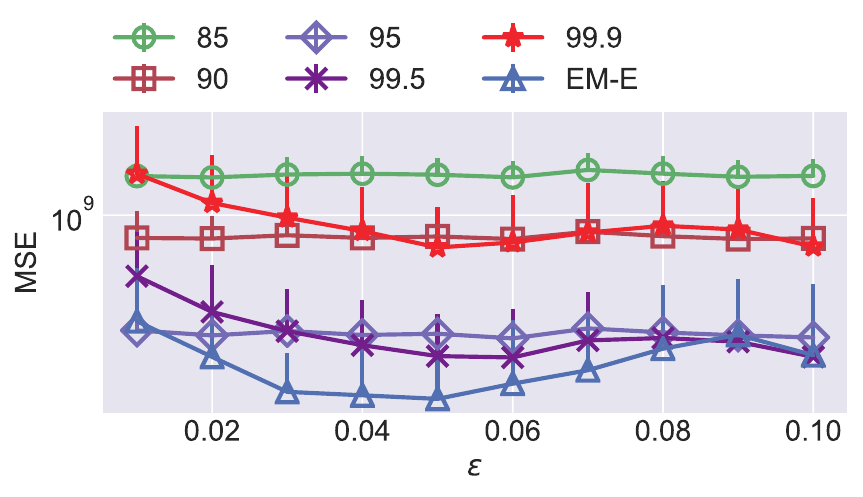}
	}
% 	\subfigure[Kosarak]{
% 		\includegraphics[width=0.23\textwidth]{figure/eps_range_medium/kosarak_m4096_ell_est}
% 	}
% 	\subfigure[POS]{
% 		\includegraphics[width=0.23\textwidth]{figure/eps_range_medium/pos_m4096_ell_est}
% 	}
% 	\subfigure[Fare]{
% 		\includegraphics[width=0.23\textwidth]{figure/eps_range_medium/fare_m65536_ell_est}
% 	}
	\subfigure[$m=65536$, larger $\epsilon$ range]{
		\includegraphics[width=0.22\textwidth]{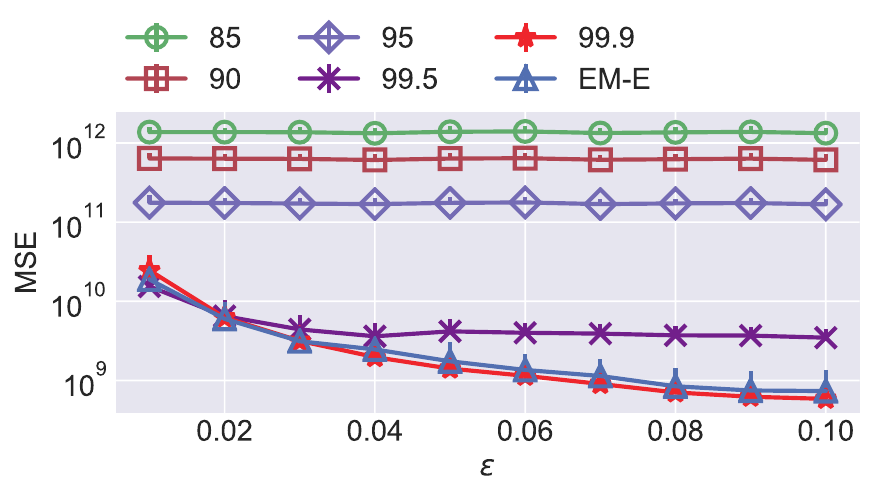}
	    \label{subfig:diff_p_m65536_eps_medium}
	}
% 	\subfigure[Kosarak]{
% 		\includegraphics[width=0.23\textwidth]{figure/eps_range_medium/kosarak_m65536_ell_est}
% 	}
% 	\subfigure[POS]{
% 		\includegraphics[width=0.23\textwidth]{figure/eps_range_medium/pos_m65536_ell_est}
% 	}
	\caption{MSE of answering range queries using $\hat{H}^c_{16}$ on DNS dataset.  The true $p$-th percentile for different $p$ values are evaluated. We also include NM-E, which uses $\epsilon=0.05$.}
	\label{fig:diff_p}
\end{figure}

\begin{figure*}[h]
	\centering
	\subfigure[Fare]{
		\includegraphics[width=0.23\textwidth]{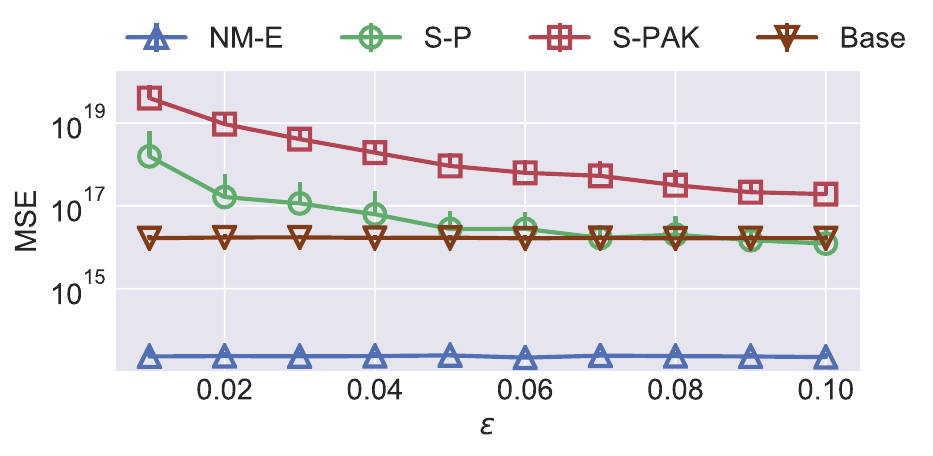}
	}
	\subfigure[DNS]{
		\includegraphics[width=0.23\textwidth]{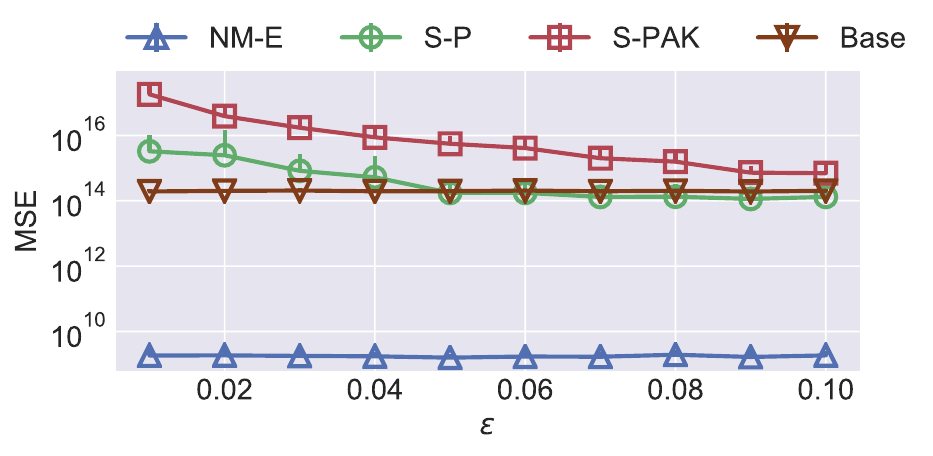}
	}
	\subfigure[Kosarak]{
		\includegraphics[width=0.23\textwidth]{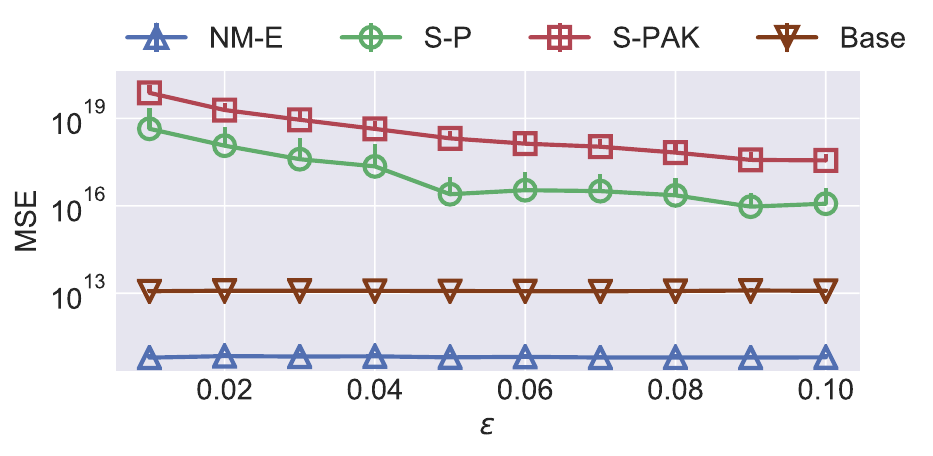}
	}
	\subfigure[POS]{
		\includegraphics[width=0.23\textwidth]{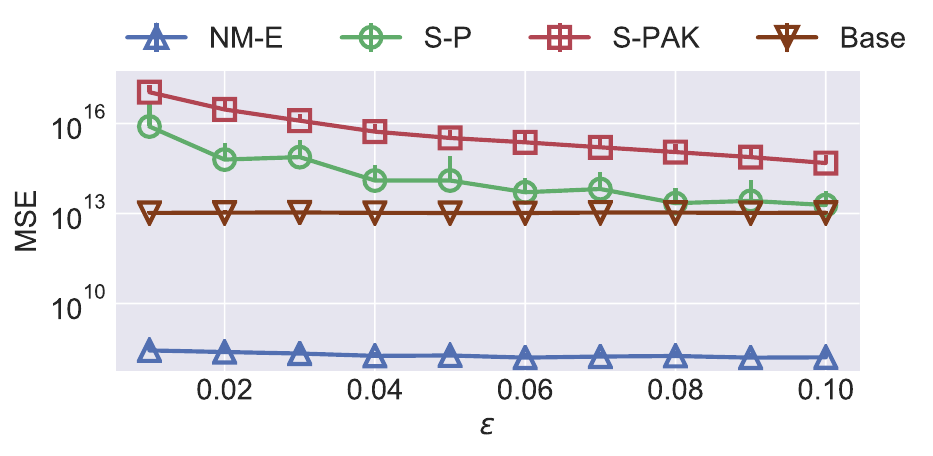}
	}
	\caption{Evaluation of different methods to find the threshold $\theta$.  We vary $\epsilon$ from $0.01$ to $0.1$ in the $x$-axis.  The $y$-axis shows the MSE of answering range queries using $\hat{H}^c_{16}$ (to make comparison clear, we use a fixed $\epsilon=0.05$ for it).  Base is a baseline method that always outputs $0$.}
	\label{fig:thres_vary_e_medium}
\end{figure*}

\begin{figure*}[h]
	\centering
	\subfigure[Fare]{
		\includegraphics[width=0.23\textwidth]{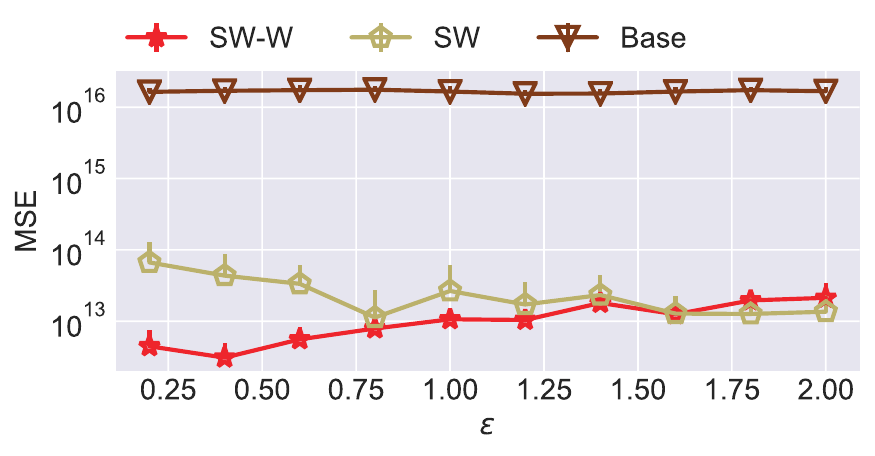}
	}
	\subfigure[DNS]{
		\includegraphics[width=0.23\textwidth]{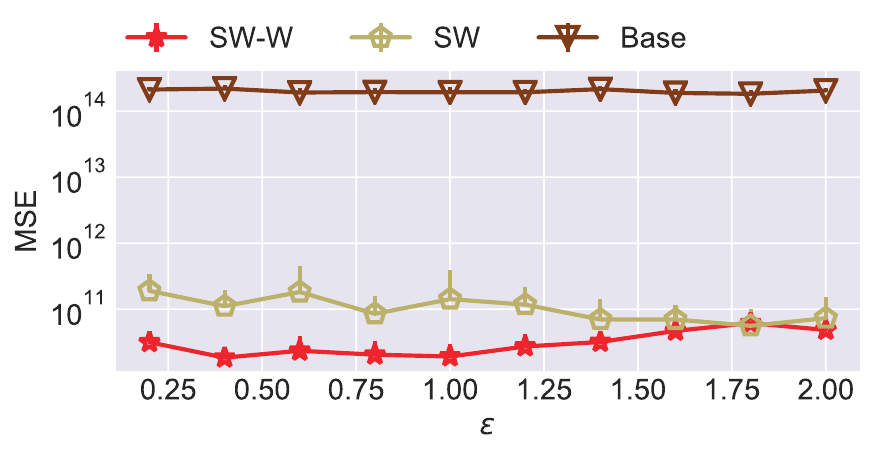}
	}
	\subfigure[Kosarak]{
		\includegraphics[width=0.23\textwidth]{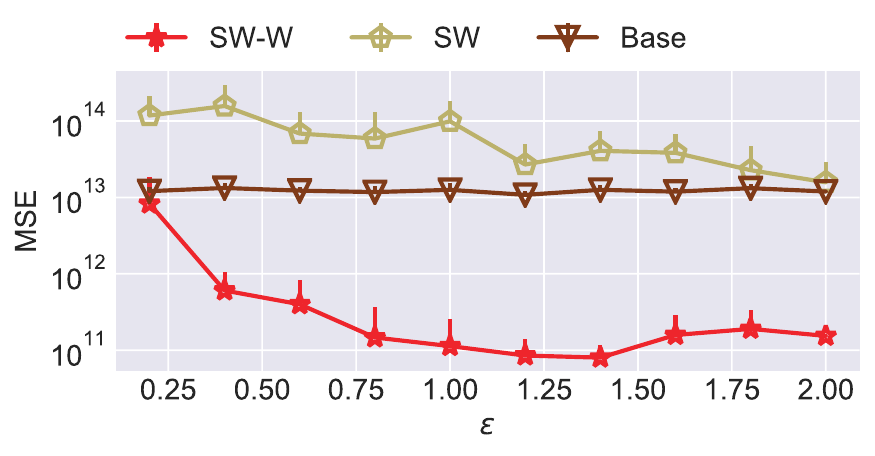}
	}
	\subfigure[POS]{
		\includegraphics[width=0.23\textwidth]{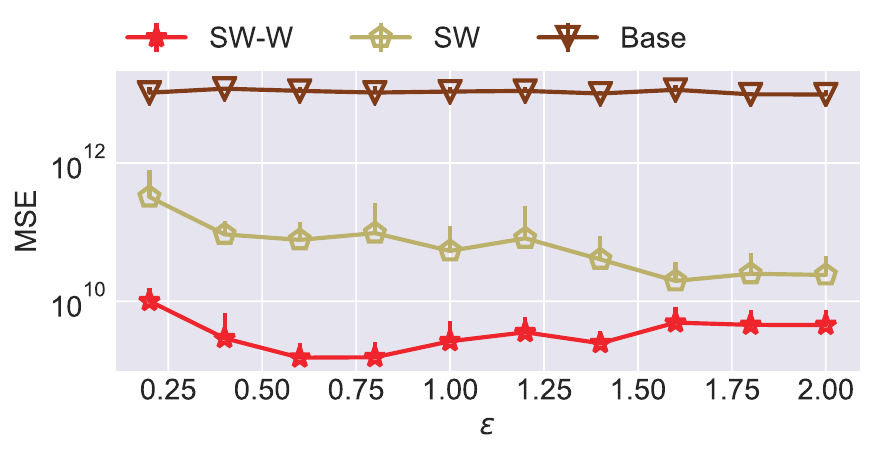}
	}
	\caption{LDP evaluation of different methods of outputting the threshold. We vary $\epsilon$ from $0.2$ to $2$ in the $x$-axis.  The $y$-axis shows the query accuracy (MSE).}
	\label{fig:ldp_threshold_vary_e}
\end{figure*}

\begin{figure}[h]
	\centering
	\subfigure{
		\includegraphics[width=0.48\textwidth]{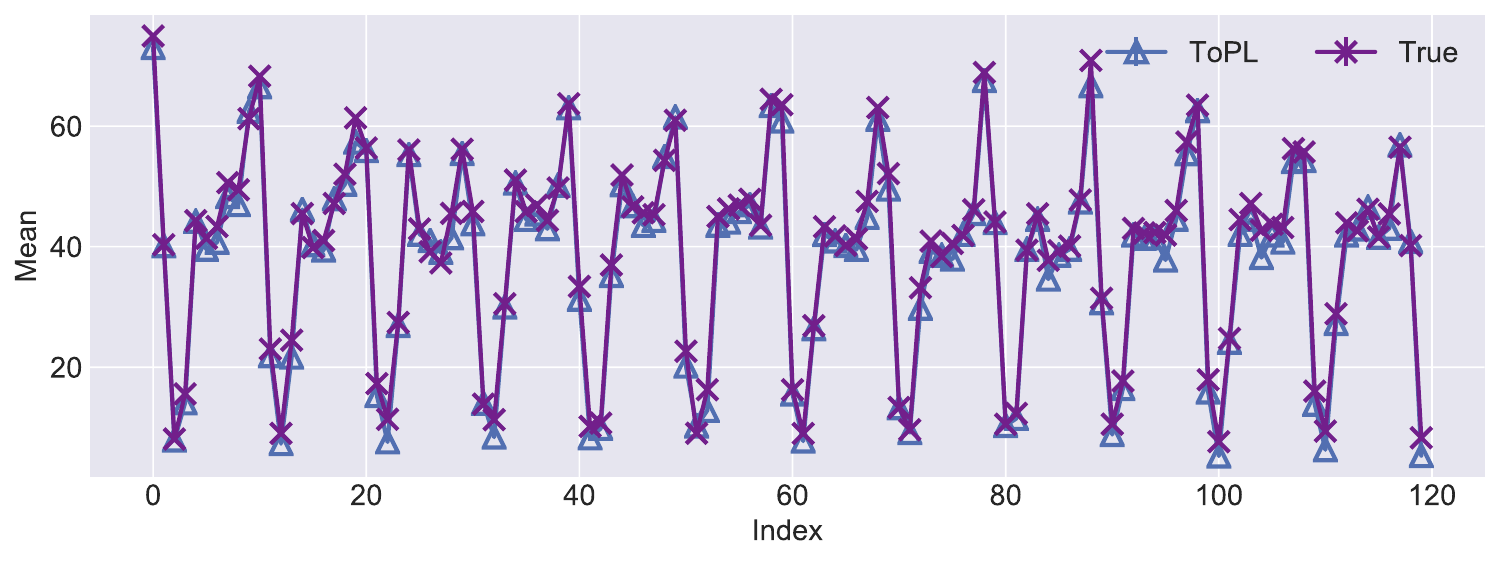}
	}
	\caption{Visualizations of the DNS stream.  The $x$-axes correspond to the time, and the $y$-axes denote the moving average.  Our \methodldp at $\epsilon=1$ can output predictions that are pretty close to the ground truth.}
	\label{fig:vis_ldp}
\end{figure}

\subsection{Comparison of Threshold Phase}
\label{sec:exp_threshold}
% \says{zhikun}{shall we just set this as subsection instead of subsubsection, and remove evaluation results?}
% We first compare different methods' performance on estimating the threshold.
After examining the performance of different methods of outputting the stream, we now switch gear to look at the algorithms for finding the threshold. 

\mypara{Setup}
Following \pak's method~\cite{ndss:perrier2019private}, we use the first $m$ values to obtain the threshold.  
To eliminate the unexpected influence of distribution change, we use the same $m$ values for now to build the hierarchy using $\hat{H}^c_{16}$ (with the best smoothing strategy called Recent smoother).  
% We will use values coming after the first $m$ to build the hierarchy in the later evaluation.
% By default, we set $m=65536$ and $\epsilon=0.1$.

\mypara{No Single Quantile Works Perfectly for All Scenarios}
To show that no single $p$-quantile can work perfectly for all scenarios, we choose the true $p$-quantile for $p\in\{85, 90, 95, 99.5, 99.9\}$ and test in different scenarios (with different $\epsilon$ and $m$ values).  We also include our threshold optimizer (NM-E) which is introduced to find a threshold only based on the estimated error and set $\epsilon=0.05$ for it.
\autoref{fig:diff_p} shows the results of answering range queries given these true percentiles and it reveals several findings.  First, the performance improves as $\epsilon$ increases for all $p$ values.  Second, in some cases, the performance improvement is negligible with respect to $\epsilon$ (e.g., $p=80$ and $85$ in \autoref{subfig:diff_p_m65536_eps_small} and $p=85$, $90$ and $95$ in \autoref{subfig:diff_p_m65536_eps_medium}).  This is because $p$ is too small in these scenarios, which makes the bias dominate the noise error.
Third, our threshold optimizer with $\epsilon=0.05$ can achieve similar performance with the optimal $p$-quantile.

% \begin{itemize}
    % \item NP-P: The non-private method that output the exact $p$-percentile.
    % \item SVT-P: The SVT method given in Algorithm~\ref{algo:svt_tau}.  
    % \item NM-P: The NM variant of SVT-P.
% \end{itemize}
% We compare our algorithm against the following methods:
% \begin{itemize}
%     \item \pak (SS-\pak + $H_2$).  This is the major method we compete with.  
%     \item \method (NM-E + $\hat{H}_{16}^c$).  This is the method we propose.
% \end{itemize}

\mypara{Varying $\epsilon$}
We then compare three methods that output $\theta$:
\begin{itemize}[leftmargin=*]
    \item NM-E: The threshold optimizer in \method.  It does not require a percentile.
    \item S-\pak: The smooth sensitivity method used by \pak.  We use $p=99.5$, as used by \pak.
    \item S-P: The original smooth sensitivity method.  Similar to S-\pak, we also use $p=99.5$.
\end{itemize}
In \autoref{fig:thres_vary_e_medium}, we compare with existing differentially private methods on finding the threshold $\theta$.
We vary the value of $\epsilon$ for obtaining $\theta$, and use $\hat{H}^c_{16}$ to answer range queries.   Note that to make the comparison clearer, we fix $\epsilon=0.05$ in $\hat{H}^c_{16}$.  In all the datasets, our proposed NM-E performs much better than existing methods in terms of MSE.  Moreover, the performance does not change much when $\epsilon$ increases. The reason is that the output of NM-E is stable even with small $\epsilon$.  Finally, both S-\pak and S-P perform worse than the baseline method, which always give $0$ regardless of input values, indicating the $\theta$ given by them is too large.

\subsection{Performance of \methodldp}
In this section, we evaluate the LDP algorithm \methodldp.  We first check the methods to find $\theta$.  Following the setting of \method, we use the first $m=65,536$ observations to obtain the threshold $\theta$.  The difference from the DP setting is that we vary $\epsilon$ in a larger range (from $0.2$ to $2$) due to the larger amount of noise of LDP.  Our method finds $\theta$ by using the Square Wave (SW) mechanism to estimate the distribution, and then minimizing \autoref{eq:range_error_ldp}.  Moreover, our approach (the final part of \autoref{sec:threshold_ldp}) modifies SW to exploit the prior knowledge that the distribution is skewed.  We use SW-W to denote this method. In addition, we include another method for comparison, which uses SW as a black box (and we use SW to denote it).  

\autoref{fig:ldp_threshold_vary_e} shows the performance of finding $\theta$ in \methodldp.  We vary the $\epsilon$ used to find $\theta$ from $0.2$ to $2$ while fixing $\epsilon$ used to output the stream to $1$.  
The performance of SW improves with $\epsilon$ in all four datasets.  The performance of our method is less stable and not improving with $\epsilon$, but SW-W can still outperform SW as well as the baseline in all the four datasets.  In the Fare and DNS datasets, the advantage of SW-W is not significant when $\epsilon\ge 1.6$.  But in the Kosarak and POS datasets, the performance of SW-W can be as large as $3$ orders of magnitude (measured by MSE of answering random range queries) \revision{.}{ compared to SW.}

In \autoref{fig:vis_ldp}, we visualize the estimated stream of our method \methodldp on the DNS dataset using $\epsilon=1$.  We run algorithms only once to demonstrate real-world usage. Clearly, \methodldp and ground truth are on similar trajectories in the figure which means our method \methodldp can give pretty accurate predictions.

%!TEX root = main.tex

\section{Related Work}
\label{sec:related}
\subsection{Dealing with Streaming Data in DP}
A number of solutions have been proposed for solving the problem of releasing real-time aggregated statistics under differential privacy (DP).  Here, besides the \pak's approach~\cite{ndss:perrier2019private}, we briefly describe several other related works.

\mypara{Event-level DP}
The first line of work is the hierarchical method for handling binary streams, proposed concurrently by Dwork et al.~\cite{stoc:dwork2010differential} and Chan et al.~\cite{tissec:chan2011private}.  Event-level DP is satisfied in this case.
Both works assumed each value in the stream is either $0$ or $1$ and proposed a differentially private continual counter release algorithm over a fixed-length binary stream with a bounded error $O\left(\left(\log^{1.5}n\right)/\epsilon\right)$ at each time step, where $n$ is the stream length. 
In addition, Chan et al.~\cite{tissec:chan2011private} extend the binary tree method to handle unlimited streams.  There is also an \revision{on-line}{online} consistency method: for each time slot, if its estimation is no bigger than the previous one, output the previous one; otherwise, increment the previous one by $1$.  We note that this method handles the binary setting, and thus focuses on ensuring each value is an integer, while our method works on the more general non-binary setting and minimizes the overall noise error.
% On another direction following Dwork et al.~\cite{stoc:dwork2010differential}, Chan et al.~\cite{tissec:chan2011private} extend the binary tree method to handle unlimited streams.  There is also an on-line consistency method proposed in Chan et al.~\cite{tissec:chan2011private}: for each time slot, if its estimation is no bigger than the previous one, output the previous one; otherwise, increment the previous one by $1$.  We note that this method handles the binary setting, and thus focuses on ensuring each value is an integer, while our method works on the more general non-binary setting and minimizes the overall noise error.

In a follow-up work by Dwork et al.~\cite{asiacrypt:dwork2015pure}, the authors further assumed the number of $1$'s in the stream is small, and proposed an online partition algorithm to improve over the previous bound.
Chen et al.~\cite{ccs:chen2017pegasus} used the similar idea but worked in a different setting (i.e., each value can be any number instead of a binary number).  The method partitions the stream into a series of intervals so that values inside each interval are ``stable'', and publishes the median of each interval.  
While Chen et al.~\cite{ccs:chen2017pegasus} works in the non-binary setting, similar to our paper and \pak's
% ~\cite{ndss:perrier2019private}
, it makes two assumptions that lead to a quite different design.  First, as mentioned above, Chen et al.~\cite{ccs:chen2017pegasus} made the stability assumption so that partitioning the stream gives better utility; second, it works in a different scenario (i.e., each value is composed of multiple users, and each user contribute at most $1$) and the sensitivity is $1$, and therefore, the threshold for truncation is not needed. 
% Thus, the sensitivity is only $1$ (while in our setting, the original sensitivity is the maximal possible value $B$, and we first find a threshold $\theta<B$).

% Instead, our solution does not assume the data to have any seasonal patterns.

% \mypara{Difference from \pgs}
% Although the design of \method is similar to \pgs, they are indeed different in several ways. First, the definitions of the privacy are different. \pgs addresses event-privacy where any event is denoted by either 0 or 1. Therefore \pgs has a sensitivity of $1$ while that of \method is $B$. With that, \pgs does not require a threshold finder.  Moreover, \pgs relies on the stability of the underlying data and uses a grouper to group close data together; while \method does not make this assumption and uses the hierarchical method to output the whole data stream while optimizing the accuracy of range queries.

\mypara{User-level DP}
\revision{}{Compared to Event-level DP which protects against the change of one single event, the user-level definition models the change of the whole data possessed by the user. } \revision{There is also work that focuses on providing the notion of user-level DP.}{}  Because the user-level DP is more challenging, proposals under this setting rely more on the auxiliary information.  In particular, Fan et al.~\cite{tkde:fan2013adaptive} proposed to release perturbed statistics at sampled timestamps and uses the Kalman filter to predict the non-sampled values and correct the noisy sampled values.  It takes advantage of the seasonal patterns of the underlying data. 
% Another line of work focuses on privately publishing stream of counters and boosting utility by taking advantage of the seasonal patterns of the underlying data.  
Another direction is the offline setting, where the server has the global view of all the values first, and then releases a streaming model satisfying DP.  In this setting, Acs and Castelluccia~\cite{kdd:acs2014case} propose an algorithm based on Discrete Fourier Transform (DFT). Rastogi and Nath~\cite{sigmod:rastogi2010differentially} further incorporate sampling, clustering, and smoothing into the process.
% This and Chen et al.~\cite{ccs:chen2017pegasus} both take advantage of the seasonal patterns of the underlying data.  

\mypara{$w$-event-level DP}
To balance the privacy loss and utility loss between user-level and event-level privacy models, relaxed privacy notions are proposed.  
Bolot et al.~\cite{icdt:bolot2013private} extended the binary tree mechanism on releasing perturbed answers on sliding window sum queries over infinite binary streams with a fixed window size and using a more relaxed privacy concept called decayed privacy.   Kellaris et al.~\cite{pvldb:kellaris2014differentially} propose $w$-event DP and two new privacy budget allocation schemes (i.e., split $\epsilon$ into individual events) to achieve it.  More recently, Wang et al.~\cite{tdsc:wang2016real} work explicitly for spatiotemporal traces, and improve the schemes of Kellaris et al. by adaptively allocating privacy budget based on the similarity of the data sources; Fioretto and Hentenryck~\cite{jair:fioretto2019optstream} improve the schemes by
% Kellaris et al.~\cite{pvldb:kellaris2014differentially} 
using a similar approach of Fan et al.~\cite{tkde:fan2013adaptive}: sampling representative data points to add noise, and smoothing to reconstruct the other points.

% In the local DP setting, this streaming setting is already in use in industrial deployments~\cite{apple-ldp,ccs:ErlingssonPK14,nips:DingKY17}.  They all work in the event-level LDP, but assume  independence along the time~\cite{apple-ldp}.  Better approaches use two-layer noise~\cite{ccs:ErlingssonPK14} and memorization~\cite{nips:DingKY17}; but they do not provide formal guarantees.  On the other hand, the formal approach~\cite{nips:JosephRUW18,soda:ErlingssonFMRTT18} suffers from limited utility.

In our work, we follow the event-level privacy model and dealing with a more extended setting where data points are from a bounded range instead of the binary domain, and publish the stream in an \revision{on-line}{online} manner.  Moreover, we do not rely on any pattern to exist in the data; and we propose methods for both DP and LDP.

\subsection{Dealing with Streams in Local DP}
In the local DP setting, most existing work focus on estimating frequencies of values in the categorical domain~\cite{ccs:ErlingssonPK14,stoc:BassilyS15,uss:WangBLJ17,tiot:YeB18,ndss:wang2020consistent}.  These techniques can be applied to other applications such as heavy hitter identification~\cite{nips:BassilyNST17,tdsc:wang2019locally,icde:WangXYHSSY18} frequent itemset mining~\cite{ccs:qin2016heavy,sp:WangLJ18}, multi-dimension data estimation~\cite{sigmod:WangDZHHLJ19,ccs:ZhangWLHC18,sigmod:CormodeJKLSW18,pvldb:CormodeKS19,pvldb:yang2020answering,pvldb:xu2020collecting}.  In the numerical/ordinal setting, previous work~\cite{focs:DuchiJW13,icde:WangXYZ19} mostly focused on estimating mean.  Recently, Li et al.~\cite{sigmod:li2019estimating} proposed the square wave mechanism for the more general task of estimating the density.  

% we are not aware of any other work in the exactly same setting as ours.  
%
There are two methods~\cite{soda:ErlingssonFMRTT18,nips:JosephRUW18} that deal with streaming data in the user-level LDP, and their data models are different from ours.  In particular, Erlingsson et al.~\cite{soda:ErlingssonFMRTT18} assume the users' values are integers.  Each user's value can change at most a few times, and each change can be either $+1$ or $-1$.  The authors proposed a hierarchical method (similar to Phase II of \pak) to estimate the average change over time and thus get the average value by accumulating the changes.
Joseph et al.~\cite{nips:JosephRUW18} assume at each time, the users' values are a sequence of bits drawn from several Bernoulli distributions, and the goal is to estimate the average of all the bits held by all users (or the average of the Bernoulli parameters).  The authors proposed a method to efficiently control each user's contribution when the server's guess (previous estimation) is accurate, and thus save privacy budget.
The authors extend the method to a categorical setting and use it to find the most frequent value held by the users.
% It satisfies user-level LDP but requires each users' value only changes a small constant times.  
% In the LDP setting, other than the basic problem of estimating the frequency, several problems are investigated.  

There is also a parallel work~\cite{infocom:wang2020towards} that works on $w$-event, metric-based LDP.  The proposed protocol assumes there is a pattern in each user's streaming data, and lets each user sample the turning point to report.

\section{Conclusion and Discussion}
\label{sec:conc}
We have presented a privacy-preserving algorithm \method to continually output a stream of observations.  \method first finds a threshold to truncate the stream using the exponential mechanism, considering both noise and bias.  
Then \method runs an \revision{on-line}{online} hierarchical structure and adds noise to the stream to satisfy differential privacy (DP).  Finally, the noisy data are smoothed to improve utility.  We also design a mechanism \methodldp that satisfies the local version of DP.  Our mechanisms can be applied to real-world applications where continuous monitoring and reporting of statistics, e.g., smart meter data and taxi fares, are required.  The design of \method and \methodldp are flexible and have the potential to be extended to incorporate more properties of the data.  We list some as follows.

% \subsection{Potential Extensions of \method}
% We discuss some potential extensions of \method.

\mypara{Shorten the Holdout of the Stream}
We follow the setting of \pak~\cite{ndss:perrier2019private} and use the first $m$ values to output the threshold $\theta$.  
If we want to start outputting the stream sooner, we can use our Threshold optimizer with only fewer observations to find a rough threshold.  During the process of outputting the stream, we can use sequential composition (in \autoref{subsec:composition}) to fine-tune the threshold.

\mypara{Update $\theta$}
We follow the setting of \pak and assume the distribution stays the same.  If the distribution changes, we can have the Threshold optimizer run multiple times (using either sequential composition to update $\theta$ while simultaneously outputting the stream, or the parallel composition theorem to block some values to update $\theta$).

\mypara{Utilizing Patterns of the Data}
If there is further information, such that the data changes slowly (e.g., the current value and the next one differ in only a small amount), or the data changes regularly (e.g., if the values show some Diurnal patterns), are given, we can potentially utilize that to improve the performance of our method as well.  
% the neighboring data does not change too much, we can design the lower levels of the hierarchy in a way similar to the grouper and smoother in \pgs, so that we utilize the stability of the data to improve the utility.

\begin{comment}
\mypara{Support Other Differential Privacy Definitions}
Our method can be extended to provide privacy for multiple events instead of single events as is done in this paper.  
Our definition also captures the more general $w$-event privacy~\cite{w-event}.  
\subsection{Using it for LDP}
We note that while Definition~\ref{def:dp} is stated for the DP notion, we have an LDP variant for it.  In general, the notion of LDP differs from DP in that users take full control of the data and do not need to trust the server in LDP.  
\begin{definition}[Event-level Local Differential Privacy]
\label{def:ldp}
For any two sequences $V = \langle v_1, v_2, \ldots, v_n \rangle$, and $V' = \langle v'_1, v'_2, \ldots, v'_n \rangle$ from the same user, where $v_i=v'_i$ for all $i\in [n]$ except one index.  For any possible output set $O$,
\begin{equation*}
 \Pr{\AA(V) \in O} \leq e^{\epsilon}\, \Pr{\AA(V') \in O} + \delta
\end{equation*}
\end{definition}
The difference from Definition~\ref{def:dp} is only that we require the stream is output by a single user.  This enables publishing streams from the same user, e.g., publishing of smart metering while the readings are all coming from the same household.  We note as Definition~\ref{def:ldp} is a special case of Definition~\ref{def:dp}, \method can also be used for Definition~\ref{def:ldp}.  
\end{comment}

\section*{Acknowledgements}
This project was supported by NSF 1931443, 2047476, a Bilsland Dissertation Fellowship, a Packard fellowship, and gift from Cisco and Microsoft.
The authors are thankful to the anonymous reviewers for their supportive reviews.
{
% please add reference to ref_new; ref is a centralized file shared by multiple projects
% \bibliographystyle{ACM-Reference-Format}
\bibliographystyle{abbrv}
\bibliography{ref,ref_new}
}

%%%%%%%%%%%%%%%%%%%%%%%%%%%%%%%%
%%%%%%%%%% appendix %%%%%%%%%%%%
%%%%%%%%%%%%%%%%%%%%%%%%%%%%%%%%
% \newpage
\appendix
% \appendixpage
% \input{app_proof.tex}
\section{Supplementary Mechanisms of DP}
% \subsection{Mechanisms of Differential Privacy}
\label{app:dp_primitive}
We review the method of smooth sensitivity that \pak~\cite{ndss:perrier2019private} use for estimating the percentile.

\mypara{Smooth Sensitivity}
% In Laplace Mechanism, the global sensitivity $\mathsf{GS}_f$ is big because any possible pair of inputs are considered.  A kind of sensitivity that is defined for a particular input $V$ is called local sensitivity.
Rather than the global sensitivity that considers any pair of neighboring sequences, the local sensitivity fixes one dataset $V$ and only considers all possible neighboring sequence $V'$ around $V$.
\[
\mathsf{LS}_{V} (f) =  \max_{V': V'\simeq V} ||f(V) - f(V')||_1.
\]

The advantage of using local sensitivity is that we only need to consider neighbors of $V$ which could result in lower sensitivity of the function $f$, and consequently lower noise added to the true answer $f$. Unfortunately, replacing the global sensitivity with local sensitivity directly (e.g., in the Laplace mechanism) violates DP. This is handled by using smooth sensitivity~\cite{stoc:nissim07} instead. 

For $b > 0$, the $b$-smooth sensitivity of $f$ at $V \in \mathcal{V}$, denoted by $\mathsf{SS}_{V, b}(f)$,  is defined as
\[
\mathsf{SS}_{V, b}(f) = \max_{V'} \left\{ \mathsf{LS}_{V'}(f) \cdot e^{-b \cdot d(V, V')} \right\},
\]
where $d(V, V')$ denotes the hamming distance between $V$ and $V'$.
The method of smooth sensitivity is given below:
$$
\AA(V) =f(V) + \frac{\mathsf{SS}_{V, b}}{a}\cdot Z,
$$
where $Z$ is a random variable from some specified distribution.  To obtain $(\epsilon, 0)$-DP, $Z$ is drawn from the Cauchy distribution with density $\propto \frac{1}{1 + |z|^\gamma}$ for $\gamma>1$.  But to use an exponentially decaying distribution (which gives good accuracy), such as standard Laplace or Gaussian distribution, one can only obtain $(\epsilon, \delta)$-DP.  For both, we set $b \le \frac{\epsilon}{-2\log(\delta)}$ as the smoothing parameter.  If we use the Laplace distribution with scale 1, $a = \frac{\epsilon}{2}$.  When the noise is standard Gaussian, then $a = \frac{\epsilon}{\sqrt{-\ln \delta}}$~\cite{stoc:nissim07}. 

% \mypara{Permute-and-Flip}
% Instead of sampling the desired value based on \autoref{eq:exp}, Permute-and-Flip~\cite{nips20:mckenna2020permute} first ``normalizes'' each score $q$ by subtracting the maximal possible score denoted by $q*$.  It then considers each score in a random order.  For each score, the method draws a random number in $x\in [0, 1)$, and compares it with the exponential of the normalized score, i.e., $y\exp{\left(\frac{\epsilon }{2\,\mathsf{GS}_q}(q - q*)\right)}$.  The first score index for which the random number $x\le y$ is returned.

% The authors reported an improvement of near $2\times$ over EM~\cite{focs:mcsherry07} on the problem of finding median and mode.  We show that using Permute-and-Flip, we can improve the results of the threshold optimizer by xx\tw{todo}.  The result is given in xx.
% Note that the main algorithm still uses EM because EM was proposed over a decade ago and is widely accepted.  We want to avoid giving readers the impression that we  improve over existing work because we use new techniques.  

\section{More Details about \pak}
\label{app:imp_ss}
\pak computes the smooth sensitivity of the empirical $p$-quantile, i.e., $\hat{x}_{p}$, as
\begin{align*}
 &\mathsf{SS}_{V, b}(\hat{x}_{p}) = \max_{k = 0,1,\ldots, m + 1}\\
&   \left\{e^{-bk} \cdot \max_{t = 0,1, \ldots, k+1} [ V^s (P+t) - V^s (P+t-k-1)]\right\}.
\end{align*}

% \begin{align*}
% \mathsf{SS}_{V, b}(\hat{x}_{p}) &=  \max_{k = 0,1,\ldots, m + 1} \{e^{-bk}\mathsf{LS}_{V'}(\hat{x}_{p}) : d(V,V') \leq k\},
% \end{align*}
% where 
% \[
% \mathsf{LS}_{V'} (\hat{x}_{p}) =  \max_{t = 0,1, \ldots, k+1}  V^s (P+t) - V^s (P+t-k-1).
% \]
Here, $V^s$ is the sorted string of the first $m$ values of $V$ in ascending order, where $V^s(i)=0$ if $i < 1$ and $B$ if $i > m$. And $P$ is the rank of $\hat{x}_{p}$. 
After computing the smooth sensitivity, Nissim et al.\cite{stoc:nissim07} sets the threshold $\theta$ as 
\[ 
\theta = \hat{x}_{p} + \frac{\mathsf{SS}_{V, b}(\hat{x}_{p})}{a} \cdot Z,
\]
where $Z$ is a random variable from some specified distribution in order to satisfy DP.
% \says{Ninghui}{don't understand the following sentence.  What does "arbitrary $\beta$" mean?}

PAK~\cite{ndss:perrier2019private} propose to bound $\Pr{\theta < x_p}$ to be arbitrarily small (by an arbitrary $\beta$) and thus uses 
\begin{align*}
 \theta & = \hat{x}_{p} + \frac{\kappa \mathsf{SS}_{V, b}(\hat{x}_{p})}{a} \cdot ( Z + G_{\text{ns}}^{-1}(1-\beta) ),
\end{align*}
where $\kappa$ is a positive real number 
% \begin{align*}
$\left(1-\frac{(e^{b}-1)G_{\text{ns}}^{-1}(1-\beta_{\text{lt}})}{a} \right)^{-1}$
% \end{align*}
and $G_{\text{ns}}$ denotes the CDF of the distribution of $Z$.

The threshold $\theta$ released via the above mechanism is differentially private since $\kappa \mathsf{SS}_{V, b}(\hat{x}_{p})$ is a smooth upper bound of $\hat{x}_{p}$ and $\kappa$ only depends on public parameters. 

In their evaluation, the authors aim to get the $99.5$-percentile and set $p=99.575, \beta=0.3\cdot 0.02$, and $\delta=1/n^2$.
% \says{zhikun}{is it necessary to put the experimental setting here?}
% \says{zhikun}{this section only introduce the algorithm for determining $\theta$ in \pak. shall we first provide an overview of \pak, and then introduce both threshold choosing and perturbing algorithms?}

\section{Consistency Algorithm}
\label{app:consist}

\mypara{Off-line Consistency~\cite{pvldb:HayRMS10}}
We use $x$ to denote a node on the hierarchy $H$, and let $\ell(x)$ to be the height of $x$ (the height of a leaf is $1$; and root is of height $h$).  We also denote $\mbox{prt}(x), \mbox{chd}(x)$, and $\mbox{sbl}(x)$ to denote the children, parent, and siblings of $x$, respectively.  
% We set $\mbox{prt}(x)=x$ and $\mbox{sbl}(x)=x$ if $x$ is the root, and $\mbox{chd}(x)=x$ if $x$ is a leaf node.  
We use $H(x)$ to denote the value corresponding to node $x$ in the hierarchy.  The first step updates the values in $H$ from bottom to top.  
% For $\ell$ from $1$ up to $h$, and for all node $x$ in level $\ell$, we update $H(x)$ by:
The leaf nodes remain the same; and for each height-$\ell$ node $x$, where $\ell$ iterates from $2$ up to $h$, we update it 
\vspace{-0.028cm}
\begin{align}
H(x) \gets \frac{b^{l} - b^{l-1}}{ b^{l} - 1 } H(x) + \frac{b^{l-1} - 1}{ b^{l} - 1 } \sum_{y \in \mbox{chd}(x)} H(y).
\label{eq:bottom_up}
% \nonumber
\end{align}

We then update the values again from top to bottom.
This time the value on root remains the same, and for each height-$\ell$ node $x$, where $\ell$ iterates from $h-1$ down to $1$, we update it
\begin{align}
% \resizebox{\hsize}{!}{$
H(x) \gets \frac{b-1}{b}H(x) + \frac{1}{b} \left(H(\mbox{prt}(x)) - \sum_{y \in \mbox{sbl}(x)} H(y)\right).
\label{eq:top_down}
% \nonumber
% $}
\end{align}
% \begin{align}
% h'[x] = \left\{
% \begin{array}{lr}
% {h}[x],  \mbox{if $x$ is a leaf node}  \\
% \frac{b^{l} - b^{l-1}}{ b^{l} - 1 } \tilde{h}[x] + \frac{b^{l-1} - 1}{ b^{l} - 1 } \sum_{y \in succ(x)} z[y],  \mbox{otherwise} \\
% \end{array}\label{eq:bottom_up}
% \right.
% \end{align}

% \mypara{Top Down}
% Let $u$ be $x$'s parent:
% \begin{eqnarray*}
% \bar{h}[x] & = &
% \left\{
% \begin{array}{l}
% z[x],~~\mbox{if $x$ is the root}\\
% z[x] + \frac{1}{b} (\bar{h}[u] - \sum_{w \in succ(u)} z[w]),~~\mbox{o.w.}
% \end{array}
% \right.
% \end{eqnarray*}

\mypara{An Online Algorithm}
We decompose the noisy values in $H$ into two parts: the true values and the pure noises,  denoted as $T$ and $N$, respectively.  $N$ is the independent noise from the Laplace mechanism (described in \autoref{subsec:dp_primitive}).  $T$ is defined by induction: for a leaf $x$, $T(x)$ corresponds to one true value $v$, and for a node $x$ in a higher level, $T(x)=\sum_{y \in \mbox{chd}(x)} T(y)$.
The true values from $T$ are consistent naturally.  Thus if $N$ is consistent, $H$ is also consistent.  To see it, consider any internal node $x$, 
\vspace{-0.05cm}
\begin{align*}
  &H(x) = T(x)+N(x) \\
= &\sum_{y\in\mbox{chd}(x)}T(y)+ \sum_{y\in\mbox{chd}(x)}N(y)\\
= &\sum_{y\in\mbox{chd}(x)}(T(y) + N(y)) = \sum_{y\in\mbox{chd}(x)}H(y).
\end{align*}

In the online consistency algorithm, we internally generate the noise hierarchy $N$ and run the steps given in \autoref{eq:bottom_up} and~\autoref{eq:top_down} to make $N$ consistent.
% More generally, for each $i$ of the $h$ counters in the perturber, $b^{h - i}$ noise values are generated, the $h$ counters then run the consistency enforcement algorithm on the noise samples (this introduces a $O(r)$ space complexity rather than $O(\log r)$ without it).  
% After this pre-processing step, the other counters are not used anymore; only the $1$-st counter will take the input stream, add the consistent noise, and output the stream.  This is because the results from counters in higher layers are already consistent with those from the $1$-st counter, thus it suffices to only output the most fine-grained result.
After this pre-processing step, we can ignore the higher-layers of the hierarchy, and use only the leaf nodes which add consistent noise to each individual value.  This is because the results from the higher-layers are already consistent with those from the leaves, thus it suffices to only output the most fine-grained result.
% Figure~\ref{fig:consistency} gives an example using a simple binary tree.

% \section{Supplementary Proofs}
% \label{app:proof}

% \mypara{Proof of \autoref{lemma:variance_bias}}
% \begin{proof}
% \begin{align*}
%     \EV{(\tilde{v}-v)^2} 
%     =& \EV{\tilde{v}^2-2v\tilde{v}+v^2}\\
%     =& \EV{\tilde{v}^2}-2v\EV{\tilde{v}}+v^2\\
%     =& \EV{\tilde{v}^2} - \EV{\tilde{v}}^2 + \EV{\tilde{v}}^2-2v\EV{\tilde{v}}+v^2\\
%     =& \left(\EV{\tilde{v}^2} - \EV{\tilde{v}}^2\right) + \left(\EV{\tilde{v}} - v\right)^2\\
%     =& \Var{\tilde{v}} + \Bias{\tilde{v}}^2
% \end{align*}
% \end{proof}

% \mypara{Proof of \autoref{thm:online_consist}}
The online consistency algorithm satisfies DP as long as it gives identical output distribution as the offline algorithm~\cite{pvldb:HayRMS10}.  Now we prove this fact.  We first restate the theorem (\autoref{thm:online_consist}):
\begin{theorem}
% \label{thm:online_consist}
The online consistency algorithm gives identical results as the off-line consistency algorithm.
\end{theorem}

\begin{proof}
We first examine the bottom-up update step.  According to \autoref{eq:bottom_up}, the updated $N(x)$ equals to
\[
\frac{b^{l} - b^{l-1}}{ b^{l} - 1 } N(x) + \frac{b^{l-1} - 1}{ b^{l} - 1 } \sum_{y \in \mbox{chd}(x)} N(y).
\]
Adding $T(x)$ to it, we have the updated $H(x)$ equals to 
% \resizebox{\linewidth}{!}{
%   \begin{minipage}{\linewidth}
\begin{align}
&\frac{b^{l} - b^{l-1}}{ b^{l} - 1 } N(x) + \frac{b^{l-1} - 1}{ b^{l} - 1 } \sum_{y \in \mbox{chd}(x)} N(y) + T(x) \nonumber\\
=&\frac{b^{l} - b^{l-1}}{ b^{l} - 1 } (N(x)+T(x)) + \frac{b^{l-1} - 1}{ b^{l} - 1 } \left(\sum_{y \in \mbox{chd}(x)} N(y)+T(x)\right)\nonumber\\
=&\frac{b^{l} - b^{l-1}}{ b^{l} - 1 } (N(x)+T(x)) + \frac{b^{l-1} - 1}{ b^{l} - 1 } \sum_{y \in \mbox{chd}(x)} (N(y)+T(y))\label{eq:proof_bottom_up}\\
=&\frac{b^{l} - b^{l-1}}{ b^{l} - 1 } H(x) + \frac{b^{l-1} - 1}{ b^{l} - 1 } \sum_{y \in \mbox{chd}(x)} H(y).\nonumber
\end{align}
% \end{minipage}
% }
\autoref{eq:proof_bottom_up} is because of the consistency of $T$ so that $T(x)=\sum_{y \in \mbox{chd}(x)} T(y)$.
This gives the same result as if we run the off-line consistency algorithm.  Similarly, during the top-down update step (in \autoref{eq:top_down}), we have
the updated $N(x)$ equals to
\[
\frac{b-1}{b}N(x) + \frac{1}{b} \left(N(\mbox{prt}(x)) - \sum_{y \in \mbox{sbl}(x)} N(y)\right).
\]
Adding $T(x)$ to it, we have the updated $H(x)$ equals to 
% \resizebox{\linewidth}{!}{
%   \begin{minipage}{\linewidth}
\begin{align}
&\frac{b-1}{b}N(x) +\frac{1}{b} \left(N(\mbox{prt}(x)) - \sum_{y \in \mbox{sbl}(x)} N(y)\right) + T(x)\nonumber\\
=&\frac{b-1}{b}(N(x)+T(x)) 
+\frac{1}{b} \left(N(\mbox{prt}(x)) - \sum_{y \in \mbox{sbl}(x)} N(y) + T(x)\right)\nonumber\\
% =&\frac{b-1}{b}H(x) \\
% +&\frac{1}{b} \left(N(\mbox{prt}(x)) - \sum_{y \in \mbox{sbl}(x)} N(y) + T(\mbox{prt}(x))-\sum_{y \in \mbox{sbl}(x)} T(y)\right)\\
% =&\frac{b-1}{b}H(x) +\frac{1}{b} \left(T(x)+N(\mbox{prt}(x)) - \sum_{y \in \mbox{sbl}(x)} N(y)\right)\nonumber\\
% &\mbox{because } T(x) = T(\mbox{prt}(x))-\sum_{y \in \mbox{sbl}(x)} T(y)\\
=&\frac{b-1}{b}H(x) + \frac{1}{b} \left(H(\mbox{prt}(x)) - \sum_{y \in \mbox{sbl}(x)} H(y)\right).\label{eq:proof_top_down}
\end{align}
% \end{minipage}
% }
which is the same result as if we run the off-line consistency algorithm.  
\autoref{eq:proof_top_down}  also holds because of the consistency of $T$ so that $T(x)=T(\mbox{prt}(x))-\sum_{y \in \mbox{sbl}(x)} T(y)$.
\end{proof}

\section{Smoothing Methods}
\label{app:smooth}

Denote $u_1, u_2, \ldots$ as the noisy estimates given by the leaves of the perturber (or the $(s+1)$-th levels of the original hierarchy).  Each $u_i$ is the noisy sum of $b^s$ values.  Let $u_0=\frac{1}{2}b^s\theta$ (initially there are no estimations from the hierarchy; we thus use half of the threshold as mean).  
% Each of $u_i$ stands for the sum over a period of $b^s$ observations.  
The smoother will take the sequence of $u$ and output the final result $\tilde{v}_i$ for each input value.  
Let $t=\lceil i/b^s\rceil$, we consider several functions:

\begin{enumerate}[leftmargin=*]
    \item Recent smoother: $\tilde{v}_i=u_{t}/b^s$.  It takes the mean of the most recent output from the perturber.
    \item Mean smoother: $\tilde{v}_i=\frac{1}{b^s t}\sum_{j=0}^{t}u_{j}$.  It takes the mean of the output from the perturber up until the moment.
    \item Median smoother: $\tilde{v}_i=\mbox{median}(u_1,\ldots,u_{t})$. Similar to the mean smoother, the median smoother takes the median of the output from the perturber up until the moment.
    \item Moving average smoother: $\tilde{v}_i=\frac{1}{b^s w}\sum_{j=t+1-w}^{t}u_{j}$.  Similar to the mean smoother, it takes the mean over the most recent $w$ outputs from the perturber.  When $t+1<w$, we use the average of the first $t+1$ values of $u$ divided by $b^s$ as $\tilde{v_i}$.
    \item Exponential smoother: $\tilde{v}_i=\frac{u_0}{b^s}$ if $t=0$, and $\tilde{v}_i = \alpha \frac{u_t}{b^s}+(1-\alpha)\tilde{v}_{i-b^s}$ if $t>0$, where $0\le \alpha \le 1$ is the smoothing parameter.  The exponential smoother put more weight on the more recent values from the hierarchy.
\end{enumerate}

\section{Mechanisms of Local Differential Privacy}
\label{app:ldp_primitive}
% To satisfy LDP, the user can transform $v$ using one-hot encoding, and then use the Laplace mechanism.  But this incurs large communication cost when the domain size $d=|\Domain|$ is large.  Moreover, it gives worse utility than perturbation-based LDP mechanisms~\cite{uss:WangBLJ17}, which we review as follows.  
In this subsection, we review the primitives proposed for LDP.  We use $v$ to denote the user's private value, and $y$ as the user's report that satisfies LDP.
% , and $\Domain$ to denote the domain of $v$.
In this section, following the notations in the LDP literature, we use $p$ and $q$ to denote probabilities.

% \subsection{LDP Mechanism for Density Estimation}
\subsection{Square Wave Mechanism for Density Estimation}
\label{app:ldp_fo}
% One basic primitive of LDP is to estimate the histogram (or density if the domain is continuous).  It is like what the Laplace mechanism solves in the centralized DP setting.

% \input{primitive_grr.tex}
% \input{primitive_olh.tex}

% \mypara{Square Wave}
Li et al.~\cite{sigmod:li2019estimating} propose an LDP method that can give the full density estimation.  
% Given the density, any statistics can be answered.
The intuition behind this approach is to try to increase the probability that a noisy reported value carries meaningful information about the input.  
% In \grr, one reports a value in $\Domain$.  
 %A reported value of $v$ indicates that the input value is more likely to be $v$ than another value in $\Domain$.   
Intuitively, if the reported value is the true value, then the report is a ``useful signal'', as it conveys the extract correct information about the true input.  If the reported value is not the true value, the report is in some sense noise that needs to be removed.  
% In \grr, the probability that a useful signal is generated is $p=\frac{e^\epsilon}{e^\epsilon+d-1}$, where $d=|\Domain|$ is the size of the domain.  When $d$ is large, $p$ is small, and \grr performs poorly.  
% The essence of \olh and other protocols is that one reports a randomly selected set of values, where one's true value has a higher probability of being selected than other values.  In some sense, each ``useful signal'' is less sharp, since it is a set of values, but there is a much higher probability that a useful signal is transmitted. 
Exploiting the ordinal nature of the domain, a report that is different from but close to the true value $v$ also carries useful information about the distribution.  
Therefore, given input $v$, we can report values closer to $v$ 
 %(i.e., $\rdv$ such that $|\rdv-v|\le b$) 
with a higher probability than values that are farther away from $v$.  
The reporting probability looks like a squared wave, so the authors call the method Square Wave method (SW for short).

Without loss of generality, we assume values are in the domain of $[0,1]$.  
To handle an arbitrary range $[\ell, r]$, each user first transforms $v$ into $\frac{v-\ell}{r-\ell}$ (mapping $[\ell, r]$ to $[0,1]$); and the estimated result is transformed back.  Define the ``closeness'' measure
$b = \frac{\epsilon e^\epsilon - e^\epsilon +1}{2e^\epsilon(e^\epsilon - 1 - \epsilon)}$,
% \begin{align*}
%     b = \frac{\epsilon e^\epsilon - e^\epsilon +1}{2e^\epsilon(e^\epsilon - 1 - \epsilon)} \ ,
% \end{align*}
the Square Wave mechanism SW is defined as: 
\begin{align}
\forall y \in [-b,1+b],  \;\Pr{\mbox{SW}(v)=y}  \!=\! \left\{
\begin{array}{lr}
\!p, & \mbox{if} \; |v-y|\le b  \ , \\
\!q, & \mbox{ otherwise} \ . \\
\end{array}\nonumber
% \label{eq:squarewave}
\right.
\end{align}

By maximizing the difference between $p$ and $q$ while satisfying the total probability adds up to $1$, we can derive $p=\frac{e^\epsilon}{2be^\epsilon + 1}$ and $q=\frac{1}{2b e^\epsilon + 1}$.
% from the LDP requirement $p \le e^\epsilon q$, the desire to maximize $p$, and the requirement that for any input, the probabilities for all outputs sum up to $1$, i.e., $(2b+1)p + (d-1)q = 1$.  Solving them results in 
% \begin{align*}
% p =  \frac{e^\epsilon}{e^\epsilon (2b+1) + d -1} ,\quad
% q =  \frac{1}{e^\epsilon (2b+1) + d -1}. 
% \end{align*}
% \begin{align*}
% p =  \frac{e^\epsilon}{2be^\epsilon + 1} \ ,\quad
% q =  \frac{1}{2b e^\epsilon + 1} \ . 
% \end{align*}

% For each input $v$, the probability mass distribution for the perturbed output looks like a square wave, with the high plateau region centered around  $v$. We thus call it the Square Wave (\SW) reporting mechanism. 
% Since $b$ has to be an integer, we take the floor and get $b = \left\lfloor \frac{\epsilon e^\epsilon - e^\epsilon +1}{2e^\epsilon(e^\epsilon - 1 - \epsilon)} d \right\rfloor$.
% Note that $b$ is a non-increasing function with $\epsilon$.  When $\epsilon$ goes to $\infty$, $b$ goes to $0$. When $\epsilon$ goes to $0$,  $b$ goes to $1/2$.
After receiving perturbed reports from all users, the server runs the Expectation Maximization algorithm to find an estimated density distribution that maximizes the expectation of observing the output.  Additionally, the server applies a special smoothing requirement to the Expectation Maximization algorithm to avoid overfitting.

\subsection{Candidate Methods for the Perturber}
\label{app:ldp_mean}
As we are essentially interested in estimating the sum over time, the following methods that estimate the mean within a population are useful.  We first describe two basic methods.  Then we describe a method that adaptively uses these two to get better accuracy in all cases.  Our perturber will use the final method.

\begin{figure*}[h]
	\centering
	\subfigure[Fare, MAE]{
		\includegraphics[width=0.23\textwidth]{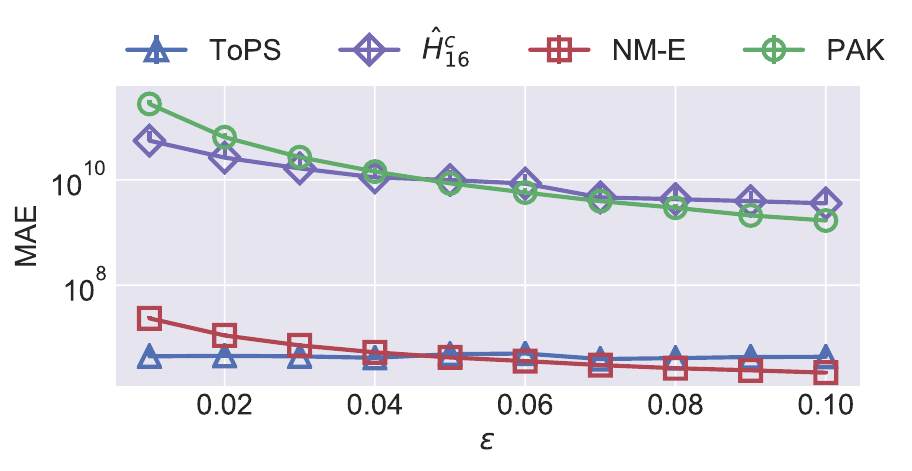}
	}
	\subfigure[DNS, MAE]{
		\includegraphics[width=0.23\textwidth]{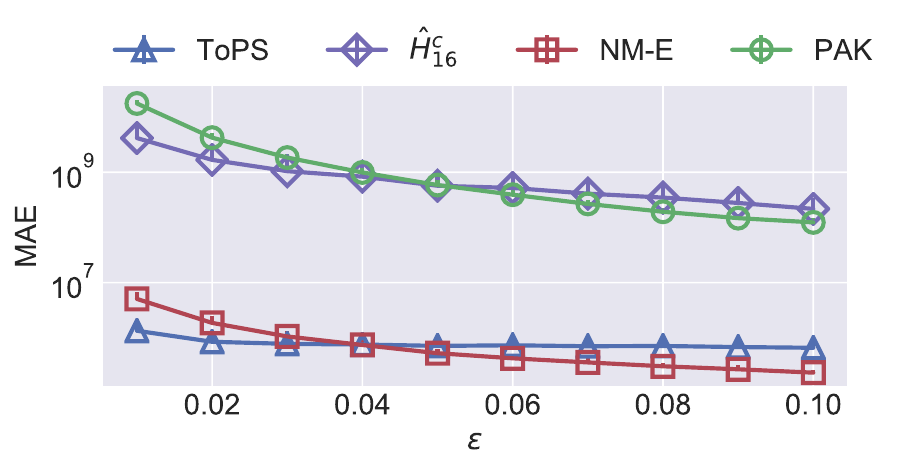}
	}
	\subfigure[Kosarak, MAE]{
		\includegraphics[width=0.23\textwidth]{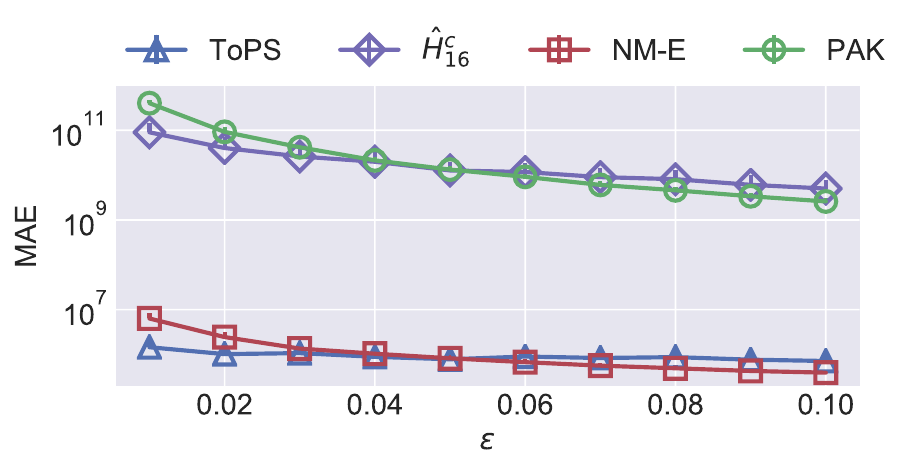}
	}
	\subfigure[POS, MAE]{
		\includegraphics[width=0.23\textwidth]{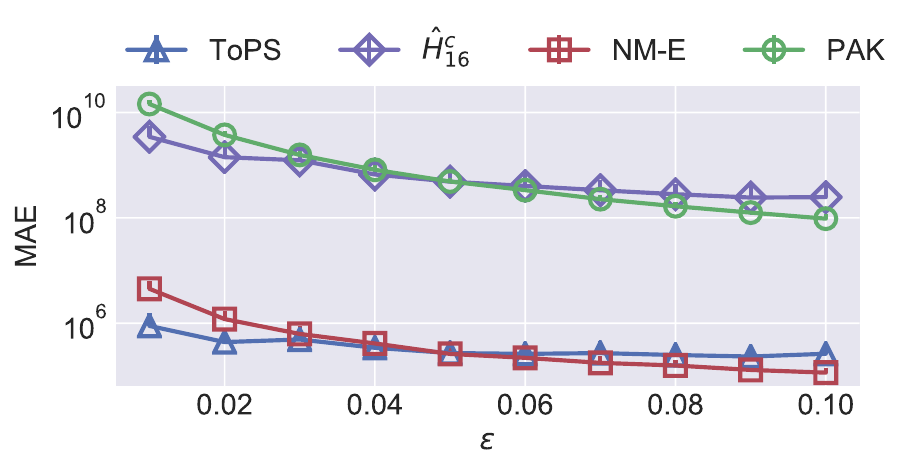}
	}
	\subfigure[Fare, MMAE]{
		\includegraphics[width=0.23\textwidth]{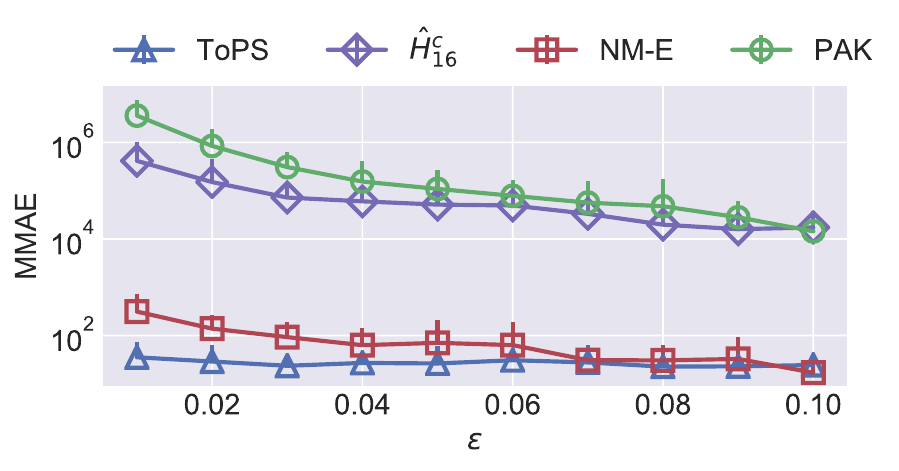}
	}
	\subfigure[DNS, MMAE]{
		\includegraphics[width=0.23\textwidth]{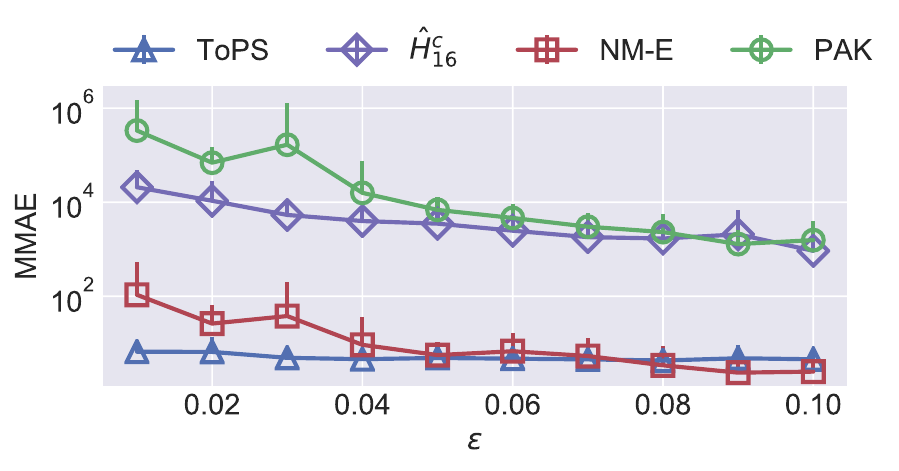}
	}
	\subfigure[Kosarak, MMAE]{
		\includegraphics[width=0.23\textwidth]{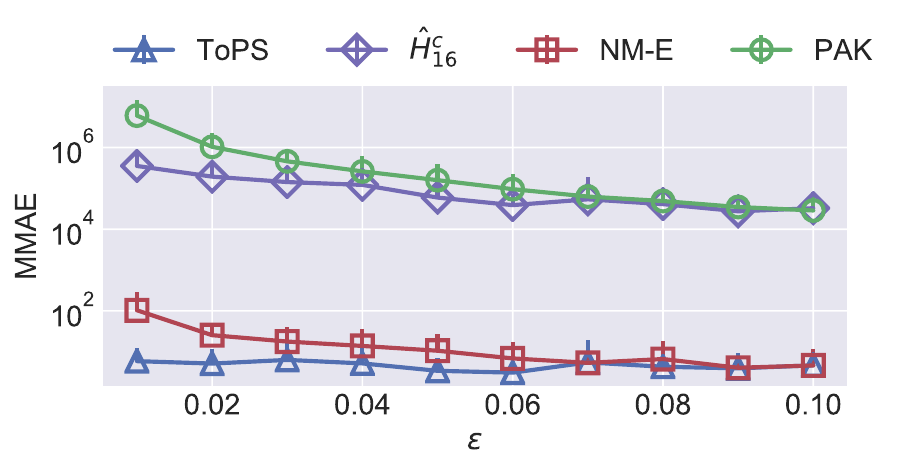}
	}
	\subfigure[POS, MMAE]{
		\includegraphics[width=0.23\textwidth]{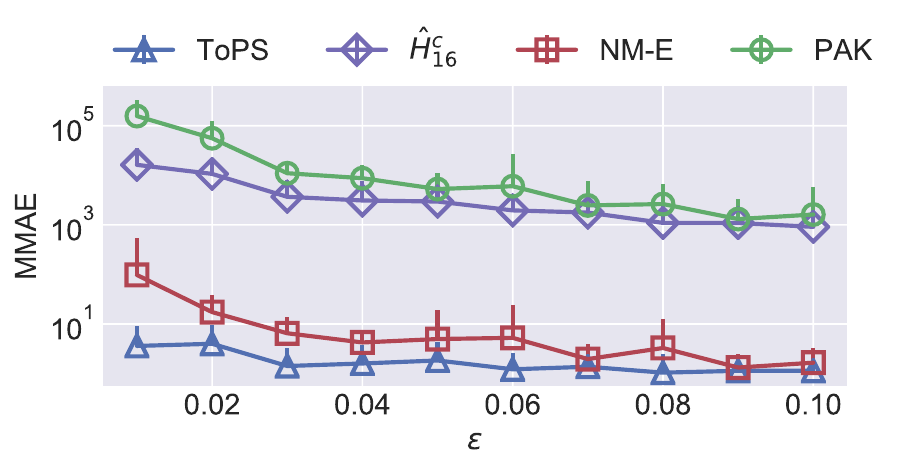}
	}
	\subfigure[Fare, MMSE]{
		\includegraphics[width=0.23\textwidth]{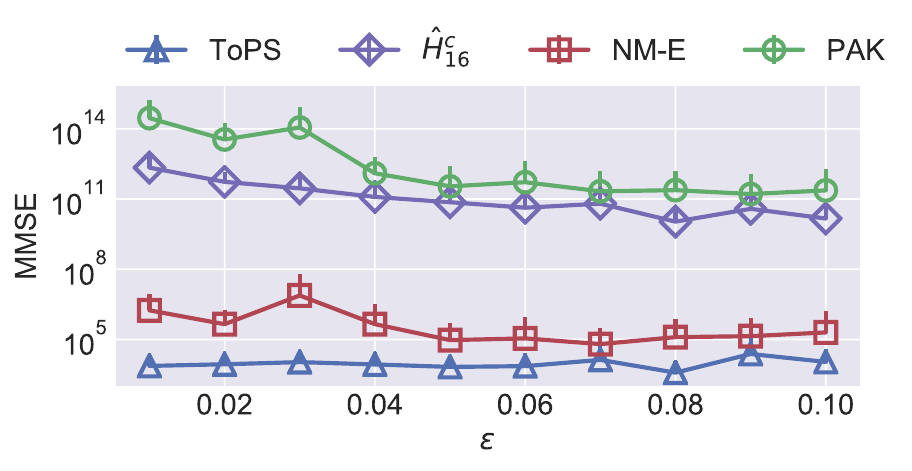}
	}
	\subfigure[DNS, MMSE]{
		\includegraphics[width=0.23\textwidth]{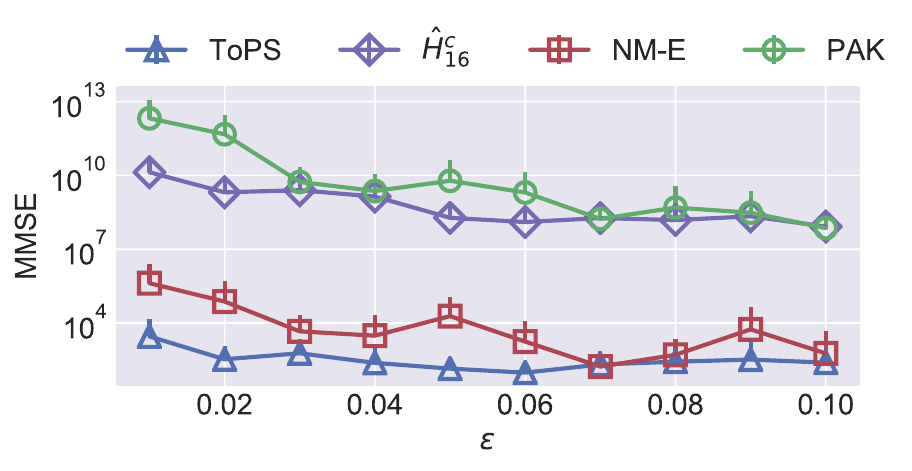}
	}
	\subfigure[Kosarak, MMSE]{
		\includegraphics[width=0.23\textwidth]{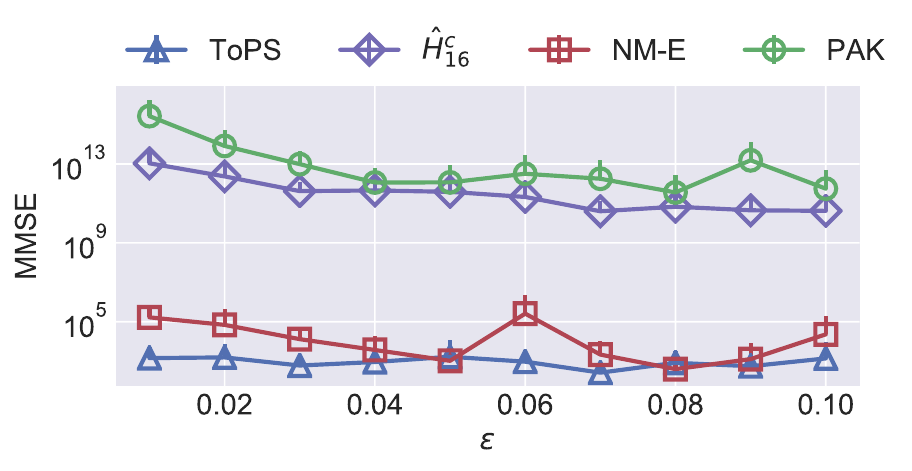}
	}
	\subfigure[POS, MMSE]{
		\includegraphics[width=0.23\textwidth]{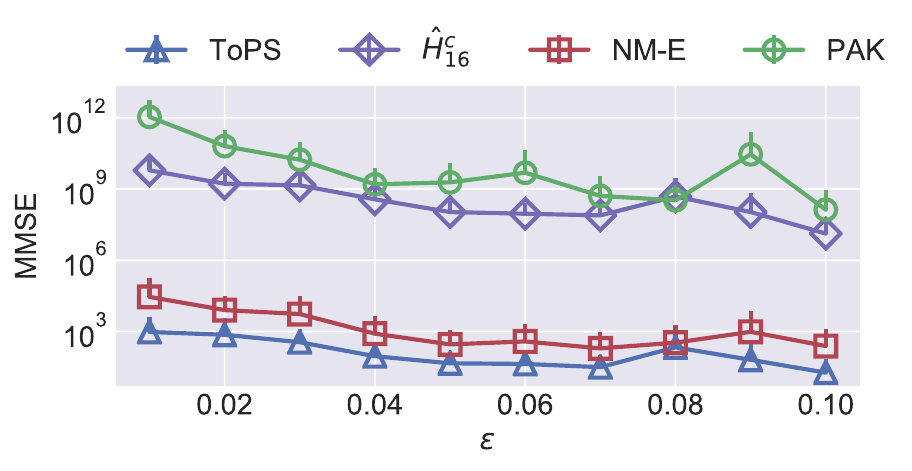}
	}
	\caption{Comparison of different methods on MAE (first row), MMAE (second row) and MMSE (third row).}
	\label{fig:query_more_vary_e_med}
\end{figure*}

\mypara{Stochastic Rounding}  
This method uses stochastic rounding to estimates the mean of a continuous/ordinal domain~\cite{jasa:DuchiJW18}.  We call it Stochastic Rounding (SR for short).  Assume the private input value $v$ is in the range of $[-1, 1]$ (otherwise, we can first projected the domain into $[-1,1]$), the main idea is to round $v$ to $v'$ so that $v'=1$ with probability $p_1=\frac{1}{2}+\frac{v}{2}$ and $v'=-1$ w/p $1 - p_1$.
This stochastic rounding step is unbiased in that $\EV{v'}=v$.
Then given a value $v' \in \{-1, 1\}$, the method runs binary random response to perturb $v'$ into $y$.  In particular, let $p=\frac{e^\epsilon}{e^\epsilon+1}$ and $q=1-p=\frac{1}{e^\epsilon+1}$, $y=v'$ w/p $p$, and $y\neq v'$ w/p $q$.  
% Since   
% \begin{align*}
%  \EV{v'} & =  (-1)\left(q+\frac{(p-q)(1-v)}{2}\right) +  q +\frac{(p-q)(1+v)}{2} \\
%         & =  (p-q)v \ .
% \end{align*}
% Let $\tilde{v} = \frac{v'}{p-q}$, we have $\EV{\tilde{v}} = v$; thus the mean of $\tilde{v}$ provides an unbiased estimate of the mean for the distribution. 
The method has variance $\left(\frac{e^\epsilon+1}{e^\epsilon-1}\right)^2-v^2$.

\mypara{Piecewise Mechanism}
Wang et al.~\cite{icde:WangXYZ19} proposed piecewise mechanism.
It is also used for mean estimation, but can get more accurate mean estimation than SR when $\epsilon> 1.29$.
In this method, the input domain is $[-1,1]$, and the output domain is $[-s, s]$, where $s=\frac{e^{\epsilon/2}+1}{e^{\epsilon/2}-1}$.  For each $v \in [-1,1]$, there is an associated range $[\ell(v),r(v)]$ where $\ell(v)=\frac{\sqrtee\cdot v - 1}{\sqrtee - 1}$ and $r(v)=\frac{\sqrtee\cdot v + 1}{\sqrtee - 1}$, such that with input $v$, a value in the range $[\ell(v),r(v)]$ will be reported with higher probability than a value outside the range.  
The high-probability range looks like a ``piece'' above the true value, so the authors call the method Piecewise Mechanism  (PM for short).
The perturbation function is defined as
\begin{align}
\forall_{y \in [-s,s]}\;\Pr{\mbox{PM}(v) = y}  = \left\{
\begin{array}{lr}
p=\frac{\sqrtee}{2}z , & \mbox{if} \; y \in [\ell(v), r(v)]  \\
q= \frac{1}{2\sqrtee}z , & \mbox{otherwise}. \\
\end{array}\nonumber
%\label{eq:pm}
\right.
\end{align}
where $z= \frac{\sqrtee - 1}{\sqrtee + 1}$.
% More precisely, we have $\ell(v)=\frac{\sqrtee\cdot v - 1}{\sqrtee - 1}$ and $r(v)=\frac{\sqrtee\cdot v + 1}{\sqrtee - 1}$.  
% The width of the range is $r(v) - \ell(v) = \frac{2}{\sqrtee - 1}$, and the center is $\frac{\ell(v)+r(v)}{2} = \frac{\sqrtee}{\sqrtee - 1}\cdot v$.  Specifically, $\PM$ works as follows:
% \begin{align*}
%     \Pr{\PM(v)=\rdv} & = \frac{\sqrtee}{2} \cdot \frac{\sqrtee - 1}{\sqrtee + 1} \mbox{ if } \rdv \in [\ell(v), r(v)],\\
%     \Pr{\PM(v)=\rdv} & = \frac{1}{2\sqrtee} \cdot \frac{\sqrtee - 1}{\sqrtee + 1} \mbox{ otherwise}.
% \end{align*}
% where $p = \frac{\sqrtee}{2} \cdot \frac{\sqrtee - 1}{\sqrtee + 1},$ and $q = \frac{1}{2\sqrtee} \cdot \frac{\sqrtee - 1}{\sqrtee + 1}$.
Compared to SR, this method has a variance of $\frac{v^2}{\sqrtee-1}+\frac{\sqrtee+3}{3(\sqrtee-1)^2}$~\cite{icde:WangXYZ19}.
% It is shown that $\tilde{v}$ is unbiased, and has better variance than \SR when $\epsilon$ is large~\cite{icde:WangXYZ19}.

\mypara{Hybrid Mechanism}
Both SR and PM incurs a variance that depends on the true value, but in the opposite direction.  In particular, when $v=\pm 1$, the variance of SR is lowest, but the variance of PM is highest. Wang et al.~\cite{icde:WangXYZ19} thus propose a method called Hybrid Mechanism (HM for short) to achieve good accuracy for any $v$.  In particular, define $\alpha=1-e^{-\epsilon/2}$, when $\epsilon>0.61$, users use PM w/p $\alpha$ and SR w/p $1-\alpha$.  When $\epsilon\le 0.61$, only SR will be called.  It is proved by Wang et al.~\cite{icde:WangXYZ19} HM gives better accuracy than SR and PM.  In particular, the worst-case variance is
\begin{align}
\Var{\tilde{v}}  = \left\{
\begin{array}{lr}
\left(\frac{e^\epsilon+1}{e^\epsilon-1}\right)^2 , & \mbox{when} \; \epsilon \le 0.61 \\
\frac{1}{\sqrtee}\left[\left(\frac{e^\epsilon+1}{e^\epsilon-1}\right)^2 + \frac{\sqrtee + 3}{3(\sqrtee-1)}\right], & \mbox{when} \; \epsilon > 0.61. \\
\end{array}
% \nonumber
\label{eq:var_hm}
\right.
\end{align}

\section{More Results}
\label{app:more_results}
Denote $Q$ as the set of the $200$ randomly generated queries, we show the results of Mean Absolute Error (MAE):
\begin{align}
    \mbox{MAE}(Q) = \frac{1}{|Q|}\sum_{(i,j)\in Q}{\left|\tilde{V}(i,j) - V(i,j)\right|}.
    \nonumber
    % \label{eq:exp_metric}
\end{align}
Moreover, we also show results for the normalized results, or the mean of range queries, namely,
\begin{align}
    \mbox{MMSE}(Q) = \frac{1}{|Q|}\sum_{(i,j)\in Q}{\left[\frac{\tilde{V}(i,j)}{j-i} - \frac{V(i,j)}{j-i}\right]^2},
    \nonumber\\
    % \label{eq:exp_metric}
    \mbox{MMAE}(Q) = \frac{1}{|Q|}\sum_{(i,j)\in Q}{\left|\frac{\tilde{V}(i,j)}{j-i} - \frac{V(i,j)}{j-i}\right|}.
    \nonumber
    % \label{eq:exp_metric}
\end{align}

\autoref{fig:query_more_vary_e_med} gives results on these metrics.  Let us first look at the first row, which gives MAE results.  The overall trend is similar to that of MSE (\autoref{fig:query_mse_vary_e_med}).  One notable difference is that the better hierarchical method does not give a better overall result (ToPS versus NM-E and $\hat{H}^c_{16}$ versus PAK).  This is because the better hierarchical method (especially the consistency step) is optimizing the squared error.  We note that Lee et al.~\cite{kdd:LeeWK15} proposed methods for optimizing absolute error ($L_1$ error), and they can be used in our setting if the target is to minimize absolute errors.

For the second and third row of \autoref{fig:query_more_vary_e_med}, which corresponds to MMAE and MMSE, respectively, we can see the overall trend and the relative performance of different methods are similar to the case of MAE and MSE.  The results are less stable though: this is because the range of the query here introduces another factor of randomness.  That is, due to the usage of the hierarchy, the larger the range, the smaller the error will be.  
For the MMAE metric, the results of our proposed \method can be as small as 1, while the existing work of PAK~\cite{ndss:perrier2019private} gives errors of $10^3$ to $10^4$, depending on the dataset.  Comparing with the flat method, where a Laplace noise on the order of $B/\epsilon$ (where $B$ is the upper bound of the data) or $\theta/\epsilon$ (if some technique of finding $\theta$ is applied) is added to each count, \method significantly improves over it.  Finally, the results of MMSE show similar trends as that of MSE.  These evaluation results further confirm the superiority of \method regardless of the evaluation metrics.

\end{document}